\documentclass[11pt]{article}
\usepackage{fullpage}
\usepackage{amsfonts}
\usepackage{amssymb}
\usepackage{amstext}
\usepackage{amsmath}
\usepackage{xspace}
\usepackage{theorem}
\usepackage{graphicx}
\usepackage{graphics}
\usepackage{colordvi}
\usepackage{color}
\usepackage{comment}
\usepackage[colorlinks=true,linkcolor=blue,citecolor=blue]{hyperref}
\usepackage{paralist}
\usepackage{enumitem}
\usepackage{bm}
\usepackage{mathpazo}
\usepackage{algpseudocode}
\usepackage{fancybox}

\newcommand{\CC}[0]{\operatorname{Constrained-Cut}}
\algnewcommand{\InlineComment}[1]{\hfill {\itshape // #1}}

\usepackage[compact]{titlesec}

\newenvironment{proofof}[1]{\noindent{\bf Proof of #1.}}
        {\hspace*{\fill}$\Box$\par\vspace{4mm}}

\titlespacing{\paragraph}{
	0pt}{
	0.1\baselineskip}{
	1em}

\newcommand{\defeq}{\stackrel{\mathrm{\scriptscriptstyle def}}{=}}

\newcommand{\DeclareProblem}[2]{\DeclareRobustCommand{#1}[0]{\textsf{#2}\xspace}}
\DeclareProblem{\SNDP}{rSNDP}

\newcommand{\abs}[1]{\lvert #1\rvert}
\newcommand{\paren}[1]{\left ( #1 \right ) }

\newcommand{\set}[1]{\left\{ #1 \right\}}

\newcommand{\parenbig}[1]{\bigl( #1 \bigr)}

\newcommand{\Hcal}[0]{\ensuremath{\mathcal H}\xspace} 
\newcommand{\Tcal}[0]{\ensuremath{\mathcal T}\xspace} 
\newcommand{\Gnet}[0]{\ensuremath{\Hcal'}\xspace}
\newcommand{\Dnet}[0]{\ensuremath{\Hcal}\xspace}

\DeclareMathOperator{\tw}{tw}
\DeclareMathOperator{\rootn}{root}

\newcommand{\be}{\begin{enumerate}}
\newcommand{\ee}{\end{enumerate}}
\newcommand{\bd}{\begin{description}}
\newcommand{\ed}{\end{description}}
\newcommand{\bi}{\begin{itemize}[noitemsep,topsep=0pt]}
\newcommand{\ei}{\end{itemize}}

\newtheorem{theorem}{Theorem}[section]
\newtheorem{lemma}[theorem]{Lemma}
\newtheorem{observation}[theorem]{Observation}

\newtheorem{corollary}[theorem]{Corollary}
\newtheorem{claim}[theorem]{Claim}
\newtheorem{proposition}[theorem]{Proposition}

\newtheorem{definition}[theorem]{Definition}

\newenvironment{proof}{\par \smallskip{\bf Proof:}}{\hfill\stopproof}
\def\stopproof{\square}
\def\square{\vbox{\hrule height.2pt\hbox{\vrule width.2pt height5pt \kern5pt
\vrule width.2pt} \hrule height.2pt}}

\newenvironment{prog}[1]{
\begin{minipage}{5.8 in}
{\sc\bf #1}
\begin{enumerate}}
{
\end{enumerate}
\end{minipage}
}

\renewcommand{\phi}{\varphi}
\newcommand{\eps}{\epsilon}

\newcommand{\poly}{\mathrm{poly}}

\newcommand{\Z}{\ensuremath{\mathbb Z}}

\newcommand{\F}{\ensuremath{\mathbb F}}

\newcommand{\cset}{{\mathcal C}}

\newcommand{\sets}{{\mathcal S}}
\newcommand{\pset}{{\mathcal P}}
\newcommand{\tset}{{\mathcal T}}
\newcommand{\tsethat}{{\widehat{\mathcal T}}}

\newcommand{\mincut}[0]{\operatorname{mincut}}
\newcommand{\Xset}{\ensuremath{X}\xspace} 
\newcommand{\card}[1]{\left| #1 \right|}
\newcommand{\nosolution}{\textsc{No Valid Solution}\xspace}

\newcommand{\chat}{\widehat{c}}
\newcommand{\Ehat}{\widehat{E}}
\newcommand{\Fhat}{\widehat{F}}
\newcommand{\Ghat}{\widehat{G}}

\newcommand{\Vhat}{\widehat{V}}
\newcommand{\Mincut}{\mincut}

\renewcommand{\O}{\widetilde{O}}
\newcommand{\bs}{\backslash}
\newcommand{\localcut}{\mathrm{LocalCut}}
\newcommand{\vol}{\mathit{vol}}

\newcommand{\cupdot}{\mathbin{\mathaccent\cdot\cup}}
\newcommand{\T}{\mathcal{T}}
\newcommand{\OfflineConnectivity}{\textsc{OfflineConnectivity}}

\newenvironment{fminipage}
{\begin{Sbox}\begin{minipage}}
		{\end{minipage}\end{Sbox}\fbox{\TheSbox}}

\newenvironment{algbox}[0]{\vskip 0.2in
	\noindent 
	\begin{fminipage}{6.3in}
	}{
	\end{fminipage}
	\vskip 0.2in
}

\newcommand{\optimal}{$O(kc^4)$\xspace}
\newcommand{\cwell}{connectivity-$c$ well-linked\xspace}
\newcommand{\sparsifier}[0]{connectivity-$c$ mimicking network\xspace}
\newcommand{\sparsifiers}[0]{\sparsifier{}s\xspace}
\newcommand{\sparsifiersTC}[0]{Connectivity-\texorpdfstring{$c$}{c} Mimicking Networks\xspace}

\newcommand{\Econtain}[0]{E^{\textit{contain}}}	
\newcommand{\Eintersect}[0]{E^{\textit{intersect}}}	

\title{Vertex Sparsification for Edge Connectivity}

	\author{Parinya Chalermsook\thanks{Aalto University, Finland \texttt{parinya.chalermsook@aalto.fi}} \and 
	Syamantak Das\thanks{Indraprastha Institute of Information Technology Delhi, India \texttt{syamantak@iiitd.ac.in}} \and 
	Bundit Laekhanukit\thanks{Shanghai University of Finance and Economics, China \texttt{bundit@sufe.edu.cn}} \and
	Yunbum Kook\thanks{KAIST, South Korea \texttt{yb.kook@kaist.ac.kr}. Part of this work was done while visiting Georgia Tech.}\and
	Yang P. Liu\thanks{Stanford University \texttt{yangpliu@stanford.edu}}\and
	Richard Peng\thanks{Georgia Tech \texttt{richard.peng@gmail.com}. Part of this work was done while visiting MSR Redmond.}\and
	Mark Sellke\thanks{Stanford University  \texttt{msellke@stanford.edu}} \and
	Daniel Vaz\thanks{Operations Research Group, Technische Universit\"{a}t M\"{u}nchen, Germany, \texttt{daniel.vaz@tum.de}. Part of this work was done while being at MPI Informatik and while visiting Aalto University.}}

\definecolor{OliveGreen}{cmyk}{0.64,0,0.95,0.40}

\begin{document}

\maketitle

\begin{abstract}

\emph{Graph compression} or \emph{sparsification} is a basic information-theoretic and computational question. A major open problem in this research area is whether $(1+\epsilon)$-approximate cut-preserving vertex sparsifiers with size close to the number of terminals exist. As a step towards this goal, we study a thresholded version of the problem: for a given parameter $c$, find a smaller graph, which we call \emph{connectivity-$c$ mimicking network}, which preserves connectivity among $k$ terminals exactly up to the value of $c$. We show that connectivity-$c$ mimicking networks with $O(kc^4)$ edges exist and can be found in time $m(c\log n)^{O(c)}$. We also give a separate algorithm that constructs such graphs with $k \cdot O(c)^{2c}$ edges in time $mc^{O(c)}\log^{O(1)}n$.

These results lead to the first data structures for answering fully dynamic offline $c$-edge-connectivity queries for $c \ge 4$ in polylogarithmic time per query, as well as more efficient algorithms for survivable network design on bounded treewidth graphs.
\end{abstract}

\pagenumbering{gobble}

\newpage

\pagenumbering{arabic}

\section{Introduction} 
\label{sec:Introduction}

\emph{Graph compression} or \emph{sparsification} is a basic information-theoretic and computational question of the following nature:
can we compute a ``compact'' representation of a graph, with fewer vertices or edges, that preserves important information?
Important examples include spanners, which preserve distances approximately up to a multiplicative factor, and cut and spectral sparsifiers \cite{BenczurK96,SpielmanT04}, which preserve cuts and the Laplacian spectrum up to an approximation factor of $(1+\eps)$.
Such edge sparsifiers allow us to reduce several algorithmic problems on dense graphs to those on sparse graphs, at the cost of a $(1+\eps)$ approximation factor. On the other hand, some computational tasks, such as routing or graph partitioning, require reducing the number of vertices (instead of edges), that is, \textit{vertex sparsification}. 

The notion of vertex sparsification we consider here is that of \emph{cut sparsification}, introduced by \cite{hagerup1998characterizing,moitra2009approximation,leighton2010extensions}. In this setting, we are given an edge-capacitated graph $G$ and a subset $\tset \subseteq V(G)$  of $k$ vertices called \emph{terminals}, and we want to construct a smaller graph $H$ that maintains all the minimum cuts between every pair of subsets of $\tset$ up to a multiplicative factor $q$, called the \textit{quality of the sparsifier}. 
More formally, we want to find a graph $H$ which contains $\tset$ as well as possibly additional vertices, such that for any $S \subseteq \tset$, the minimum cut between $S$ and $\tset \setminus S$ in $G$ and $H$ agree up to the multiplicative factor of $q$. 
Ideally, the size of the sparsifier (i.e., $|V(H)|$) should only depend on $|\tset|$ and not the size of $G$. 

There have been several results regarding tradeoffs between the quality $q$ and the size of cut sparsifiers. One line of work considers the case where the sparsifier $H$ has no additional vertices beyond the terminals. Here, an upper bound of $O(\log k /\log\log k)$ \cite{moitra2009approximation,leighton2010extensions,CharikarLLM10,EGKRTT10,MakarychevM10} and lower bound of $\Omega(\sqrt{\log k}/\log\log k)$ \cite{MakarychevM10} are known. In a different direction, quality $1$-sparsifiers (known as \emph{mimicking networks}) with $2^{2^k}$ vertices were shown to exist \cite{hagerup1998characterizing,khan2014mimicking}, and a lower bound of $2^{\Omega(k)}$ is also known \cite{krauthgamer2013mimicking}. Also, Chuzhoy \cite{Chuzhoy12} studied the problem of obtaining the best possible trade-offs between quality and size of sparsifiers, and showed that quality-$3$ sparsifiers with $O(Z^3)$ vertices exist, where $Z$ is the total capacity of all terminals. A major open problem is whether a quality-$(1+\eps)$ cut sparsifier of size $\O(k/\poly(\eps))$ exists; so far, this is only known for special graph classes, such as quasi-bipartite graphs~\cite{andoni2014towards,AbrahamDKKP16}.

The aim of this paper is to study a related graph sparsifier that is suitable for applications in designing fast algorithms for \textit{connectivity problems}. In particular, we consider the following problem: given an edge-capacitated graph $G$ with $k$ terminals $\tset$ and a constant $c$, construct a graph $H$ with $\tset \subseteq V(H)$ that maintains all minimum cuts up to size $c$ among terminals. Precisely, we want, for all subsets $S \subseteq \tset$, that $\min\left(c, \mincut_G(S,\tset\bs S)\right) = \min\left(c, \mincut_H(S,\tset\bs S)\right)$, where, for disjoint subsets $A,B \subseteq V(G)$, we define $\mincut_G(A,B)$ as the value of a minimum cut between $A$ and $B$ in graph $G$. In this case, we call $H$ a \emph{\sparsifier} of $G$ (See Definition~\ref{def:DefMain}).

Our main result (Theorem \ref{thm:Main}) shows that every graph $G$ with integer edge capacities admits a \sparsifier with \optimal edges (so \optimal vertices as well), and we show a near-linear time algorithm to compute it.
\begin{theorem}
	\label{thm:Main}
	Given any edge-capacitated graph $G$ with $n$ vertices, $m$ edges,
	along with a set $\tset$ of $k$ terminals and a value $c$,
	there are algorithms that construct a \sparsifier $H$ of $G$ with
	\begin{enumerate}
		\item \label{part:Main1} \optimal edges in time $O(m\cdot (c\log n)^{O(c)})$,
		\item \label{part:Main2} $k \cdot O(c)^{2c}$ edges
		in time $O(m \cdot c^{O(c)} \log^{O(1)}n)$.
	\end{enumerate}
\end{theorem}
In fact, the algorithm for Part~\ref{part:Main1} constructs the optimal {\em contraction-based mimicking network}, so any existential improvement to the size bound of such mimicking networks would immediately translate to an efficient algorithm.\footnote{Formally, if there exists a \sparsifier with $kf(c)$ edges that can be obtained by only contracting edges in $G$, for some function $f$, then our algorithm finds a mimicking network with at most $O(kf(c))$ edges.} Our second algorithm is more efficient, while blowing up the size of the mimicking network obtained. We believe that our dependence on $c$ is suboptimal -- we were only able to construct instances that require at least $2kc$ edges in the \sparsifier, and are inclined to believe that an upper bound of $O(kc)$ is likely.

Theorem \ref{thm:Main} has direct applications in fixed-parameter tractability and dynamic graph data structures (see Sections \ref{sec:ourresults} and \ref{sec:Applications}). In fact, our results and techniques have already been used to give a deterministic $n^{o(1)}$ update time fully dynamic algorithm for $c$-connectivity for all $c = o(\log n)$ \cite{JS20}. Additionally, our results are motivated in part by elimination-based graph algorithms \cite{KyngLPSS16,KyngS16}, which we discuss in Section \ref{sec:further}. In this way, we believe that achieving $(1+\eps)$-quality cut sparsifiers of size $\O(k/\poly(\eps))$, an analogue of approximate Schur complements for cuts, may have broad applications in graph algorithms and data structures.

\subsection{Our Results}
\label{sec:ourresults}
Our proof of existence of \sparsifiers with \optimal edges (and thus \optimal vertices as well) involves extending the recursive approach of \cite{Chuzhoy12} using \emph{well-linked sets} to a thresholded setting. The construction of \cite{Chuzhoy12} for cut sparsifiers maintains a partition of the vertices of the graph $G$. For each partition piece, the algorithm either finds a sparse cut to recurse on, or contracts the piece. \cite{Chuzhoy12} then bounds the deterioration in quality from these contractions. The main differences between our approach and \cite{Chuzhoy12} are:
\begin{itemize}
\item We introduce an extension of well-linkedness to a thresholded $c$-connectivity setting.
\item We do not run the recursion all the way down. Instead, we use a kernelization result on mimicking networks via gammoid representative sets \cite{Kratsch12} to bottom out the recursion.
\end{itemize} Additionally, in order to obtain near-linear running times for our constructions, we combine the {\em expander decomposition} technique \cite{SaranurakW19} with several other combinatorial results that allow us to build the desired \sparsifier.

We would like to note that the results of \cite{Kratsch12,fomin2016efficient} already give \sparsifiers of size $\poly(k,c).$ However, the dependence on $k$ is at least quadratic, and the algorithms for computing them run in at least quadratic time, as these results use linear algebra on matroids. Therefore, their results do not give more efficient algorithms for the applications of dynamic connectivity and subset $c$-EC below.

Theorem~\ref{thm:Main} has applications in data structures for dynamic edge connectivity. The problem of dynamic $c$-edge-connectivity is to design an algorithm which supports edge additions, deletions, and $c$-edge-connectivity queries between pairs of vertices as efficiently as possible, preferably in nearly constant $\O(1)$\footnote{Throughout, we use $\O(\cdot)$ to hide $\poly\log(n)$ factors. In particular, $\O(1) = \poly\log(n).$} amortized update time. For online fully dynamic algorithms, such results are only known for $c \le 3$ \cite{HenzingerK99,HolmLT01}.
Even in the simpler \emph{offline} model introduced by Eppstein~\cite{Eppstein94}, where the algorithm sees all queries at the beginning, the
only result for $c \ge 4$ is, to our knowledge, an unpublished offline fully
dynamic algorithm for $c = 4, 5$ by Molina and
Sandlund~\cite{MolinaS18:unpublished}, which requires about $\sqrt{n}$ time per
query.
The fact that even offline algorithms are not known for dynamic $c$-connectivity when $c > 5$ shows a serious gap in understanding of dynamic flow algorithms. We make significant progress towards shrinking this gap, and show in Section~\ref{sec:dynaconapplication} that, by combining Theorem~\ref{thm:Main} Part~\ref{part:Main2} with a divide and conquer algorithm for processing queries, we achieve nearly constant amortized time for offline fully dynamic $c$-edge-connectivity.
\begin{theorem}
	\label{thm:Dynacon}
	There is an offline algorithm that on an initially empty graph $G$ answers $q$ edge insertion, deletion, and $c$-connectivity queries between arbitrary pairs of vertices in amortized $\O(c^{O(c)})$ time per query.
\end{theorem}

Finally, \sparsifiers are perhaps the most natural object that can be used to ``pass along'' connectivity information between sub-problems in the dynamic programming framework. We illustrate this concept by presenting an additional application. 
The \textsf{Subset $c$-Edge-Connectivity} (or Subset $c$-EC) problem is the following:
given a graph $G=(V,E)$ with costs on edges, and a terminal set, find the cheapest subgraph $H$ in which every pair of terminals is $c$-connected. We show in Section~\ref{sec:sndpapplication} that Theorem~\ref{thm:Main} speeds up the running time for solving this problem in low treewidth graphs.
\begin{theorem}
	\label{thm:SNDP} 
	There is an algorithm that exactly solves {\sf Subset $c$-EC} on an input graph $G$ with $n$ vertices in time $n \exp\left(O(c^4\tw(G)\log(\tw(G)c)\right)$, where $\tw(G)$ denotes the treewidth of $G$.
\end{theorem}
This is an improvement over~\cite{ChalermsookDELV18} in which the running time was doubly-exponential in both $c$ and $\tw(G)$.
Furthermore, the existence of a conditional lower bound of $(3-\eps)^{\tw(G)}$ even when $c=1$, under the assumption of the strong exponential time hypothesis, implies that the dependence of our running time on $\tw(G)$ is almost optimal. Additionally, our dynamic programming based algorithm shows that any improvement to the edge bound \optimal in Theorem~\ref{thm:Main} gives an improvement for Theorem~\ref{thm:SNDP}.

\subsection{Related Work}
\label{sec:further}

We believe that our work has potential connections to dynamic data structures, elimination-based graph algorithms, and approximation algorithms and sparsification.
\begin{sloppypar}
\paragraph{Static and Dynamic $c$-Edge-Connectivity Algorithms.}
The study of efficient algorithms for computing graph
connectivity has a long history,
including the study of max-flow algorithms~\cite{GoldbergT14},
near-linear time algorithms for computing global min-cut~\cite{Karger00},
and most recently, progress in exact~\cite{Madry13,Madry16,CohenMSV17}
and approximate max-flow algorithms~\cite{Sherman17,Peng16,KelnerLOS14,Sherman13}.
The $c$-limited edge connectivity case can be solved in $O(mc)$ time
statically, and is also implied by $\eps$-approximate routines
by setting $\eps < 1/c$.
As a result, it is a natural starting point for developing routines
that can answer multiple flow queries on the same graph.
\end{sloppypar}

The question of computing max-flow between multiple pairs
of terminal vertices dates back to the Gomory-Hu tree~\cite{GomoryH61},
which gives a tree representation of all $s$-$t$ min-cuts.
However, such tree structures do not extend to arbitrary
subsets of vertices, and to date, have proven difficult
to maintain dynamically.
As a result, previous works on computing cuts between
a subset of vertices have gone through the use of tree-packing
based certificates.
These include results on computing the minimum cut separating terminals~\cite{ColeH03},
as well as the construction of $c$-limited Gomory-Hu trees~\cite{HariharanKP07,BhalgatHKP08}.

Such results are in turn used to compute max-flow between multiple
pairs of vertices~\cite{AbboudKT19:arxiv}.
These problems have received much attention in fine grained complexity,
since their directed versions are difficult~\cite{AbboudWY15,AbboudW14}, 
and it is not known whether computing $(1+\eps)$-approximate versions 
of these is possible.
From this perspective, the $c$-limited version is a natural starting point
towards understanding the difficulty of computing $(1+\eps)$-approximate
all-pairs max-flows in both static and dynamic graphs.

For the problem of finding a \sparsifier,
the construction time of these vertex sparsifiers is critical for their use in data structures~\cite{PengSS19}.
For a moderate number of terminals (e.g., $k = n^{0.1}$), nearly-linear time constructions of vertex sparsifiers with $\poly(k)$
vertices were previously known only when $c \leq 5$~\cite{PengSS19, MolinaS18:unpublished}. 
To our knowledge, the only results for maintaining exact $c$ connectivity for $c \geq 4$
are incremental algorithms \cite{DinitzV94,DinitzV95,DinitzW98,GoranciHT16}, in addition to the aforementioned fully dynamic algorithm
for $c = 4, 5$ by Molina and Sandlund \cite{MolinaS18:unpublished}, which took about $\sqrt{n}$ time per query.

Furthermore, since this work was originally released, the concept of \sparsifiers along with the techniques of this work has been used to design deterministic $n^{o(1)}$ time fully dynamic algorithms for exact $c$-connectivity for all $c = o(\log n)$ \cite{JS20}.

\paragraph{Elimination-based graph algorithms:} The study of \sparsifiers in this paper can also be viewed in the context of vertex
reduction / elimination based graph algorithms. Such algorithms are closely related to the widely used and highly practically
effective multigrid methods, which until very recently have been viewed
as heuristics with unproven bounds. Even in the static setting,
the only worst-case bounds for multi-grid and elimination based
algorithms have been in the setting of linear systems~\cite{KyngLPSS16,KyngS16},
by utilizing a combination of vertex and edge sparsifications.
Compared to the tree-like Laplacian solvers, sparse vertex elimination has a multitude of advantages:
they are readily parallelizable~\cite{KyngLPSS16},
and can be more easily adapted to data structures
that handle dynamic graphs~\cite{DurfeeKPRS17,DurfeeGGP19}.

Important properties of such routines is that the size of sparsifier is linear in the number of terminals, the construction can be computed in nearly linear time, and they are $(1+\eps)$-quality approximations. In particular, guaranteeing a $(1+\eps)$ approximation is essential as there are often multiples stages in elimination algorithms, so losing $\omega(1)$-quality at each stage would be detrimental. Our vertex-elimination routine combines all these properties and thus meets all criteria of previous elimination based routines \cite{KyngLPSS16,KyngS16}. On the other hand, most of the work on vertex sparsification to date has been on shortest path metrics~\cite{vdBrandS19:arxiv,Chechik18}, and/or utilizes algebraic techniques~\cite{vdBrandS19:arxiv,Kratsch12,FafianieHKQ16}. As a result, these routines, when interpreted as vertex elimination routines, either incur errors, or have size super-linear in the number of terminals.

\paragraph{Other notions of approximate sparsification.} 
Without using any additional vertices, the best known upper and lower bounds on the quality of vertex cut sparsifiers are $O(\log{k} / \log\log{k})$ \cite{CharikarLLM10,MakarychevM10} and $\Omega(\sqrt{\log k}/\log\log k)$ \cite{MakarychevM10} respectively. 
\cite{Chuzhoy12} presents a quality-$O(1)$, size-$O(C^3)$ sparsifier, computable in time $\poly(n) \cdot 2^C$, where $C$ denotes the total capacity of the edges incident on the terminals. It is open whether there are quality-$(1+\eps)$ and size $\poly(k/\eps)$ vertex sparsifiers for edge connectivity, and we see Theorem~\ref{thm:Main} as a first step towards achieving this goal.

Additionally, there has been significant work on vertex sparsification in approximation
algorithms~\cite{MakarychevM10,CharikarLLM10,EnglertGKRTT14,
	KrauthgamerR17,Kratsch12,AssadiKYV15,FafianieKQ16,FafianieHKQ16,
	GoranciR16,GoranciHP17b}. Recently, vertex sparsifiers were also shown to be
closely connected with dynamic graph
data structures~\cite{GoranciHP17a,PengSS19,GoranciHP18,DurfeeGGP19}.

There has also been work on mimicking networks on special graph classes. Krauthgamer and Rika~\cite{krauthgamer2013mimicking} presented a mimicking network of size $O(k^2 2^{2k})$ for planar graphs with $k$ terminals, nearly matching the lower bound~\cite{karpov2018exponential}. 
When all terminals lie on the same face, mimicking networks of size $O(k^2)$ are known~\cite{goranci2017improved,KrauthgamerR17}.
An upper bound of $O(k\cdot 2^{2^{\tw(G)}})$ is known for bounded-treewidth graphs \cite{chaudhuri2000computing}. 

\subsection{Structure of the Paper}
\label{subsec:Structure}

In Section~\ref{sec:Preliminaries}, we give preliminaries for our algorithms.
In Section~\ref{sec:overview}, we sketch our approach to the main results, the existence of a \sparsifier with \optimal edges and an algorithm to construct \sparsifiers of the optimal size, which we elaborate in detail in Section~\ref{sec:MainSection}.
In Section~\ref{sec:EfficientAlgorithms}, we take a different approach to make an algorithm more efficient.
We finalize our paper with detailed explanation on applications in Section~\ref{sec:Applications}.
\section{Preliminaries}
\label{sec:Preliminaries}

Our focus will be on cuts with at most $c$ edges. Our algorithms will involve contractions, which naturally lead to multigraphs. Therefore, we view capacitated graphs $(G,w)$ as multigraphs with $\min(w_e,c)$ copies of an edge $e$. Hence, we only deal with undirected, unweighted multigraphs.

Furthermore, we assume that each terminal vertex $t \in \tset$ has degree at most $c$ through the following operation: for $t \in \tset$ add a new vertex $t'$ and $c$ edges between $t$ and $t'$. As any cut separating $t$ and $t'$ has size at least $c$, this operation preserves all cuts of size at most $c$.

\subsection{Cuts, Minimum Cuts, and \texorpdfstring{$(\tset,c)$-equivalency}{(T,c)-equivalence}}

For a graph $G=(V,E)$ and disjoint subsets $A,B \subseteq V$, let $E_G(A,B)$ denote the edges with one endpoint in $A$ and the other in $B$. The set of cuts in $G$ consists of $E_G(X,V\bs X)$ for $X \subseteq V$. For a subset $X \subseteq V$ the \emph{boundary} of $X$, denoted by $\partial X$, is $E_G(X,V\bs X)$.
For subsets $A,B \subseteq V$, define $\mincut_G(A,B)$ to be the minimum cut separating $A,B$ in $G$. If $A \cap B \neq \emptyset$, then $\mincut_G(A,B) = \infty$. Formally, we have \[ \mincut_G(A,B) = \min_{\substack{S \subseteq V \\ A \subseteq S, B \subseteq V \bs S}} |E_G(S,V \bs S)|. \]
Furthermore, we let $\mincut(G,A,B)$ be the edges in a minimum cut between $A,B$ in $G$, so that $\mincut_G(A,B) = |\mincut(G,A,B)|$. 
If multiple minimum cuts exist, the choice is arbitrary, and does not affect our results.
For disjoint $A, B \subseteq V$, we sometimes write their disjoint union as $A \cupdot B \defeq A \cup B$ to emphasize that $A, B$ are disjoint.
We define the \emph{thresholded minimum cut} as $\mincut_G^c(A,B) \defeq \min(c, \mincut_G(A,B)).$
This definition then allows us to formally define $(\tset,c)$-equivalence.
\begin{definition}
	\label{def:DefMain}
	Let $G$ and $H$ be graphs both containing terminals $\tset$.
	We say that $G$ and $H$ are $(\tset, c)$-equivalent if
	for any subset $\tset_1 \subseteq \tset$ of the terminals we have that
	\[
		\mincut^c_H\left(\tset_1, \tset \bs \tset_1\right) =
		\mincut^c_G\left(\tset_1, \tset \bs \tset_1\right).
	\]
If $G$ and $H$ are $(\tset,c)$-equivalent, then we also say that $H$ is a \sparsifier for $G$.
\end{definition}
A \emph{terminal cut} is any cut that has at least one terminal from $\tset$ on both sides of the cut. The minimum terminal cut is the terminal cut
with the smallest number of edges. We denote by $(\tset, c)$-cuts the terminal cuts with at most $c$ edges.

We present several useful observations about the notion of $(\tset, c)$-equivalence.
\begin{lemma}
	\label{lem:Smaller}
	If $G$ and $H$ are $(\tset,c)$-equivalent, then
	for any subset of terminals $\tsethat \subseteq \tset$
	and any $\chat \leq c$, $G$ and $H$ are also
	$(\tsethat, \chat)$-equivalent.
\end{lemma}
\begin{lemma}
	\label{lem:AddEdges}
	If $G$ and $H$ are $(\tset, c)$-equivalent, then for
	any additional set of edges $\Ehat$ with endpoints in $\tset$,
	$G \cup \Ehat$ and $H \cup \Ehat$
	are also $(\tset, c)$-equivalent.
\end{lemma}
When used in the reverse direction, this lemma says that we can remove edges, as long as we include their endpoints as terminal vertices (Corollary~\ref{cor:RemoveEdges}).
We complement this partitioning process by showing that sparsifiers on disconnected graphs can be built separately (Lemma~\ref{lem:Independent}).
\begin{corollary}
  \label{cor:RemoveEdges}
	Let $\Ehat$ be a set of edges in $G$ with endpoints
	$V(\Ehat)$, and $\tset$ be terminals in $G$.
	If $H$ is $(\tset \cup V(\Ehat), c)$-equivalent
	to $G \bs \Ehat$, then $H \cup \Ehat$
	is $(\tset, c)$-equivalent to $G$.
\end{corollary}
\begin{lemma}
	\label{lem:Independent}
	If $G_1$ is $(\tset_1, c)$-equivalent to $H_1$,
	and $G_2$ is $(\tset_2, c)$-equivalent to $H_2$,
	then the vertex-disjoint union
	of $G_1$ and $G_2$,
	is $(\tset_1 \cup \tset_2, c)$-equivalent
	to the vertex-disjoint union of $H_1$ and $H_2$.
\end{lemma}
When considering \sparsifiers, we can restrict our
attention to sparse graphs~\cite{NagamochiI92}.
For completeness, we prove the following lemma in Appendix~\ref{proofs:EdgeReduction}.
\begin{lemma}
	\label{lem:EdgeReduction}
	Given any graph $G = (V, E)$ on $n$ vertices and any $c \geq 0$,
	we can find in $O(cm)$ time a graph $H$ on the same
	$n$ vertices, but with at most $c(n - 1)$ edges,
	such that $G$ and $H$ are $(V, c)$-equivalent.
\end{lemma}

\subsection{Contractions}
For a graph $G$ and an edge $e \in E(G)$, we let $G/e$ denote the graph 
obtained from $G$ by identifying the endpoints of $e$ as a single vertex; 
we say that we have \emph{contracted} the edge $e$.
The new vertex is marked as a terminal if at least
one of its endpoints was a terminal.
For a subset of edges $\Ehat \subseteq E$, we let $G/\Ehat$
denote the graph obtained from $G$ by contracting all edges in $\Ehat$.
For any vertex set $X \subseteq V$, we denote by $G/X$ the graph obtained from $G$ by contracting every edge in $G[X]$.

For multigraphs, minimum cuts are monotonically increasing under contractions.
\begin{lemma}
	\label{lem:CutMonotone}
	For any subset of vertices $V_1$ and $V_2$ in $V$,
	and any set of edges $\Ehat$, it holds that
	\[
	\mincut_G\left(V_1, V_2\right)
	\leq
	\mincut_{G/\Ehat}\left(V_1, V_2\right).
	\]
\end{lemma}
All our mimicking networks in Theorem \ref{thm:Main} are produced by contracting edges of $G$.
\section{Overview of our Approach}
\label{sec:overview}
In this section, we give an overview for our proof of Theorem~\ref{thm:Main} Part~\ref{part:Main1}. We first present our contraction-based approach to construct \sparsifiers with \optimal edges, and then show how to generically convert contraction-based approaches into efficient algorithms.

\paragraph{Existence of \sparsifiers with \optimal edges.}
Recall that, in our setup, we have a graph $G$ with a set of $k$ terminals $\tset \subseteq V$, and wish to construct a graph $H$ with \optimal edges which is $(\tset,c)$-equivalent to $G$.
Our algorithm constructs $H$ by contracting edges of $G$ whose contraction does not affect the terminal cuts of size at most $c$.
To find these \emph{non-essential edges}, we intuitively perform a recursive procedure to identify \emph{essential edges}, i.e., edges that are involved in terminal cuts of size at most $c$. After finding this set of essential edges in $G$, we contract all remaining edges.

At a high level, this recursive procedure finds a ``small cut" in $G$, marks these edges as essential, and recurses on both halves. We formalize this notion of small cut through the definition of well-linkedness, variations of which have seen use throughout flow approximation algorithms \cite{Chuzhoy12,RackeST14}. Here, we introduce a thresholded version of well-linkedness.\footnote{There are two notions of well-linkedness in the literature: edge linkedness and vertex linkedness. Here, our work focuses on edge linkedness. For discussions and definitions of vertex linkedness, we refer the readers to~\cite{reed1997tree}}
\begin{definition}
\label{def:welllinked}
For a graph $G$, we call a subset $X \subseteq V$ \emph{connectivity-$c$ well-linked}
if for every bipartition $(A, B)$ of $X$, we have 
$|E_G(A,B)| \geq \min ( |\partial A \cap \partial X|, |\partial B\cap \partial X|, c)$.
\end{definition}

If a bipartition $(A, B)$ of $X$ satisfies $|E_G(A,B)| < \min ( |\partial A \cap \partial X|, |\partial B\cap \partial X|, c)$, we say that $E_G(A,B)$ is a \emph{violating cut}, as it certifies that $X$ is not \cwell. In this way, a violating cut corresponds to the ``small cut" in $G$ whose edges we mark as essential. Conversely, we show in Lemma \ref{lem:ContractWellLinked} that \emph{all edges} inside a \cwell set are non-essential, i.e., may be freely contracted.

Our full recursive algorithm is as follows. We maintain a partition of $V \bs \tset = X_1 \cupdot X_2 \cupdot \dots \cupdot X_p$, where $p$ denotes the number of pieces, and track the potential function $\sum_{i=1}^p |\partial(X_i)|$. Initially, we let there be a single piece $X = V \bs \tset$, so that the potential value is $|\partial X| = |\partial(V \bs \tset)| \le kc$, by our assumption in Section \ref{sec:Preliminaries} that all terminals have degree at most $c$. We recursively refine the partition until each $X_i$ is either \cwell or $|\partial(X_i)| \le 2c-1$. 
More precisely, if $|\partial(X_i)| \ge 2c$ but $X_i$ is not \cwell, let $E_G(A, B)$ be a violating cut of $X_i$ for a bipartition $(A, B)$ of $X_i$; we then remove $X_i$ and add $A, B$ to our partition. After this partitioning process terminates, the well-linked pieces among $X_1,X_2,\cdots,X_p$ may be contracted as discussed. For the pieces with $|\partial(X_i)| \le 2c-1$ we make tricky manipulation on the boundary edges $\partial(X_i)$ as in Lemma~\ref{lem:boundaryAssumption} and then work on the line graph of $X_i$.
Applying a kernelization result (see Lemma~\ref{lem:kw12}), which develops from matroid theory (gammoid in particular) and the representative sets lemma (see Theorem~\ref{thm:fomin}), to the line graph gives rise to a fruitful result (see Lemma \ref{lem:c3}) which is more tailored to our edge-cut problem.
It allows us to contract those pieces down to $O(c^3)$ edges and maintain $(\tset,c)$-equivalence.

It suffices to argue that the number of pieces $p$ in the partition is at most $O(kc)$ at the end, so that our total edge count is $O(kc \cdot c^3) = O(kc^4).$ To show this, note that by Definition \ref{def:welllinked}, for a violating cut $E_G(A, B)$ of $X$, we have that $\max(|\partial A|, |\partial B|) \le |\partial X|-1$ and $|\partial A|+|\partial B| \le 2(c-1) + |\partial X|$. The former shows that our recursive procedure terminates, and combining the latter with our potential function bounds the number of pieces at the end. A more formal analysis is given in Section \ref{ssec:Existence}.

Note that the only non-constructive part of the above proof is the assumption that we can find a violating cut. 
However, we believe that an algorithm with very efficient running time is unlikely to exist, as this seems like a non-trivial instance of the non-uniform sparsest cut problem  (in Appendix~\ref{sec:CutFinding}, we present an algorithm with running time $2^{O(c^2)} k^2m$, which could be of independent interest).
Hence, we present another procedure that does not rely on computing a violating cut.

\begin{sloppypar}
\paragraph{Efficient algorithm for constructing contraction-based \sparsifiers.} 
Our above analysis, in fact, shows that all but \optimal edges of $G$ may be contracted while still giving a graph which is $(\tset,c)$-equivalent to $G$ (see Theorem \ref{thm:kc4}).
\end{sloppypar}

A natural high level approach for an algorithm would be to go through the edges $e$ of $G$ sequentially and check whether contracting $e$ results in a $(\tset,c)$-equivalent graph. If so, we contract $e$, and otherwise, we do not. Our analysis shows that at most \optimal edges will remain at the end, and in fact that proving a better existential bound improves the guarantees of such an algorithm.

Unfortunately, we do not know how to decide whether contracting an edge $e$ maintains all $(\tset,c)$-cuts even in polynomial time. To get around this, we enforce particular structure on our graph by performing an \emph{expander decomposition}. Expanders, defined formally in Definition \ref{def:expander}, are governed by their conductance $\phi$, and satisfy that any cut of size at most $c$ has at most $c\phi^{-1}$ vertices on the smaller side. For a fixed parameter $\phi$, \cite{SaranurakW19} have given an efficient algorithm to remove $O(m\phi\log^3 n)$ edges from $G$ such that all remaining components are expanders with conductance at least $\phi$ (see Lemma \ref{lem:ExpanderDecompose}). We now mark all the removed edges as essential, delete them, and mark their endpoints as additional terminals. Corollary \ref{cor:RemoveEdges} and Lemma \ref{lem:Independent} show that it suffices to work separately with each remaining component, which are guaranteed to be expanders with conductance at least $\phi$. Note that the total number of terminals is now $k + O(m\phi \log^3 n)$.

In order to check each edge $e$ and decide whether it can be safely contracted, we first enumerate all cuts in the graph with at most $c$ edges, of which there are at most $n(c\phi^{-1})^{2c}$, using the fact that the small side of any cut of size at most $c$ has at most $c\phi^{-1}$ vertices in a graph of conductance $\phi$ (see Lemma \ref{lem:EnumerateCuts}). For each cut, we find the induced terminal partition, and all involved edges. This allows us to find the minimum cut value for any terminal partition, as long as it is at most $c$.
Now, for an edge $e$ we check for all minimum cuts of size at most $c$ that it is involved in, whether there is another minimum cut separating terminals that does not involve $e$. If so, $e$ may be contracted, and otherwise, it cannot.
In case we contract $e$, we delete the minimum cuts containing $e$ in the enumeration.
Since cuts are monotone under contraction (see Lemma~\ref{lem:CutMonotone}) and a cut in $G/e$ is also a cut in $G$ (see Lemma~\ref{lem:ContractionCut}), the remaining minimum cuts in the enumeration correspond to the minimum cuts of $G/e$.
Hence, we do not have to rebuild the set of all small cuts during the algorithm.
As there are at most $n(c\phi^{-1})^{2c}$ total cuts of size at most $c$, this algorithm may be executed in time $\O(nc(c\phi^{-1})^{2c})$ using some standard data structures.

Finally, we discuss how to make our algorithm efficient, even though the total number of terminals increased to $k + O(m\phi \log^3 n)$ after the expander decomposition. We set $\phi^{-1} = O(c^4\log^3 n)$, and note that the number of edges in our \sparsifier for $G$ is $O(kc^4 + mc^4\phi \log^3 n) \le m/2$ as long as $m$ is a constant factor larger than $kc^4.$ Now we repeat this procedure until our \sparsifier has \optimal edges, which requires $O(\log m)$ iterations. Details are given in Section \ref{ssec:Algorithm}.

\section{Existence and Algorithm for Sparsifiers with \optimal edges}

\label{sec:MainSection}
We first show the existence of a \sparsifier with \optimal edges in Section~\ref{ssec:Existence}, based on contractions of \cwell sets and replacement of sets with sparse boundary.
Then, in Section~\ref{ssec:Algorithm}, we design a $O(m(c\log n)^{O(c)})$ time algorithm to find a \sparsifier whose size matches the best guarantee achievable via contractions.

\subsection{Existence of Sparsifiers with \optimal edges via Contractions}
\label{ssec:Existence}

Given a graph $G$ and $k$ terminals $\tset$, our construction of a \sparsifier with \optimal edges leverages a recursion scheme, where we maintain a partition of the vertices $X = V\bs\tset$, and track the total number of boundary edges of the partition as a potential function.
This approach naturally introduces the notion of well-linkedness, a standard tool for studying flows and cuts, in order to refine the partition.
Additionally, we must stop recursion at sets with sufficiently sparse boundary to guarantee that the recursion terminates without branching exponential times.
The recursion results in a decomposition of $V\setminus \tset$ into at most $kc$ clusters, each of which is either a \cwell set or a set with sparse boundary.
Then we contract the \cwell sets and change the sets with sparse boundary into equivalent \sparsifiers with $O(c^3)$ edges.
This procedure results in the following theorem.
\begin{theorem}
	\label{thm:kc4}
	For a graph $G$ with $k$ terminals, there is a subset $E'$ of $E(G)$ such that the size of $E'$ is $O(kc^4)$ and the graph with all edges except $E'$ contracted, $G / (E\backslash E')$, is $(\tset, c)$-equivalent to $G$.
\end{theorem}

To this end, we elaborate the procedure with details in Section~\ref{sssec:description}. To handle sets with sparse boundary, we use a known kernelization result based on matroid theory and the representative sets lemma.
Unfortunately, these results mostly discuss vertex cuts, so in Section~\ref{sssec:reduction} we build a gadget to transform a given graph, from which we wish to obtain a \sparsifier, into a new graph whose minimum \textit{vertex} cut of any partition of terminals corresponds to a minimum \textit{edge} cut of the corresponding partition of terminals in the original graph.

\subsubsection{Existence Proof: \textsc{PolySizedcNetwork}}
\label{sssec:description}

As discussed in Section~\ref{sec:overview}, it is desirable to find \cwell sets, because they can be contracted without changing connectivity.
Its proof is deferred to Appendix \ref{proofs:ContractWellLinked}.
\begin{lemma}
	\label{lem:ContractWellLinked}
	Let $X$ be a \cwell set in $G$, and $\tset$
	be terminals disjoint with $X$ (i.e.\ $X \cap \tset = \emptyset$).
	Then $G/X$ is $(\tset, c)$-equivalent to $G$. 
\end{lemma}

The recursive procedure \textsc{PolySizedcNetwork} takes a subset $X$ of $V\setminus \tset$ and bisects it if there exists a violating cut.
Since finding a violating cut $E_G(A, B)$ of $X$ guarantees that $|\partial A|, |\partial B| < |\partial X|$, the recursion ends up reaching the base case in which the number of boundary edges is at most $2c-1$.
If there are no violating cuts, contracting $X$ also halts the recursion.

\begin{figure}
	\begin{algbox}
		$H = \textsc{PolySizedcNetwork}(G)$
		
		\underline{Input}: undirected unweighted multi-graph $G$.\\
		\underline{Output:} a \sparsifier $H$.
		\bi
			\item[ ] If $|\partial G| \leq 2c - 1$:
			\bi
				\item[-] Return, based on Lemma~\ref{lem:c3}, a $(\tset', c)$-equivalent \sparsifier with terminals $\tset'$, where $\tset'$ is the set of vertices incident to boundary edges $\partial G$.
			\ei
			\item[ ] Else:
			\bi
				\item[-] Find a violating cut $(V_1, V_2)$ if it exists and then, return \textsc{PolySizedcNetwork}($G[V_1]$) and \textsc{PolySizedcNetwork}($G[V_2]$).
				\item[-] If no violating cut exists, contract $G$ to a single vertex.
			\ei
			\item[ ] Return a \sparsifier $H$.
		\ei
	\end{algbox}
	\vspace{-6mm}
	\caption{Pseudocode for \textsc{PolySizedcNetwork}}
	\label{fig:polysizednetwork}
\end{figure}

Hence, $\textsc{PolySizedcNetwork}(V\setminus \tset)$ just partitions $V\setminus \tset$ into pieces being either \cwell sets or sets whose number of boundary edges is at most $2c-1$.
\sloppy
The contraction of a \cwell set for $(\tset, c)$-equivalency is justified by Lemma~\ref{lem:ContractWellLinked}.
Also, Corollary~\ref{cor:RemoveEdges} and Lemma~\ref{lem:Independent} justify the replacement of a set $X$ with no terminals by a $(\tset', c)$-equivalent graph, where we introduce boundary vertices in $X$ as the tentative terminals $\tset'$.

The number of edges in a \sparsifier returned by the procedure depends on 
\begin{inparaenum}[(i)]
\item the number of pieces, and \label{sec41expl:1}
\item how small equivalent sparsifiers a subgraph with $O(c)$ boundary edges has. \label{sec41expl:2}
\end{inparaenum}

For \eqref{sec41expl:1}, the number of smaller pieces being either \cwell or $|\partial X|\leq 2c -1$ simply matches with the number of branching during the recursion induced by the existence of a violating cut.
It is bounded by the following decreasing invariant, which decreases by at least $1$ for each branching.

\begin{lemma}
	\label{lem:DecreasingInvariant}
	When \textsc{PolySizedcNetwork} splits a given $X$ into $\{X_i\}_{i=1}^l$, $\sum_{i\leq l} \max( |\partial (X_i)| - 2c + 1, 0)$ is a decreasing invariant with respect to separation induced by a violating cut for some $X_i$.	
\end{lemma}
\begin{proof}
	The number of pieces only increases when a piece $X$ is divided into $A$ and $B$ by finding a violating cut.
	Denote $k = |\partial X|, k_1 = |\partial A \cap \partial X|, k_2 = |\partial B\cap \partial X|$,  and $l = |E_G(A, B)|$.
	Note that $|\partial A| = k_1 + l$ and $|\partial B| = k_2 + l$.
	From the definition of a violating cut, we have $l < \min (k_1, k_2, c)$ and clearly, $|\partial A|, |\partial B| < |\partial X|$.
	When $X$ is divided, the term $|\partial X| - 2c + 1$ is replaced by $\max( |\partial A| - 2c + 1, 0) + \max( |\partial B| - 2c + 1, 0)$.
	If either $A$ or $B$ happens to have less than $2c - 1$ boundary edges, the latter term is strictly smaller.
	Now, let $|\partial A|, |\partial B| \geq 2c -1$.
	Then,
	\begin{align*}
	|\partial A| - 2c + 1 + |\partial B| - 2c + 1 &= k_1 + l - 2c + 1 + k_2 + l - 2c + 1\\
	&= k - 2c + 1 + 2(l - c) + 1 < k - 2c + 1.
	\end{align*}
\end{proof}

The above lemma says that the number of branching is bounded by $|\partial X|$ since the decreasing invariant begins with $|\partial X|$.
By the assumption that terminals have degree at most $c$, we have $|\partial(V\setminus \tset)|$ pieces, which is upper bounded by $kc$.

For (\ref{sec41expl:2}), due to a tricky preprocessing on boundary edges as in Lemma~\ref{lem:boundaryAssumption}, we may assume that a set with $O(c)$ boundary edges can be viewed as a set with $O(c)$ tentative terminals, each of which has degree $1$.
This preprocessing gives rise to the following lemma with its proof presented in Section~\ref{sssec:reduction}.

\begin{lemma}
	\label{lem:c3}
	Let $G = (V, E)$ be a graph with a set $\tset$ of $O(c)$ terminals and each terminal have degree $1$.
	There is a subset $E'$ of $E$ with $|E'| = O(c^3)$ and $G/(E\bs E')$ is a \sparsifier for $G$.
\end{lemma}

We can combine series of lemmas to show Theorem~\ref{thm:kc4}.

\begin{proofof}{Theorem \ref{thm:kc4}}
We show that \textsc{PolySizedcNetwork} returns a \sparsifier with \optimal edges for a graph $G$. First, applying Lemma \ref{lem:DecreasingInvariant} shows that the total number $p$ of pieces in the partition $V\setminus \tset = X_1 \cup X_2 \cup \cdots \cup X_p$ is at most $kc$, as $|\partial (V\setminus \tset)| \le kc$ by our assumption that all terminals have degree at most $c$.

To bound the total number of edges in the final sparsifier, we must analyze two contributions. First, the total number of boundary edges over all partition pieces is at most $\left|\bigcup_{i=1}^p (\partial X_i) \right| \le O(kc^2)$, as the total number of boundary edges may increase by $c$ each time we split a partition piece into two pieces, and there are at most $kc$ pieces. The other contribution is from the partition pieces $X_i$ with $|\partial (X_i)| \le 2c-1$. The total number of edges from this is at most $kc \cdot O(c^3) = O(kc^4)$ by applying Lemma \ref{lem:c3}.

To verify that the returned graph is indeed a \sparsifier, it suffices to apply Lemma \ref{lem:ContractWellLinked} to argue that we can contract well-linked pieces. Then we can use Corollary \ref{cor:RemoveEdges} to delete all boundary edges in $\bigcup_{i=1}^p \partial X_i$, and then use Lemma \ref{lem:Independent} and build a \sparsifier on each $X_i$ separately.
\end{proofof}

\subsection{Algorithm for the Optimal Sparsifiers: Contracting Non-Essential Edges}
\label{ssec:Algorithm}

We use many forms of graph partitioning and operations, such as adding and deleting edges among terminals (Lemma~\ref{lem:Smaller},~\ref{lem:AddEdges}, and Corollary~\ref{cor:RemoveEdges}), and connected components may be handled separately (Lemma \ref{lem:Independent}).
These observations form the basis of our divide-and-conquer scheme, which repeatedly deletes edges, adds terminals, and works on connected components of disconnected graphs.
Our approach in fact utilizes \textit{expander decomposition} elaborated in Section~\ref{sssec:EnumerateSmallCuts} to split a graph into several expanders.
Removing the inter-cluster edges, it sparsifies each expander by contracting non-essential edges in the expander, the contraction of each still preserves the value of a minimum cut up to $c$ between any partition of terminals.
Then it glues together all the sparsified expanders via the inter-cluster edges to obtain a \sparsifier.
This one pass reduces the number of edges by half.
Repeating several passes leaves \textit{essential} edges in the end, leading to the \sparsifier of the optimal size, which is currently \optimal.
\begin{theorem}
	\label{thm:sparsifyGraph}
	For a graph $G$ with $n$ vertices, $m$ edges, and $k$ terminals, there exists an algorithm which successfully finds a \sparsifier with \optimal edges in $O(m(c\log n)^{O(c)})$.
\end{theorem}

In fact, our algorithm guarantees a stronger property: If there exists a \sparsifier with $kf(c)$ edges that can be obtained by only contracting edges in $G$, for some function $f$, then our algorithm finds a \sparsifier with $O(kf(c))$ edges.
We present the proof in three parts. In Section~\ref{sssec:EnumerateSmallCuts} and \ref{sec: phi sparse}, we explain two sub-routines that are used in our algorithm. 
The description of the algorithm is in Section~\ref{sec:FirstAlgo:EssentialEdges}.

\subsubsection{Enumeration of Small Cuts via Expander Decomposition}
\label{sssec:EnumerateSmallCuts}

To achieve Theorem~\ref{thm:sparsifyGraph}, 
we utilize insights from recent results on finding $c$-vertex
cuts~\cite{NanongkaiSY19a,NanongkaiSY19b:arxiv,ForsterY19:arxiv},
namely that in a well connected graph, any cut of size at most $c$
must have a very small side.
This notion of connectivity is
formalized through the notion of graph conductance.
\begin{definition}
\label{def:expander}
	In an undirected unweighted graph $G = (V, E)$, denote
	the volume of a subset of vertices, $\vol(S)$, as the
	total degrees of its vertices.
	The conductance of a cut $S$ is then
	\[
	\Phi_G\left( S \right)
	=
	\frac{\abs{\partial\left( S \right)}}
	{\min\left\{ \vol\left( S \right), \vol\left(V \bs S \right) \right\}},
	\]
	and the conductance of a graph $G = (V, E)$ is the minimum
	conductance of a subset of vertices:
	\[
	\Phi\left( G \right)
	= \min_{S \subseteq V} \Phi_{G}\left( S \right).
	\]
\end{definition}
We use expander decomposition to reduce to the case where the graph has high conductance.
\begin{lemma}
	\label{lem:ExpanderDecompose}
	(Theorem 1.2 of~\cite{SaranurakW19})
	There exists an algorithm \textsc{ExpanderDecompose}
	that for any undirected unweighted graph $G$ and
	any parameter $\phi$, decomposes in $O(m\phi^{-1}\log^{4}{n})$
	time $G$ into pieces $\{G_i\}$ of conductance at least $\phi$
	so that at most $O(m \phi \log^{3}n)$ edges are between
	the pieces.
\end{lemma}
Note that if a graph has conductance $\phi$,
any cut $(S,V\bs S)$ of size at most $c$ must have
\begin{equation}
\min\left\{ \vol\left( S \right), \vol\left(V \bs S \right) \right\}
\leq
c\phi^{-1}.
\label{eq:ExpanderImba}
\end{equation}
In a graph with expansion $\phi$, we can enumerate all cuts of size at most $c$ in time exponential in $c$ and $\phi$.
As a side note, the time complexity of both the results in Theorem~\ref{thm:Main} are dominated by 
the $c^{O(c)}$ term, essentially coming from this enumeration.
Hence, a more efficient algorithm on enumeration may open up the possibility toward a faster algorithm for finding a \sparsifier.
\begin{lemma}
	\label{lem:EnumerateCuts}
	In a graph $G$ with $n$ vertices and conductance $\phi$, there exists an algorithm that enumerates all cuts of size at most $c$
	with connected smaller side in time $O(n(c\phi^{-1})^{2c})$.
\end{lemma}
\begin{proof}
	We first enumerate over all starting vertices.
	For a starting vertex $u$,
	we repeatedly perform the following process.
	\bi
	\item[1.] Perform a DFS from $u$ until it reaches more than
	$c \phi^{-1}$ vertices.
	\item[2.] Pick one of the edges among the reached vertices as
	a cut edge.
	\item[3.] Remove that edge, and recursively start another DFS
	starting at $u$.
	\ei
	After we have done this process at most $c$ times, we check whether
	the edges form a valid cut, and store it if so.
	
	By Equation~\ref{eq:ExpanderImba}, the smaller side of the
	cut involves at most $c \phi^{-1}$ vertices.
	Consider such a cut with $S$ as the smaller side,
	$F = E(S, V \setminus S)$, and $|S| \leq c \phi^{-1}$.
	Then if we picked some vertex $u \in S$ as the starting point,
	the DFS tree rooted at $u$ must contain some edge in $F$
	at some point.
	Performing an induction with this edge removed then gives
	that the DFS starting from $u$ will find this cut.
	Because there can be at most $O((c \phi^{-1})^2)$ different
	edges picked among the vertices reached, the total work performed
	in the $c$ layers of recursion is $O((c \phi^{-1})^{2c})$.
\end{proof} 

It suffices to enumerate all such cuts once at the start,
and reuse them as we perform contractions.
\begin{lemma}
\label{lem:ContractionCut}
	If $F$ is a cut in $G / \Ehat$, then $F$ is also a cut in $G$.
\end{lemma}
Note that this lemma also implies that an expander stays
so under contractions.
So, we do not even need to re-partition the graph as we recurse.

\subsubsection{Sparsifying Procedure (\texorpdfstring{$\varphi$-}{phi-}\textsc{Sparsify} )} 
\label{sec: phi sparse} 

We need a subroutine used to sparsify a graph with conductance $\varphi$ and terminals $\tset$.
This subroutine named as $\varphi$-\textsc{Sparsify} takes such a graph and enumerates all cuts of size at most $c$ by a smaller side through Lemma~\ref{lem:EnumerateCuts}.
Then it sparsifies the expander by checking if the contraction of each edge still preserves $(\tset, c)$-equivalency.
Formally, we can contract an edge $e$ while preserving $(\tset, c)$-equivalency if and only if for any partition $(\tset_1, \tset_2)$ of terminals $\tset$, there exists a $(\tset_1, \tset_2)$-mincut not containing the edge $e$.
For convenience, we call such an edge $e\in E(G)$ as \textit{contractible} in $G$.
Sequentially checking all edges in $G$ and contracting some if possible, we show that $\varphi$-\textsc{Sparsify} only leaves at most $O(|\tset|c^4)$ ``essential edges'' which appear in a minimum cut of any partition of the terminals.

Through the enumeration of all cuts of size at most $c$ by a smaller side of the cut (Lemma~\ref{lem:EnumerateCuts}), $\varphi$-\textsc{Sparsify} forms an auxiliary graph $H$ for efficient tracking of minimum cuts of partitions of terminals as follows:
$V(H)$ is the disjoint union of $P$, $C$, and $E_0$, where
\bi
\item[1.] $P$ is the set of an induced partition of terminals $\tset$ during the enumeration,
\item[2.] $C$ is the set of a minimum cut separating a partition of terminals in $P$,
\item[3.] $E_0$ is the set of edges included in a minimum cut in $C$,
\ei
and for $p \in P, c \in C$, and $e \in E_0$, add an edge $pc$ to $E(H)$ if $c$ is a minimum cut of $p$, and an edge $ce$ to $E(H)$ if $e \in c$.

For a given query edge $e \in E(G)$, the algorithm deletes all nodes (minimum cuts) $N(e) \subseteq C$, also removing the incident edges to $N(e)$.
Then it checks if there is a node (partition $p$) in $P$ whose degree becomes $0$ after the deletion.
If so, it means that the edge $e$ appears in all minimum cuts of the partition $p$, leading the algorithm to undo the deletion.
Otherwise, it means that the algorithm may contract $e$ and actually obtains a $(\tset, c)$-equivalent graph $G/e$.

In the case that it contracts a contractible edge $e$, we should make sure that the auxiliary graph from $G/e$ is equal to $H$ with $N(e)$ deleted.
First of all, a minimum cut $F$ of size at most $c$ inducing a partition of terminals in $G/e$ is also a cut of the partition in $G$ (see Lemma~\ref{lem:ContractionCut}).
As $G/e$ and $G$ are $(\tset, c)$-equivalent, $F$ must be a minimum cut of the partition in $G$ as well.
For the opposite direction, a minimum cut of a partition of terminals in $G$, which does not contain $e$, is also a minimum cut of the partition in $G/e$, since the value of minimum cuts non-decreases under contraction (see Lemma~\ref{lem:CutMonotone}).
Therefore, we only need to enumerate all cuts of size at most $c$ $O(1)$ times and to create an auxiliary graph at the very beginning of $\varphi$-\textsc{Sparsify}, and simply update the auxiliary graph in response to contraction of edges without re-enumerating all cuts of size at most $c$ in contracted graphs.

In this way, scanning through each edge in sequence, $\varphi$-\textsc{Sparsify} checks if each edge is contractible in $G$ with contractible edges (in their turns) already contracted.
In the end, it returns a $(\tset, c)$-equivalent graph $G/X$, where $X$ is the set of contractible edges in each turn.

\begin{lemma}
	\label{lem:sparsifyExpander}
	For a graph $G$ with conductance $\varphi$, $n$ vertices, $m$ edges, and $k$ terminals,
	the algorithm $\varphi$-\textsc{Sparsify} returns a \sparsifier with \optimal edges in $O(m + nc(c\varphi^{-1})^{2c})$ time. 
\end{lemma}
\vspace{-3mm}
\begin{proof} Observe that once an edge $e$ is not contractible in its turn, $e$ never becomes contractible when checking other subsequent edges.
In other words, for $i < j$, if the algorithm checks the edge $e$ in $i^{th}$ iteration and marks it as non-contractible, then $e$ is still not contractible when checking another edge in $j^{th}$ iteration.
Hence, after checking all edges, the remaining edges $E^*$are not contractible in $G/(E-E^*)$.
As Theorem~\ref{thm:kc4} guarantees the existence of contractible edges as long as the number of remaining edges are larger than \optimal, so the graph with the contractible edges contracted, $G/(E-E^*)$, returned by $\varphi$-\textsc{Sparsify} has at most \optimal edges as well.
	
For the running time, the algorithm enumerates all minimum cuts of size at most $c$ in $O(n\cdot (c\varphi^{-1})^{2c})$ and updating an auxiliary graph requires as many references as the number of edges in the auxiliary graph.
As each minimum cut has size at most $c$, double counting on the number of edges results in the running time $O(m + nc \cdot (c \varphi^{-1})^{2c})$.
\end{proof}

\subsubsection{Putting things together}
\label{sec:FirstAlgo:EssentialEdges}
Now we join all the sparsified graphs via the removed inter-cluster edges and reduce the total number of edges by half.
Repeating this procedure until no more edges are contractible, we can build a $(\tset, c)$-equivalent graph with at most \optimal essential edges.
We present the algorithm \textsc{EfficientPolySized} with details in Figure~\ref{fig:efficientpolysized} and with analysis as follows,
where $C'$ is a constant such that in Lemma~\ref{lem:ExpanderDecompose} the number of edges between the pieces are at most $C' m\varphi \log^3 n$.
\begin{figure}
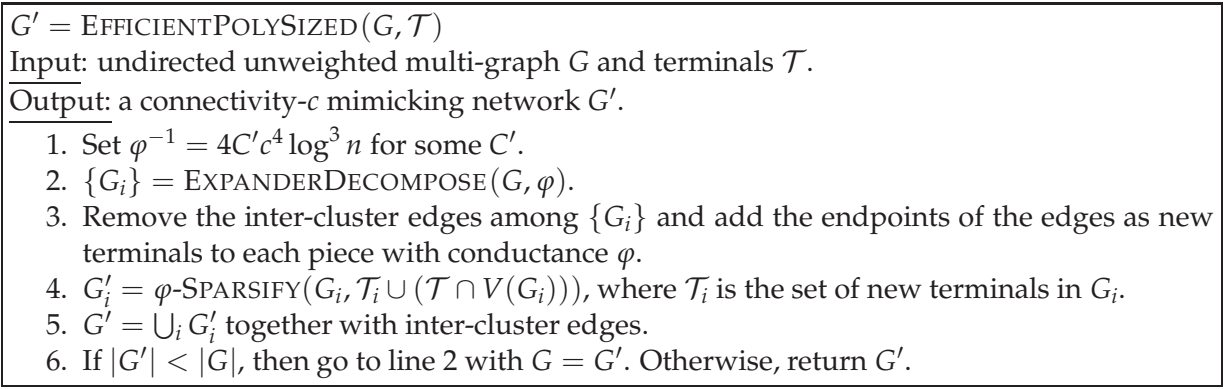

	\begin{algbox}
		$G' = \textsc{EfficientPolySized}(G, \tset)$\\
		\underline{Input}: undirected unweighted multi-graph $G$ and terminals $\tset$.\\
		\underline{Output:} a \sparsifier $G'$.		
		\bi
			\item[1.] Set $\varphi^{-1} = 4C'c^4\log^3 n$ for some $C'$.
			\item[2.] $\{G_i\} = \textsc{ExpanderDecompose}(G, \varphi)$.
			\item[3.] Remove the inter-cluster edges among $\{G_i\}$ and add the endpoints of the edges as new terminals to each piece with conductance $\varphi$.
			\item[4.] $G_i' = \varphi$-\textsc{Sparsify}$(G_i, \tset_i \cup (\tset \cap V(G_i)))$, where $\tset_i$ is the set of new terminals in $G_i$.
			\item[5.] $G' = \bigcup_i G_i'$ together with inter-cluster edges.
			\item[6.] If $|G'| < |G|$, then go to line $2$ with $G = G'$.
			Otherwise, return $G'$.
		\ei
	\end{algbox}
	\vspace{-8mm}
	\caption{Pseudocode for \textsc{EfficientPolySized}}
	\label{fig:efficientpolysized}
\end{figure}
Theorem \ref{thm:Main} part \ref{part:Main1} follows from the analysis of \textsc{EfficientPolySized}.
{
	\renewcommand{\thetheorem}{\ref{thm:sparsifyGraph}}
\begin{theorem}
	For a graph $G$ with $n$ vertices, $m$ edges, and $k$ terminals, the algorithm \textsc{EfficientPolySized} successfully finds a \sparsifier with \optimal edges in $O(m\cdot (c \log n)^{O(c)})$.
\end{theorem}
	\addtocounter{theorem}{-1}
}
\begin{proof}
	By Corollary~\ref{cor:RemoveEdges} and Lemma~\ref{lem:sparsifyExpander}, \textsc{EfficientPolySized} successfully finds a \sparsifier of $G$.
	For the size, we prove a more general statement that if the optimal number of edges in a \sparsifier of a graph with $k$ terminals is $k\cdot p(c)$ for a polynomial $p$, then \textsc{EfficientPolySized} returns a sparsifier with $O(kp(c))$ edges.
	
	We show by induction that after $i^{th}$ iteration, the number of remaining edges is at most $kp(c)\sum_{r=0}^{i-1}\frac{1}{2^r} + \frac{m}{2^i}$, which is bounded by $2kp(c) + \frac{m}{2^i}$.
	Hence, after $O(\log m)$ iterations, the algorithms yields a \sparsifier with $O(kp(c))$ edges.
	
	Observe that $\varphi = 1/(4C'p(c)\log^3 n)$ satisfies $\varphi \cdot (C' p(c) \log^3 n + C' \log^3 n) \leq \frac{1}2$.
	In the first iteration, the total number of terminals is bounded by $k + mC'\varphi \log^3 n$.
	Hence, in line 5, the total number of edges in $G'$ is bounded by 
	\[
		(k + mC'\varphi \log^3 n)p(c) + mC'\varphi \log^3 n
	\leq kp(c) + m\varphi \cdot (C'p(c)\log^3 n + C'\log^3 n) \leq kp(c) + \frac{m}2
	\]
	Using the similar argument for $i^{th}$ iteration and induction hypothesis, we have
	\begin{align*}
	(k + (kp(c)\sum_{r=0}^{i-2}\frac{1}{2^r} &+ \frac{m}{2^{i-1}})C'\varphi \log^3 n)p(c) + (kp(c)\sum_{r=0}^{i-2}\frac{1}{2^r} + \frac{m}{2^{i-1}})C'\varphi \log^3 n\\
	&\leq kp(c) + (kp(c)\sum_{r=0}^{i-2}\frac{1}{2^r} + \frac{m}{2^{i-1}})/2
	\leq kp(c)\sum_{r=0}^{i-1}\frac{1}{2^r} + \frac{m}{2^i}.
	\end{align*}
	
	For time complexity part, it is dominated by $\varphi$-\textsc{Sparsify} which takes time $O(m\cdot c^{O(c)} \log^{6c}n \cdot \log m) = O(m(c \log n)^{O(c)})$ as desired.
\end{proof}

\subsection{Proof of Lemma \ref{lem:c3}: Transforming Edge Cuts to Vertex Cuts}
\label{sssec:reduction}

As seen above, \textsc{PolySizedcNetwork} replaces a set with sparse boundary by a \sparsifier.
Here we present a key lemma used for this subroutine,
which reduces our problem to the problem of identifying essential vertices in preserving the value of minimum \textit{vertex} cuts.
The notion of vertex cuts is closely related with the notion of vertex-disjoint paths, which takes advantages of a well-developed theory from gammoid and representative sets.
The following result in \cite{Kratsch12} is what we will make use of in essence.

\begin{lemma}[\cite{Kratsch12}]
	\label{lem:kw12}
	Let $G = (V, E)$ be a directed graph, and $X \subseteq V$ a set of terminals.
	We can identify a set $Z$ of $O(|X|^3)$ vertices such that for any $A, B\subseteq X$, a minimum $(A, B)$-vertex cut in $G$ is contained in $Z$.
\end{lemma}

Note that the above lemma addresses a minimum vertex cut, not an edge cut and it holds under digraphs setting.
However, we can still replace digraphs with undirected graphs;
for given an undirected graph $G = (V, E)$, simply orient each edge in the both directions to obtain a directed graph $\widehat{G}$ and then apply the above result to $\widehat{G}$.

This simple but amazing result develops in the context of proving the usefulness of \emph{matroid theory} to \emph{kernelization}.
As a set of vectors in vector fields enjoys removal of redundant vectors through the notion of linear independence, a general notion of independence potentially leads to identification of \textit{important} objects in a problem.
A matroid is one which abstracts and generalizes the notion, while especially developed in abstracting central notions in graph theory.
Formally, for a given set $E$, a matroid $M = (E, \mathcal{I})$ means a collection $\mathcal{I}$ of all the `independent' subsets of $E$.
As motivated from linear independence in vector fields, the most basic matroid is derived from a matrix $A$ over a field $\F$ and the sets of all sets of columns linearly independent over $\F$.
Some matroids $M$ happen to have an intrinsic matrix $A_M$ over $\F$ which represents the matroid in the way described above.
We call such matroids as \textit{linear} matroids over $\F$ and its matrix $A_M$ as a representation matrix.

When a directed graph $G = (V, E)$ with source vertices $S\subseteq V$ is given, a gammoid collects all the sets $T \subseteq V$ to which there exists $|T|$ vertex-disjoint paths from a subset of $S$.
The size of the largest such $T$ is called the rank of the gammoid and naturally corresponds to the rank of matrices over $\F$.
Among various matroids originated from graph theory, a \textit{gammoid} naturally bridges the gap between vertex cuts and matroid theory since the size of a minimum vertex cut between two sets has to do with the number of vertex-disjoint paths between the two sets.

The last key notion is a $q$-representative set.
For a given matroid $(E, \mathcal{I})$, a family $\sets$ of subsets of size $p$ and any given $Y\subseteq E$ with $|Y|\leq q$, a $q$-representative set $\widehat{\sets} \subseteq \sets$ contains a set $\widehat{X} \subseteq E$ with $\widehat{X}\cap Y = \emptyset$ and $\widehat{X} \cup Y$ being independent (i.e., $\widehat{X}\cup Y \in \mathcal{I}$) whenever $\sets$ has such a set satisfying the same condition.

The previous studies \cite{lovasz1977flats, marx2009parameterized, Kratsch12, fomin2016efficient} have been eager to find a representative set in polynomial time, which is also of independent interest, for a given representation matrix and we introduce the following theorem.

\begin{theorem}[\cite{fomin2016efficient}]
	\label{thm:fomin}
	Let $M = (E, \mathcal{I})$ be a linear matroid of rank $p + q = k$ given together with its representation matrix $A_M$ over a field $\F$.
	Let $\sets = \{S_1,...,S_t\}$ be a family of independent sets of size $p$.
	Then a $q$-representative family $\widehat{\sets}$ for $\sets$ with at most $\binom{p+q}p$ sets can be found in $O(\binom{p+q}p tp^{\omega} + t \binom{p+q}p^{\omega -1})$ operations over $\F$, where $\omega < 2.373$ is the matrix multiplication exponent.
\end{theorem}

Since a gammoid is a linear matroid and its notion of independence highly has to do with a minimum vertex cuts of two sets, the notion of a representative set has a key connection to the following problem:
In a directed graph $G$ with vertex subsets $S$ and $T$, can we find a set $Z$ which contains a minimum $(A, B)$ vertex cut for any $A\subseteq S$ and $B\subseteq T$?
Even though a na\"ive answer can be $Z = V(G)$, the identification of a representative set significantly reduces the size of $Z$ with dependencies on $p,q,$ and $|\sets|$ in the above setting.

\paragraph{The Subroutine}
We identify and contract unnecessary edges especially by using knowledge from vertex sparsification preserving a vertex connectivity.
To exploit such fruitful results, we leverage a natural correspondence between edges and vertices via working on the \emph{line graph} of a graph.
In this way, edges appearing in a minimum edge cut of a partition of terminals in the original graph are given by identifying their corresponding vertices in the line graph through Lemma~\ref{lem:kw12}, which contains a minimum vertex cut of any partition of terminals.

Before making this connection clear, we can make a further assumption by preprocessing the boundary edges of an induced subgraph.
This preprocessing relies on the following, which readily follows from Lemma~\ref{lem:ContractWellLinked}.

\begin{observation}
	\label{obs:SplitVertices}
	Let $G$ be a graph with terminals $\tset$ and $v\in V(G)$.
	A subdivision of an edge $uv$, which is to replace an edge $uv$ with a path $uwv$ through a new vertex $w$, results in a $(\tset, c)$-equivalent graph. 
\end{observation}

\vspace{-3mm}
\begin{proof}
	The set $\{u, w\}$ is a \cwell set.
\end{proof}

Recall in \textsc{PolySizedcNetwork} that when $H$ has at most $2c-1$ boundary edges, we mark the endpoints in $V(H)$ of the boundary edges $\partial H$ (i.e., $V(H)\cap V(\partial H)$) as tentative terminals $\widehat{\tset}$ and then replace $H$ with a smaller equivalent one.
In this case, despite $|\widehat{\tset}|\leq 2c-1$, we do not know how many incident edges $\widehat{\tset}$ would have.
By Observations~\ref{obs:SplitVertices}, we can assume not only that $|\widehat{\tset}| = O(c)$, but also that each terminal has degree $1$.

\begin{lemma}
	\label{lem:boundaryAssumption}
	When working on an induced subgraph $H$ of a graph $G$ with its tentative terminals $\widehat{\tset}$ coming from the endpoints of the boundary edges in $\partial H$ (i.e., $V(H) \cap V(\partial H)$),
	we may assume that each terminal in $\widehat{\tset}$ has degree $1$.
\end{lemma}

\begin{proof}
	Apply a subdivision of each boundary edge $uv \in \partial H$ with $v\in V(H)$, introducing one vertex $w_{uv}$.
	We extend $H$ to the induced subgraph $G[V(H) \cup \{w_{uv} : v \in V(H) \cap V(\partial H)\}]$, denoted by $H'$, and its boundary $\partial (H')$ becomes the edges $uw_{uv}$ for each $u\in V(H)^c \cap V(\partial H)$.

	When we work on $H'$ via Corollary~\ref{cor:RemoveEdges}, new terminals $\widehat{\tset}'$ becomes $\{w_{uv} | uv \in \partial H\}$, and it is clear that each terminal $w_{uv}$ has the unique neighbor $v$ in $H'$ (i.e., deg$_{H'}(w_{uv})=1$).
\end{proof}

Note that this manipulation on boundary of a piece has no impact on boundary of other pieces.
When sparsifying a set with $O(c)$ boundary edges, we can make the stronger assumption as in Lemma~\ref{lem:boundaryAssumption}; the piece has $O(c)$ tentative terminals and each terminal has degree $1$.

{
	\renewcommand{\thetheorem}{\ref{lem:c3}}
	\begin{lemma}
		Let $G = (V, E)$ be a graph with a set $\tset$ of $O(c)$ terminals and each terminal have degree $1$.
		There is a subset $E'$ of $E$ with $|E'| = O(c^3)$ and $G/(E\bs E')$ is a \sparsifier for $G$.
	\end{lemma}
	\addtocounter{theorem}{-1}
}

\begin{proof}
	Let $v(e)$ be the corresponding vertex in the line graph $L(G)$ for an edge $e$ in $E$ and $v_t$ be the unique neighbor in $G$ of each terminal $t$ in $\tset$.
	We enlarge the line graph $L(G)$ by adding a copy $t'$ of $t$ with its unique edge $t'v(tv_t)$.
	Also let $\tset'$ be the set of such terminal copies and $L(G)'$ be the enlarged one.

	Viewing $\tset'$ as $X$ in Lemma~\ref{lem:kw12}, we have a subset $Z$ of $V(L(G)')$ with $O(c^3)$ vertices, which contains a minimum $(A', B')$-vertex cut of any bipartition $(A', B')$ of $\tset'$.
	We slightly change $Z$ as follows: remove terminals $t'$ in $Z$ (if any) and add the unique neighbor $v(tv_t)$ of each terminal $t'$.
	Note that this perturbed set, denoted by $Z'$, still has $O(c^3)$ vertices but no intersection with $\tset'$.

	We show $Z'$ also contains a minimum vertex cut between any bipartition of terminals.
	Suppose that a partition $(A', B')$ of $\tset'$ has a minimum vertex cut $C$ overlapping with $\tset'$.
	For any $t' \in C \cap \tset'$, the minimum vertex cut $C$ does not include the unique neighbor $v(tv_t)$ of $t'$; otherwise $C-t'$ is a smaller minimum $(A', B')$-vertex cut.
	Thus we can replace $C$ by another minimum cut $C-t'+v(tv_t)$.
	Repeating this operation on all terminals in $C$, we end up having a minimum $(A', B')$-vertex cut disjoint from $\tset'$ such that the modified minimum vertex cut is contained in $Z'$ in light of the construction of $Z'$.

	Lastly, we take $E'$ as $\{e\in E(G) : v(e) \in Z'\}$ and claim $G/(E\bs E')$ is $(\tset, c)$-equivalent to $G$.
	For partition $(A, B)$ of $\tset$ in $G$, a minimum edge cut $C$ between $A$ and $B$ corresponds to a vertex cut between $A'$ and $B'$ that consists of the corresponding vertices of the edges in $C$, thus $\mincut_{L(G)'}(A', B')\leq \mincut_G(A,B)$.
	By similar reasoning for the opposite direction, we have $\mincut_{L(G)'}(A', B')\geq \mincut_G(A,B)$ and thus $\mincut_{L(G)'}(A', B') =  \mincut_G(A,B)$.
	It implies that even after contracting all edges in $E\bs E'$, we still retain an edge cut of size $\mincut_G(A, B)$.
\end{proof}

The last paragraph makes more sense by relying on the max-flow and min-cut theorem and the Menger's theorem.
For example, there are $\mincut_G(A, B)$ edge-disjoint paths between $A$ and $B$, which can be exactly transformed into $\mincut_G(A, B)$ vertex-disjoint paths between $A'$ and $B'$.
It implies that a minimum $(A', B')$-vertex cut has size at least $\mincut_G(A, B)$ (i.e., $\mincut_G(A, B)\leq \mincut_{L(G)'}(A', B')$).

After sparsifying the enlarged piece $H'$ with $O(c)$ boundary edges into a smaller equivalent graph with $O(c^3)$ edges, the all edges $w_{uv}v$ for each $v\in H$ still remain, since $Z'$ contains the vertex $v(w_{uv}v)$ and so $E'$ also contains $w_{uv}v$.
As $\{w_{uv}, v\}$ is a \cwell set, we may contract it as if there were no any operations introducing additional vertices $w_{uv}$ at the very beginning.
Therefore, our sparsifier with \optimal edges can be still found by only contracting edges, which means outputs of our algorithm in Section~\ref{ssec:Algorithm} have at most \optimal edges as well.

\section{More Efficient Algorithms for \sparsifiersTC}
\label{sec:EfficientAlgorithms}

In this section we present a faster algorithm at the expense of the size by using the notion of ``important'' edges elaborated in
Section~\ref{ssec:terminology}.

Equipped with these notions, we prove the existence of \sparsifiers constructed from important edges in Section~\ref{subsec:Existence}.
Then we revisit in Section~\ref{subsec:PolyTime} the original approach for speedup (in Section~\ref{subsec:Existence}) and achieve a result implying Theorem~\ref{thm:Main} Part~\ref{part:Main2} by utilizing expander decomposition and local cut algorithms in Section~\ref{sec:ExpanderLocalCut}.

\subsection{Equivalence, Cut Containment, and Cut Intersection}
\label{ssec:terminology}

Our faster recursive algorithm works more directly with the notion of equivalence defined in Definition \ref{def:DefMain}.
This algorithm identifies a set of important edges, $\Ehat$, and forms $H$ by contracting all edges in $E \setminus \Ehat$.

Observe that as long as $\Ehat$ is small, contracting $E \setminus \Ehat$ still results in a graph with few vertices and edges.
Therefore, our goal is find a set $\Ehat$ of important edges to keep in $H$ such that the size of $\Ehat$ is not much larger than $|\tset|.$
We will show for the purpose of being $(\tset, c)$-equivalent, a sufficient condition is that every $(\tset,c)$-cut can be formed using edges from only $\Ehat$.
This leads to the definition of $\Ehat$ containing all $(\tset,c)$-cuts, which was also used in~\cite{MolinaS18:unpublished} for the $c \le 5$ setting.

\begin{definition}[Cut containment]
	\label{def:Contain}
	In a graph $G = (V, E)$ with terminals $\tset$, a subset of edges $\Econtain \subseteq E$ is said to contain all $(\tset, c)$-cuts if for any partition $\tset = \tset_1 \cupdot \tset_2$ with $\mincut_G(\tset_1,\tset_2) \le c$ there is a cut $F \subseteq \Econtain$ such that
	\bi
	\item[1.] $F$ has size equal to $\mincut_G(\tset_1, \tset_2)$,
	\item[2.] $F$ is also a cut between $\tset_1$ and $\tset_2$.
	That is, $\tset_1$ and $\tset_2$ are disconnected in $G \setminus F$.
	\ei
\end{definition}

Note that this is different than containing all the minimum cuts: on a length $n$ path with two endpoints as terminals, any intermediate edge contains a minimum terminal cut, but there are up to $n-1$ different such minimum cuts.

If $\Econtain$ contains all $(\tset,c)$-cuts, we may contract all edges in $E \setminus \Econtain$ to obtain a $(\tset, c)$-equivalent graph $H$ of $G$.
\begin{lemma}
	\label{lem:EdgesToSparsifier}
	If $G = (V, E)$ is a connected graph with terminals $\tset$, and $\Econtain$ is a subset of edges that contain all $(\tset, c)$-cuts, then the graph
	\[
	H = G / \left(E \bs \Econtain \right)
	\]
	is $(\tset, c)$-equivalent to $G$, and has at most $|\Econtain| + 1$ vertices.
\end{lemma}

\begin{proof}
	Consider any cut using entirely edges in $\Econtain$: contracting edges from $E \bs \Econtain$ will bring together vertices on the same side of the cut.
	Therefore, the separation of vertices given by this cut also exists in $H$ as well.
	
	To bound the size of $H$, observe that contracting all edges of $G$ brings it to a single vertex.
	That is, $H / \Econtain$ is a single vertex: uncontracting an edge can increase the number of vertices by at most $1$, so $H$ has at most $|\Econtain| + 1$ vertices.
\end{proof}

We can also state Corollary \ref{cor:RemoveEdges} and Lemma \ref{lem:Independent} in the language of edge containment.

\begin{lemma}
	\label{lem:ContainRemoveEdges}
	Let $\Ehat$ be a set of edges in $G$ with endpoints $V(\Ehat)$, and $\tset$ be terminals in $G$.
	If edges $\Econtain$ contain all $(\tset \cup V(\Ehat), c)$-cuts in $G \bs \Ehat$, then $\Econtain \cup \Ehat$ contains all $(\tset, c)$-cuts in $G$.
\end{lemma}

\begin{lemma}
	\label{lem:ContainIndependent}
	If the edges $\Econtain_1 \subseteq E(G_1)$ contain all $(\tset_1,c)$-cuts in $G_1$, and the edges $\Econtain_2 \subseteq E(G_2)$ contain all $(\tset_2,c)$-cuts in $G_2$, then $\Econtain_1 \cup \Econtain_2$ contains all the $(\tset_1 \cup \tset_2, c)$-cuts in the vertex disjoint union of $G_1$ and $G_2$.
\end{lemma}

These motivate us to gradually build up $\Econtain$ through a further intermediate definition.

\begin{definition}
	\label{def:Intersect}
	In a graph $G = (V, E)$ with terminals $\tset$, a subset of edges $\Eintersect \subseteq E$ intersects all $(\tset, c)$-cuts for some $c > 0$ if for any partition $\tset = \tset_1 \cupdot \tset_2$ with $\mincut_G(\tset_1, \tset_2) \leq c$, there exists a cut $F = E(V_1, V_2)$ such that:
	\bi
	\item[1.] $F$ has size $\mincut_G(\tset_1,\tset_2)$,
	\item[2.] $F$ induces the same separation of $\tset$: $V_1 \cap \tset = \tset_1$, $V_2 \cap \tset = \tset_2$.
	\item[3.] $F$ contains at most $c - 1$ edges from any connected component of $G \bs \Eintersect$.
	\ei
\end{definition}

\paragraph{Reduction to Cut Intersection}
Based on Definition~\ref{def:Intersect}, we can reduce the problem of finding a set $\Econtain$ which contains all $(\tset, c)$-cuts to the problem of finding a set $\Eintersect$ which intersects all small cuts.
Formally, the deletion of an intersecting edge set $\Eintersect$ separates $(\tset,c)$-cuts of $G$ into edge sets of size $c - 1$.
Each of these smaller cuts happens on one of the connected components of $G \setminus \Eintersect$, and can thus be considered independently when we construct the containing sets of $G \setminus \Eintersect$.

This is done by first finding an intersecting set $\Eintersect$, and then recursing on each component in the (disconnected) graph with $\Eintersect$ removed, but with the endpoints of $\Eintersect$ included as terminals as well.
This will increase the number of edges and terminals, but allow us to focus on $(c-1)$-connectivity in the components, leading to our recursive scheme.
The overall algorithm simply iterates this process until $c$ reaches $1$, as shown in Figure~\ref{fig:GetContainingEdges}.
All arguments until now can be summarized as the following stitching lemma.

\begin{lemma}
	\label{lem:IntersectToContain}
	Let $G = (V, E)$ be a graph with terminals $\tset$, and $\Eintersect \subseteq E$ be a set of edges that intersects all $(\tset, c)$-cuts.
	For $\tsethat = \tset \cup V\left( \Eintersect \right)$, let $\Econtain \subseteq E \setminus \Eintersect$ be a set of edges that contains all $(\tsethat, c-1)$ cuts in the graph $(V, E \bs \Eintersect)$.
	Then $\Econtain \cup \Eintersect$ contains all $(\tset, c)$-cuts in $G$.
\end{lemma}

\begin{proof}
	Consider a partition $\tset = \tset_1 \cupdot \tset_2$ with $\mincut_{G}(\tset_1, \tset_2) \le c$.
	Since $\Eintersect$ intersects all $(\tset,c)$-cuts, there is cut of size $\hat{c}$ separating $\tset_1$ and $\tset_2$, which has at most $c-1$ edges in each component of $G \bs \Eintersect$.
	
	Combining this with Lemmas \ref{lem:ContainRemoveEdges} and \ref{lem:ContainIndependent} shows that if $\Econtain$ contains all $(\tset \cup V(\Eintersect),c-1)$-cuts in $G\bs \Eintersect$, then $\Econtain \cup \Eintersect$ contains all $(\tset,c)$-cuts in $G$.
\end{proof}

\begin{figure}
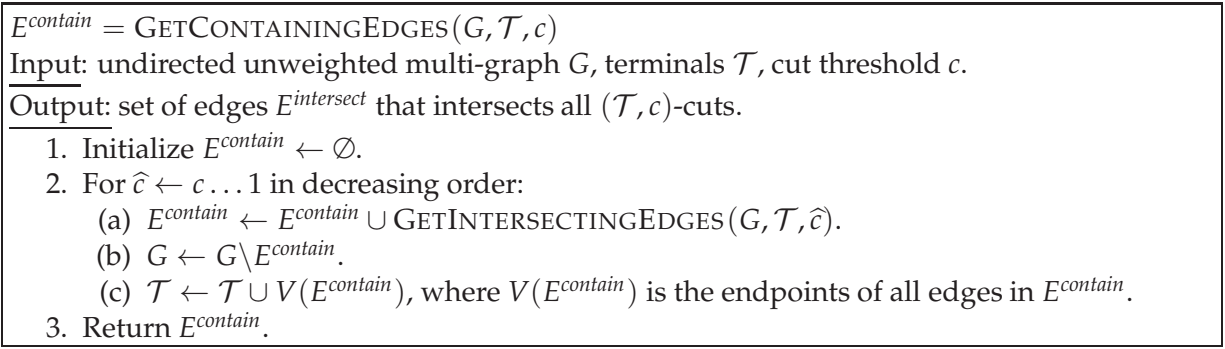

	\begin{algbox}
		$\Econtain = \textsc{GetContainingEdges}(G, \tset, c)$\\
		\underline{Input}: undirected unweighted multi-graph $G$,
		terminals $\tset$, cut threshold $c$.\\
		\underline{Output:} set of edges $\Eintersect$ that intersects all
		$(\tset,c)$-cuts.
		\bi
			\item[1.] Initialize $\Econtain \leftarrow \emptyset$.
			\item[2.] For $\chat \leftarrow c \ldots 1$ in decreasing order:
			\bi
				\item[(a)] $\Econtain \leftarrow \Econtain \cup \textsc{GetIntersectingEdges}(G, \tset, \chat)$.
				\item[(b)] $G \leftarrow G \bs \Econtain$.
				\item[(c)] $\tset \leftarrow \tset \cup V(\Econtain)$, where $V(\Econtain)$ is the endpoints of all edges in $\Econtain$.
			\ei
			\item[3.] Return $\Econtain$.
		\ei
	\end{algbox}
	\vspace{-4mm}
	\caption{Pseudocode for finding a set of edges that contain all the $(\tset,c)$-cuts.}
	\label{fig:GetContainingEdges}
\end{figure}

The following theorem shows the bounds for generating a set of edges $\Eintersect$ that intersects all $(\tset, c)$-cuts.
Its correctness follows from Lemma~\ref{lem:IntersectToContain}.
\begin{theorem}
	\label{thm:IntersectingEdges}
	For any parameter $\phi$, value $c$, and graph $G$ with terminals $\tset$, there exists an algorithm that generates a set of edges $\Eintersect$ that intersects all $(\tset, c)$-cuts:
	\vspace{-3mm}
	\begin{enumerate} \setlength{\itemsep}{-3pt}
		\item	\label{part:IntersectingEdges:Slower} with size at most $O((\phi m \log^4{n} + |\tset|) \cdot c )$ in $\O(m(c\phi^{-1})^{2c})$ time.
		\item 	\label{part:IntersectingEdges:Faster} with size at most $O((\phi m \log^{4}n + |\tset|) \cdot c^2)$ in $\O(m \phi^{-2} c^7)$ time.
	\end{enumerate}
\end{theorem}

While Theorem~\ref{thm:IntersectingEdges} Part~\ref{part:IntersectingEdges:Slower} developed in Section~\ref{subsec:Existence} provides a slow subroutine, we are able to modify the argument in Section~\ref{subsec:PolyTime} and then take further steps, expander decomposition and local cut algorithms, in Section~\ref{sec:ExpanderLocalCut} to obtain Theorem~\ref{thm:IntersectingEdges} Part~\ref{part:IntersectingEdges:Faster}.

In essence, we use Theorem \ref{thm:IntersectingEdges} Part \ref{part:IntersectingEdges:Faster} to prove Theorem~\ref{thm:Main} Part~\ref{part:Main2} in Appendix~\ref{proofs:Main2}.
As the size of $\Econtain$ multiplies by $O(c^2)$ every iteration, the total size of $\Econtain$ at the end is $O(c)^{2c}$, as desired in Theorem~\ref{thm:Main} Part~\ref{part:Main2}.

\subsection{Efficient Algorithm: Recursive Constructions}
\label{sec:SecondAlgo:Recursion}

In this section, we give recursive algorithms for finding sets of edges that intersect all $(\tset,c)$-cuts (as defined in Definition~\ref{def:Intersect}).
In Section~\ref{subsec:Existence}, we show the existence of a small set of edges that intersects all $(\tset, c)$-cuts.
In Section~\ref{subsec:PolyTime}, we give a polynomial-time construction that outputs a set of edges whose size is slightly larger than those given in Section~\ref{subsec:Existence}.

Our routines are based on recursive contractions.
Suppose we have found a terminal cut $F = E(V_1, V_2)$.
Then any cut $\Fhat$ overlapping with both $G[V_1]$ and $G[V_2]$, or $F$, will have only at most $c-1$ edges in common with $G[V_1]$ or $G[V_2]$.
Thus, it suffices for us to focus on cuts that lie entire in one half, which we assume without loss of generality is $G[V_1]$.

Since none of the edges in $F$ and $G[V_2]$ are used, we can work equivalently on the graph with all these edges contracted.
The progress made by this, on the other hand, may be negligible: consider, for example, the extreme case of $F$ being a matching, and no edges are present in $G[V_2]$.
On the other hand, if $G[V_2]$ is connected, then it will become a single terminal vertex in addition to $V_1$, and the two halves that we recurse on add up to a size that's only slightly larger than $G$.

Thus, our critical starting point is to look for cuts $(V_1, V_2)$, where both $G[V_1]$ and $G[V_2]$ are connected, and contain two or more terminals.
We first look for such cuts through exhaustive enumeration in Section~\ref{subsec:Existence}, and show that when none are found, we can simply terminate by taking all minimum cuts with one terminal on one side, and the other terminals on the other side.
Unfortunately, we do not have a polynomial time algorithm for determining the existence of a cut $(V_1,V_2)$ with size at most $c$ such that both $G[V_1]$ and $G[V_2]$ are connected and have at least $2$ terminals.

In Section~\ref{subsec:PolyTime}, we take a less direct, but poly-time computable approach based on computing the minimum terminal cut among the terminals $\tset$.
Both sides of this cut are guaranteed to be connected by the minimality of the cut.
However, we cannot immediately recurse on this cut due to it possibly containing only one terminal on one side.
We address this by defining maximal terminal separating cuts: minimum cuts with only that terminal on one side, but containing as many vertices as possible.
The fact that such cuts can only grow $c$ times until their sizes exceed $c$ allows us to bound the number of cuts recorded by the number of terminals, times an extra factor of $c$, for a total of $O(kc^2)$ edges in the sparsifier.

\subsubsection{Existence}
\label{subsec:Existence}

Our divide-and-conquer scheme relies on the following observation about when $(\tset, c)$-cuts are able to interact completely with both sides of a cut.

\begin{lemma}
	\label{lem:Partition}
	Let $F$ be a cut given by the partition $V = V_1 \cupdot V_2$ in $G = (V, E)$ such that both $G[V_1]$	and $G[V_2]$ are connected, and $\tset_1 = V_1 \cap \tset$ and $\tset_2 = V_2 \cap \tset$ be the partition of $\tset$ induced by this cut.
	If $\Eintersect_1$ intersects all $(\tset_1 \cup \{v_2\}, c)$-terminal cuts in $G / V_2$, the graph formed by contracting all of $V_2$ into a single vertex $v_2$, and similarly $\Eintersect_2$ intersects all $(\tset_2 \cup \{v_1\}, c)$-terminal cuts in $G / V_1$, then $\Eintersect_1 \cup \Eintersect_2 \cup F$ intersects all $(\tset, c)$-cuts in $G$ as well.
\end{lemma}

\begin{proof}
	Consider some cut $\Fhat$ of size at most $c$.
	If $\Fhat$ uses an edge from $F$, then it has at most $c - 1$ edges in $G \bs F$, and thus in any connected component as well.
	If $\Fhat$ has at most $c - 1$ edges in $G[V_1]$, then every connected component in $(G \setminus F) \setminus \Eintersect_1$ has at most $c - 1$ edges from $\Fhat$. This follows because removing $F$ has already disconnected $V_1$ and $V_2$, and removing $\Eintersect_1$ can only further disconnects the remaining components.
	
	The only remaining case is when $\Fhat$ is entirely contained in one of the sides.
	Without loss of generality assume that $\Fhat$ is entirely contained in $V_1$, i.e., $\Fhat \subseteq E(G[V_1])$.
	Because no edges from $G[V_2]$ are removed and $G[V_2]$ is connected, all of $\tset_2$ must be on one side of the cut, and can therefore be represented by a single vertex $v_2$.
	Using the induction hypothesis on the cut $\Fhat$ in $G / V_2$ with the terminal separation given by all of $\tset_2$ replaced by $v_2$ gives that $\Fhat$ has at most $c - 1$ edges in any connected component of $\left( G / V_2 \right) \bs \Eintersect_1$.
	Since connected components remain intact under contracting connected subsets, we conclude that $\Fhat$ has at most $c - 1$ edges in any connected components of $G \bs \Eintersect_1$ as well.
\end{proof}

However, to make progress on such a partition, we need to contract at least two terminals together with either $V_1$ or $V_2$.
This leads to our key definition of a non-trivial $\tset$-separating cut:
\begin{definition}
	\label{def:Nontrivial}
	A non-trivial $(\tset,c)$-cut is a separation of $V$ into $V_1 \cupdot V_2$ such that:
	\bi
		\item[1.] the subgraphs induced by $V_1$ and $V_2$ (i.e., $G[V_1]$ and $G[V_2]$) are both connected.
		\item[2.] $|V_1 \cap \tset| \geq 2$, $|V_2 \cap \tset| \geq 2$.
	\ei
\end{definition}

Such cuts are critical for partitioning and recursing on the two resulting pieces.
The connectivity of $G[V_1]$ and $G[V_2]$ is necessary for applying Lemma \ref{lem:Partition}, and $|V_1 \cap \tset| \geq 2$, $|V_2 \cap \tset| \geq 2$ are necessary to ensure that making this cut and recursing makes progress.

We now study the set of graphs $G$ and terminals $\tset$ for which a non-trivial cut exists.
For example, consider the case when $G$ is a star graph (a single vertices with $n-1$ vertex connected to it) and all vertices are terminals.
In this graph, the side of the cut not containing the center can only have a single vertex; hence, there are no non-trivial cuts.

We can, in fact, prove the converse: if no such interesting separations exist, we can terminate by only considering the $|\tset|$ separations of $\tset$ formed with one terminal on one of the sides.
We define these cuts to be the $s$-isolating cuts.

\begin{definition}
	\label{def:IsolatingCut}
	For a graph $G$ with terminal set $\tset$ and some $s \in \tset$, an $s$-isolating cut is a partition of the vertices $V = V_{A} \cupdot V_{B}$ such that $s$ is the only terminal in $V_A$, i.e., $s \in V_{A}$, $(\tset \bs \{u\}) \subseteq V_{B}$.
\end{definition}

\begin{lemma}
	\label{lem:OPLemma}
	If $\tset$ is a subset of at least $4$ terminals in an undirected graph $G$ that has no non-trivial $\tset$-separating cut of size at most $c$, then the union of all $s$-isolating cuts of size at most $c$:
	\[
	\Eintersect
	\leftarrow
	\bigcup_{\substack{s \in \tset \\
			\mincut_G\left(\left\{s\right\}, \tset \bs \left\{s\right\}\right) \leq c}}
	\Mincut\left(G,\left\{s\right\}, \tset \bs \left\{s\right\}\right),
	\]
	contains all $(\tset, c)$-cuts of $G$. 
\end{lemma}

\begin{figure}[!h]
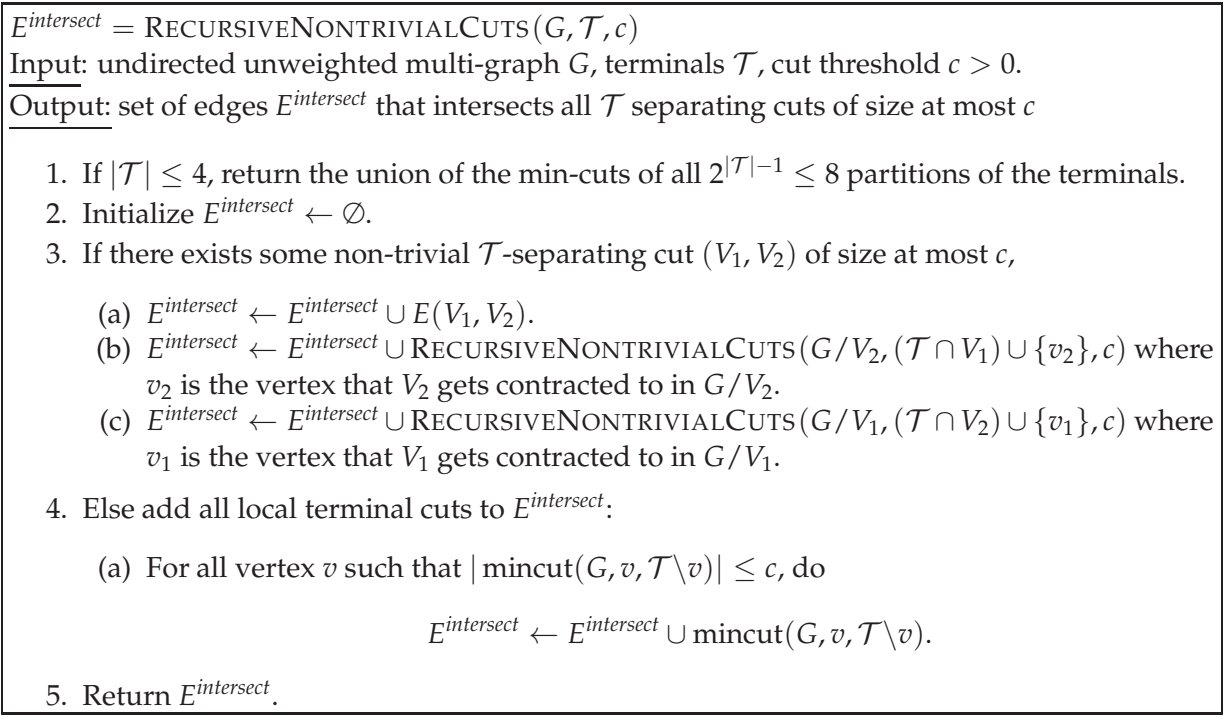

	\begin{algbox}
		$\Eintersect = \textsc{RecursiveNontrivialCuts}(G, \tset, c)$
		
		\underline{Input}: undirected unweighted multi-graph $G$,
		terminals $\tset$,
		cut threshold $c > 0$.
		
		\underline{Output:} set of edges $\Eintersect$ that intersects all $\tset$
		separating cuts of size at most $c$
		\begin{enumerate}
			\setlength{\itemsep}{-3pt}
			\item If $|\tset| \leq 4$, return the union of the
			min-cuts of all $2^{|\tset| - 1} \leq 8$ partitions of the terminals.
			\item Initialize $\Eintersect \leftarrow \emptyset$.
			\item \label{ln:Partition} If there exists some non-trivial
			$\tset$-separating cut $(V_1, V_2)$ of size at most $c$,
			\begin{enumerate}
				\setlength{\itemsep}{-2pt}
				\item $\Eintersect \leftarrow \Eintersect \cup E(V_1, V_2)$.
				\item $\Eintersect \leftarrow \Eintersect \cup \textsc{RecursiveNontrivialCuts}(G / V_2, (\tset \cap V_1) \cup \{v_2\}, c)$ where $v_2$ is the vertex that $V_2$
				gets contracted to in $G / V_2$.
				\item $\Eintersect \leftarrow \Eintersect \cup \textsc{RecursiveNontrivialCuts}(G / V_1, (\tset \cap V_2) \cup \{v_1\}, c)$ where $v_1$ is the vertex that $V_1$
				gets contracted to in $G / V_1$.
			\end{enumerate}
			\item Else add all local terminal cuts to $\Eintersect$:
			\begin{enumerate}
				\setlength{\itemsep}{0pt}
				\item \label{ln:OP} For all vertex $v$ such that $|\Mincut(G, v, \tset \bs v)| \leq c$,
				do
				\[
				\Eintersect \leftarrow \Eintersect \cup
				\Mincut(G, v, \tset \bs v).
				\]
			\end{enumerate}
			\item Return $\Eintersect$.
		\end{enumerate}
	\end{algbox}
	\vspace{-4mm}
	\caption{Algorithm for finding a set of edges that intersects all terminal cuts of size $\leq c$.}
	\label{fig:RecursiveNonTrivialCuts}
\end{figure}

\begin{proof}
	Consider a graph with no non-trivial $\tset$-separating cut of size at most $c$, but there is a partition of $\tset$, $\tset = \tset_1 \cupdot \tset_2$, such that the minimum cut between $\tset_1$ and $\tset_2$, $V_1$ and $V_2$, has at most $c$ edges, and $|\tset_1|, |\tset_2| \ge 2.$
	
	Let $F$ be one such cut, and consider the graph
	\[
	\Ghat = G / \left(E \bs F\right),
	\]
	that is, we contract all edges except the ones on this cut.
	Note that $\Ghat$ has at least $2$ vertices.
	
	Consider a spanning tree $T$ of $\Ghat$.
	By minimality of $F$, each node of $T$ must contain at least one terminal. Otherwise, we can keep one edge from such a node without affecting the distribution of terminal vertices.
	
	We now show that no vertex of $T$ can contain $|\tset|-1$ terminals.
	If $T$ has exactly two vertices, then one vertex must correspond to $\tset_1$ and one must correspond to $\tset_2$, so no vertex has $|\tset|-1$ terminals.
	If $T$ has at least $3$ vertices, then because every vertex contains at least one terminal, no vertex in $T$ can contain $|\tset|-1$ vertices. 
	
	Also, each leaf of $T$ can contain at most one terminal, otherwise deleting the edge adjacent to that leaf forms a non-trivial cut.
	
	Now consider any non-leaf node of the tree, say $r$.
	As $r$ is a non-leaf node, it has at least two different neighbors that lead to leaf vertices.
	
	Let us make $r$ the root of this tree and consider some neighbor of $r$, say $x$.
	If the subtree rooted at $x$ has more than $2$ terminals, then cutting the $rx$ edge results in two components, each containing at least two terminals (the component including $r$ has at least one other neighbor that contains a terminal).
	Thus, the subtree rooted at $x$ can contain at most one terminal, and must therefore be a singleton leaf.
	
	Hence, the only possible structure of $T$ is a star centered at $r$ (which may contain multiple terminals) in which each leaf has exactly one terminal in it.
	This in turn implies that $\Ghat$ must also be a star, i.e., $\Ghat$ has the same edges as $T$ but possibly with multi-edges.
	This is because any edges incident to a leaf of a star forms a connected cut.
	
	By minimality, each cut separating the root from a leaf is a minimal cut for that single terminal, and these cuts are disjoint.
	Thus, taking the union of edges of all these singleton cuts gives a cut that partitions $\tset$ in the same way and has the same size.	
\end{proof}

Note that Lemma \ref{lem:OPLemma} is not claiming all the $(\tset,c)$-cuts of $\tset$ are singletons. Instead, it says that any $(\tset,c)$-cut can be formed from a union of single terminal cuts.

Combining Lemma \ref{lem:Partition} and \ref{lem:OPLemma}, we obtain the recursive algorithm in 
Figure~\ref{fig:RecursiveNonTrivialCuts}, which demonstrates the existence of $O(|\tset| \cdot c)$ sized $(\tset, c)$-cut-intersecting subsets. If there is a non-trivial $\tset$-separating cut, the algorithm in Line~\ref{ln:Partition} finds it and recurses on both sides of the cut using Lemma \ref{lem:Partition}.
Otherwise, by Lemma \ref{lem:OPLemma}, the union of the $s$-isolating cuts of size at most $c$ contains all $(\tset,c)$-cuts, so the algorithm keeps the edges of those cuts in Line~\ref{ln:OP}.

\begin{lemma}
	\label{lem:RecursiveNontrivialCutsCorrectness}
	\textsc{RecursiveNontrivialCuts} as shown in 
	Figure~\ref{fig:RecursiveNonTrivialCuts} correctly returns
	a set of $(\tset, c)$-cut-intersecting edges of size at most
	$O(|\tset| \cdot c)$.
\end{lemma}

\begin{proof}
	The correctness of the algorithm can be argued by induction.
	The base case, where we terminate by adding all min-cuts with one terminal on one side, follows from Lemma~\ref{lem:OPLemma}, while the inductive case follows from applying Lemma~\ref{lem:Partition}.
	
	It remains to bound the size of $\Eintersect$ returned.
	Once again there are two cases: The first case is when we terminate with the union of singleton cuts. Each such cut has size at most $c$, thus summing to the total of $|\tset| \cdot c$.
	
	The second is the recursive case, which can be viewed as partitioning $k \geq 4$ terminals into two instances of sizes $k_1$ and $k_2$ where $k_1 + k_2 = k + 2$ and $k_1, k_2 \geq 3$.
	Note that the total values of $|\textsc{Terminals}| - 2$ across all the recursion instances is strictly decreasing, and is always positive.
	So, the recursion can branch at most $|\tset|$ times, implying that the total number of edges added is at most $O(c \cdot |\tset|)$.
\end{proof}

In fact, we may modify \textsc{RecursiveNontrivialCuts} to take extra steps for marginal speedup by utilizing expander decomposition.

\begin{proofof}{Theorem~\ref{thm:IntersectingEdges} Part~\ref{part:IntersectingEdges:Slower}}
	First, we perform expander decomposition, remove the inter-cluster edges, and add their endpoints as terminals.
	
	Now, we describe the modifications to \textsc{RecursiveNontrivialCuts} that make it efficient.
	
	Lemma~\ref{lem:AddEdges} and Lemma~\ref{lem:Independent} allow us to consider the pieces separately.
	
	Now at the start of each recursive call, enumerate all cuts of size at most $c$, and store the vertices on the smaller side, which by Equation~\ref{eq:ExpanderImba} above has size at most $O(c \phi^{-1})$.
	When such a cut is found, we only invoke recursion on the smaller side (in terms of volume).
	For the larger piece, we can continue using the original set of cuts found during the search.
	
	To use a cut from a pre-contracted state, we need to:
	\bi
	\item[1.] check if all of its edges remain (using a union-find data structure).
	\item[2.] check if both portions of the graph remain connected upon removal of this cut -- this can be done by explicitly checking the smaller side, and certifying the bigger side using a dynamic connectivity data structure by removing all edges from the smaller side.
	\ei
	Since we contract each edge at most once, the total work done over all the larger side is at most
	\[ \O\left( m \left( c\phi^{-1} \right)^{2c} \right), \] where we have included the logarithmic factors from using the dynamic connectivity data structure.
	Furthermore, the fact that we only recurse on things with half as many edges ensures that each edge participates in the cut enumeration process at most $O(\log n)$ times.
	Combining these then gives the overall running time.
\end{proofof}

\subsubsection{Polynomial-Time Construction}
\label{subsec:PolyTime}

It is not clear to us how the previous algorithm in Section~\ref{subsec:Existence} could be implemented in polynomial time.
While incorporating expander decomposition makes our algorithm faster, its running time still has the $\log^{O(c)}{n}$ term (as stated in Theorem~\ref{thm:IntersectingEdges}~Part~\ref{part:IntersectingEdges:Slower}).
In this section, we give a more efficient algorithm that returns sparsifiers of larger size, but ultimately leads to the faster running time given in Theorem~\ref{thm:IntersectingEdges}~Part~\ref{part:IntersectingEdges:Faster}.
It was derived by working backwards from the termination condition of taking all the cuts with one terminal on one side in Lemma~\ref{lem:OPLemma}.

Recall that a terminal cut is a cut with at least one terminal one both sides.
The algorithm has the same high level recursive structure, but it instead only finds the minimum terminal cut or certifies that its size is greater than $c$.
This takes $O(m+nc^3 \log n)$ time using an algorithm by Cole and Hariharan~\cite{ColeH03}.
\begin{theorem}[Minimum Terminal Cut \cite{ColeH03}]
\label{thm:ColeH03}
Given graph $G$ with terminals $\tset$ and constant $c$, there is an $O(m+nc^3\log n)$ time algorithm which computes the minimum terminal cut on $\tset$ or certifies that its size is greater than $c$.
\end{theorem}

It is direct to check that both sides of a minimum terminal cut are connected.
This is important towards our goal of finding a non-trivial $\tset$-separating cut, defined in Definition \ref{def:Nontrivial}.

\begin{lemma}
	If $(V_A, V_B)$ is the global minimum $\tset$-separating cut in a connected graph $G$, then both $G[V_A]$ and $G[V_B]$ must be connected.
\end{lemma}

\begin{proof}
	Suppose for the sake of contradiction that $V_{A}$ is disconnected as $V_{A} = V_{A1} \cupdot V_{A2}$.
	Without loss of generality assume $V_{A1}$ contains a terminal.
	Also, $V_{B}$ contains at least one terminal because $(V_A, V_B)$ is $\tset$-separating.
	
	Then as $G$ is connected, there is an edge between $V_{A1}$ and $V_{B}$.
	Then the cut $(V_{A1}, V_{A2} \cup V_{B})$ has strictly fewer edges crossing, and also terminals on both sides, a contradiction to $(V_A, V_B)$ being the minimum $\tset$-separating cut. 
\end{proof}

The only bad case that prevents us from recursing is when the minimum terminal cut has a single terminal $s$ on some side.
That is, one of the $s$-isolating cuts from Definition~\ref{def:IsolatingCut} is also a minimum terminal cut.
We can cope with it via an extension of Lemma \ref{lem:Partition}.
Specifically, we show that for a cut with both sides connected, we can contract a side of the cut along with the cut edges before recursing.
\begin{lemma}
	\label{lem:IsolateContract}
	Let $F$ be a cut given by the partition $V = V_1 \cupdot V_2$ in $G = (V, E)$ such that both $G[V_1]$	and $G[V_2]$ are connected, and $\tset_1 = V_1 \cap \tset$ and $\tset_2 = V_2 \cap \tset$ be the partition of $\tset$ induced by this cut.
	If $\Eintersect_1$ intersects all $(\tset_1 \cup \{v_2\}, c)$-terminal cuts in $G/V_2/F$, the graph formed by contracting all of $V_2$ and all edges in $F$ into a single vertex $v_2$, and similarly $\Eintersect_2$ intersects all $(\tset_2 \cup \{v_1\}, c)$-terminal cuts in $G/V_1/F$, then $\Eintersect_1 \cup \Eintersect_2 \cup F$ intersects all $(\tset, c)$-cuts in $G$.
\end{lemma}

\begin{proof}
	Consider some cut $\Fhat$ of size at most $c$.
	If $\Fhat$ uses an edge from $F$, then it has at most $c - 1$ edges in $G \bs F$, and thus in any connected component as well.
	If $\Fhat$ has at most $c - 1$ edges in $G[V_1]$, then no connected component in $V_1$ can have $c$ or more edges because removing $F$ already disconnected $V_1$ and $V_2$, and removing $\Eintersect_1$ can only further disconnect the remaining components.
	
	The only remaining case is when $\Fhat$ is entirely contained on one of the sides.
	Without loss of generality assume $\Fhat$ is entirely contained in $V_1$, i.e., $\Fhat \subseteq E(G[V_1])$.
	Because no edges from $G[V_2]$ and $F$ are removed and $G[V_2]$ is connected, all edges in $G[V_2]$ and $F$ must not be cut and hence can be contracted into a single vertex $v_2$.
	Using the induction hypothesis on the cut $\Fhat$ in $G/V_2/F$ with the terminal separation given by all of $\tset_2$ replaced by $v_2$ gives that $\Fhat$ has at most $c-1$ edges in any connected component of $\left(G/V_2/F \right) \bs \Eintersect_1$.
	Since connected components are unchanged under contracting connected subsets, we get that $\Fhat$ has at most $c-1$ edges in any connected components of $G \bs \Eintersect_1$ as well.
\end{proof}

Now, a natural way to handle the case where a minimum terminal cut has a single terminal $s$ on some side is to use Lemma \ref{lem:IsolateContract} to contract across the cut to make progress.
However, it may be the case that for some $s \in \tset$, there are many minimum $s$-isolating cuts: consider for example the length $n$ path with only the endpoints as terminals.
If we always pick the edge closest to $s$ as the minimum $s$-isolating cut, we may have to continue $n$ rounds, and thus add all $n$ edges to our set of intersecting edges.

To remedy this, we instead pick a ``maximal" $s$-isolating minimum cut.
One way to find a maximal $s$-isolating cut is to repeatedly contract across an $s$-isolating minimum cut using Lemma \ref{lem:IsolateContract} until its size increases.
At that point, we add the last set of edges found in the cut to the set of intersecting edges.
We have made progress because the value of the minimum $s$-isolating cut in the contracted graph must have increased by at least $1$.
While there are many ways to find a maximal $s$-isolating minimum cut, the way described here extends to our analysis in Section \ref{sec:ExpanderLocalCut}.

Pseudocode of this algorithm is shown in  Figure~\ref{fig:GetIntersectingEdgesTerminal}, and the procedure for the repeated contractions to find a maximal $s$-isolating cut described in the above paragraph is in Line~\ref{ln:IsolateContract}.

\begin{figure}[!t]
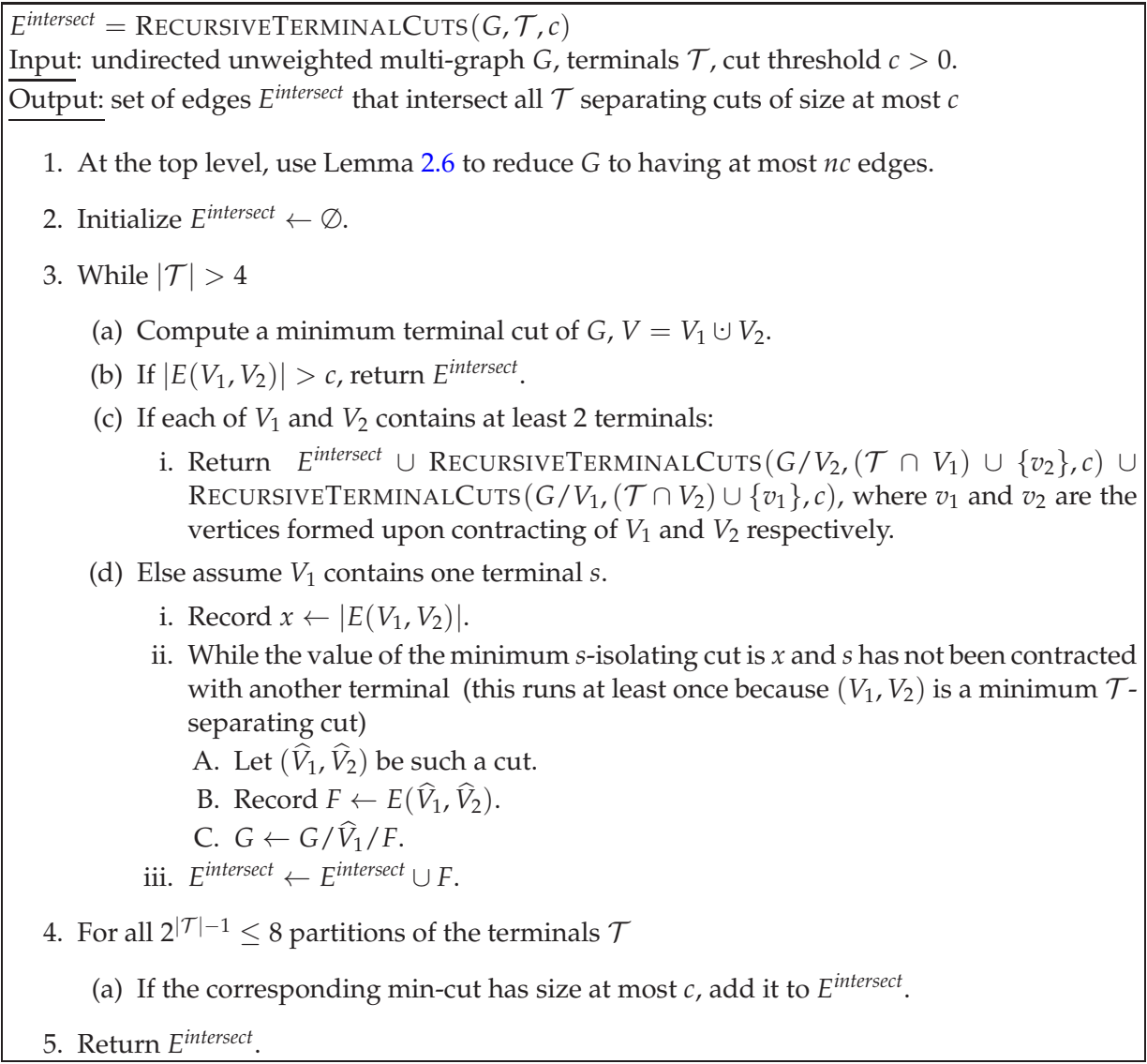

	\begin{algbox}
		$\Eintersect = \textsc{RecursiveTerminalCuts}(G, \tset, c)$
		
		\underline{Input}: undirected unweighted multi-graph $G$,
		terminals $\tset$, cut threshold $c > 0$.
		
		\underline{Output:} set of edges $\Eintersect$ that intersect all $\tset$
		separating cuts of size at most $c$
		
		\begin{enumerate}
			\item At the top level, use Lemma \ref{lem:EdgeReduction} to reduce $G$ to
			having at most $nc$ edges.
			\item Initialize $\Eintersect \leftarrow \emptyset$.
			\item While $|\tset| > 4$
			\begin{enumerate}
				\item Compute a minimum terminal cut of $G$, $V = V_1 \cupdot V_2$. \label{ln:compmincut}
				\item If $|E(V_1, V_2)| > c$, return $\Eintersect$.
				\item If each of $V_{1}$ and $V_2$ contains at least $2$ terminals:
				\begin{enumerate}
					\item Return $\Eintersect \cup \textsc{RecursiveTerminalCuts}(G/V_2, (\tset \cap V_1) \cup \{v_2\}, c) \cup \textsc{RecursiveTerminalCuts}(G / V_1, (\tset \cap V_2) \cup \{v_1\}, c)$, where $v_1$ and $v_2$ are the vertices
					formed upon contracting of $V_1$ and $V_2$ respectively. \label{ln:Return}
				\end{enumerate}
				\item Else assume $V_1$ contains one terminal $s$. \label{ln:IsolateContract}
				\begin{enumerate}
					\item Record $x \leftarrow |E(V_1, V_2)|$.
					\item While the value of the minimum $s$-isolating cut is $x$ and $s$ has not been
					contracted with another terminal \label{ln:Condition}
					\label{ln:LoopToMaximal}
					(this runs at least once because $(V_1, V_2)$ is a minimum $\tset$-separating cut)
					\begin{enumerate}
						\item Let $(\Vhat_1, \Vhat_2)$ be such a cut.
						\item Record $F \leftarrow E(\Vhat_1, \Vhat_2)$.
						\item $G\leftarrow G/\Vhat_1/F.$ \label{ln:ContractedGraph}
					\end{enumerate}
					\item $\Eintersect \leftarrow \Eintersect \cup F$.
				\end{enumerate}
			\end{enumerate}
			\item  For all $2^{|\tset|-1} \leq 8$ partitions of the terminals $\tset$
			\begin{enumerate}
				\item If the corresponding min-cut has size at most $c$,
				add it to $\Eintersect$.
			\end{enumerate}
			\item Return $\Eintersect$.
		\end{enumerate}
	\end{algbox}
	\vspace{-4mm}
	\caption{Recursive algorithm using minimum terminal cuts for finding a set of edges that intersect all terminal cuts of size $\leq c$.}
	\label{fig:GetIntersectingEdgesTerminal}
\end{figure}

\paragraph{Discussion of algorithm in Figure \ref{fig:GetIntersectingEdgesTerminal}.}
We clarify some lines in the algorithm of Figure \ref{fig:GetIntersectingEdgesTerminal}.
If the algorithm finds a non-trivial $\tset$-separating cut as the minimum terminal cut, it returns the result of the recursion in Line~\ref{ln:Return}, and does not execute any of the later lines in the algorithm.
In Line~\ref{ln:Condition}, in addition to checking that the $s$-isolating minimum cut size is still $x$, we also must check that $s$ does not get contracted with another terminal.
Otherwise, contracting across that cut makes global progress by reducing the number of terminals by $1$.
In Line~\ref{ln:ContractedGraph}, note that we can still view $s$ as a terminal in $G\leftarrow G/\Vhat_1/F$, as we have assumed that this contraction does not merge $s$ with any other terminals.

\begin{lemma}
	\label{lem:WeakerVersion}
	For any graph $G$, terminals $\tset$, and a value $c$, the algorithm \textsc{RecursiveTerminalCuts} as shown in Figure~\ref{fig:GetIntersectingEdgesTerminal} runs in $O(n^2c^4\log n)$ time and returns a set at most $O(|\tset|c^2)$ edges that intersect all $(\tset,c)$-cuts.
\end{lemma}

\begin{proof}
	We assume $m \le nc$ throughout, as we can reduce to this case in $O(mc)$ time by Lemma \ref{lem:EdgeReduction}.
	In line \ref{ln:compmincut}, we use Theorem \ref{thm:ColeH03}.
	
	Note that the recursion in Line~\ref{ln:Return} can only branch $O(|\tset|)$ times, by the analysis in Lemma \ref{lem:RecursiveNontrivialCutsCorrectness}. 
	Similarly, the case where $s$ gets contracted with another terminal in Line~\ref{ln:Condition} can only occur $O(|\tset|)$ times.
	
	Therefore, we only create $O(|\tset|)$ distinct terminals throughout the algorithm.
	Let $s$ be a terminal created at some point during the algorithm. By monotonicity of cuts in Lemma \ref{lem:CutMonotone}, the minimum $s$-isolating cut can only increase in size $c$ times.
	Hence, $\Eintersect$ is the union of $O(|\tset|c)$ cuts of size at most $c$. Therefore, $\Eintersect$ has at most $O(|\tset|c^2)$ edges.
	
	To bound the runtime, we use the total number of edges in the graphs in our recursive algorithm as a potential function.
	Thus, initially, this potential function has value $m$.
	Note that the recursion of Line~\ref{ln:Return} can increase the potential function by $c$; hence, the total potential function increase throughout the algorithm is bounded by $m + O(c|\tset|) = O(nc)$.
	
	Each loop of Line~\ref{ln:Condition} decreases our potential function by at least $1$ from contractions.
	Thus, the total runtime of the loop involving Line~\ref{ln:Condition} can be bounded by
	\[
	O\left(mc\right) + O\left(m+nc^3\log n\right)
	=
	O\left(nc^3\log n\right),
	\]
	where the former term is from running a maxflow algorithm up to flow $c$, and the latter is from applying Theorem \ref{thm:ColeH03}.
	As the total increase in the total potential function is at most $O(nc)$, the loop in Line~\ref{ln:Condition} can only execute $O(nc)$ times, for a total runtime of $O(n^2c^4 \log n)$ as desired.
\end{proof}

Our further speedup of this routine in Section~\ref{sec:ExpanderLocalCut} also uses a faster variant of \textsc{RecursiveTerminalCuts} as base case, which happens when $|\tset|$ is too small.
Here the main observation is that a single maxflow computation is sufficient to compute a ``maximal" $s$-isolating minimum cut, instead of the repeated contractions performed in \textsc{RecursiveTerminalCuts}.

\begin{lemma}
	\label{lem:StrongerVersion}
	For any graph $G$, terminals $\tset$, and a value $c$, there is an algorithm that runs in $O(mc + n|\tset|c^4 \log n)$ time and returns a set at most $O(|\tset|c^2)$ edges that intersect all $(\tset,c)$-cuts.
\end{lemma}
\begin{proof}
	We modify \textsc{RecursiveTerminalCuts} as shown in Figure~\ref{fig:GetIntersectingEdgesTerminal} and its analysis as given in Lemma \ref{lem:WeakerVersion} above.
	Specifically, we modify how we compute a maximal $s$-isolating minimum cut in Line~\ref{ln:Condition}.
	For any partition $\tset = \tset_1 \cupdot \tset_2$, by submodularity of cuts it is known that there is a unique maximal subset $V_1 \subseteq V$ such that
	\begin{align*}
	\tset_1 & \subseteq V_1,\\
	\tset_2 & \subseteq V\bs V_1,\\
	\left|E\left(V_1, V_2\right)\right|
	&=
	\left|\Mincut\left(G,\tset_1,\tset_2\right)\right|.
	\end{align*}
	Also, this maximal set can be computed in $O(mc)$ time by running the Ford-Fulkerson augmenting path algorithm with $\tset_2$ as source and $\tset_1$ as sink.
	The connectivity value of $c$ means at most $c$ augmenting paths need to be found, and the set $V_2$ can be set to the vertices that can still reach the sink set $\tset_2$ in the residual graph~\cite{FulkersonH75}. Now set $V_1 = V\bs V_2.$
	Thus by setting $\tset_1 \leftarrow \{s\}$, we can use the corresponding computed set $V_1$ as the representative of the maximal $s$-isolating minimum terminal cut.
	
	Now we analyze the runtime of this procedure.
	First, we reduce the number of edges to at most $nc$ in $O(mc)$ time.
	As in the proof of Lemma \ref{lem:WeakerVersion}, all graphs in the recursion have at most $O(nc)$ edges.
	The recursion in Line~\ref{ln:Return} can only branch $|\tset|$ times, and we only need to compute $O(c|\tset|)$ maximal $s$-isolating minimum terminal cuts throughout the algorithm.
	Each call Theorem \ref{thm:ColeH03} takes $O(m+nc^3\log n) = O(nc^3 \log n)$ time, for a total runtime of $O(nc^3 \log n \cdot c|\tset|) = O(n|\tset|c^4 \log n)$ as desired.
\end{proof}

\subsubsection{Using Local Cut Algorithms}
\label{sec:ExpanderLocalCut}
A local cut algorithm is a tool that has recently been developed. Given a vertex $v$, there exists a local cut algorithm that determines whether there is a cut of size at most $c$ such that the side containing $v$ has volume at most $\nu$ in time linear in $c$ and $\nu$.
\begin{theorem}[Theorem 3.1 of \cite{NanongkaiSY19b:arxiv}]
	\label{thm:LocalFlow}
	Let $G$ be a graph and let $v \in V(G)$ be a vertex.
	For a connectivity parameter $c$ and volume parameter $\nu$, there is an algorithm running in time $\O(c^2\nu)$ that with high probability either
	\bi
		\item[1.] Certifies that there is no cut of size at most $c$ such that the side with $v$ has volume at most $\nu$.
		\item[2.] Returns a cut of size at most $c$ such that the side with $v$ has volume at most $130c\nu.$ 
	\ei
\end{theorem}
We now formalize the notion of the smallest cut that is \emph{local} around a vertex $v$.
\begin{definition}[Local cuts]
	For a vertex $v \in G$, define $\localcut(v)$ to be
	\[ \min_{\substack{V = V_1 \cupdot V_2 \\ v \in V_1 \\ \vol(V_1) \le \vol(V_2)}} |E(V_1,V_2)|.\]
\end{definition}

We now combine Theorem \ref{thm:LocalFlow} with the observation from Equation~\ref{eq:ExpanderImba} in order to control the volume of the smaller side of the cut in an expander.

\begin{lemma}
	\label{lem:ExpanderCutFaster}
	Let $G$ be a graph with conductance at most $\phi$, and let $\tset$ be a set of terminals.
	If $|\tset| \ge 500c^2\phi^{-1}$ then for any vertex $s \in \tset$ we can with high probability in $\O(c^3\phi^{-1})$ time either compute $\localcut(s)$ or certify that $\localcut(s) > c.$
\end{lemma}

\begin{proof}
	We run binary search on the size of the minimum terminal cut with $s$ on the smaller side, and apply Theorem \ref{thm:LocalFlow}.
	The smaller side of a terminal cut has volume at most $c\phi^{-1}$.
	Therefore, if $|\tset| \ge 500c^2\phi^{-1}$, then the cut returned by Theorem \ref{thm:LocalFlow} for $\nu = c\phi^{-1}$ will always be a terminal cut, as $130\nu c \le |\tset|/2$.
	The runtime is $\O(\nu c^2) = \O(c^3 \phi^{-1})$ as desired.
\end{proof}

We can substitute this faster cut-finding procedure into \textsc{RecursiveTerminalCuts} to get the faster running time stated in Theorem~\ref{thm:IntersectingEdges} Part~\ref{part:IntersectingEdges:Faster}.

\begin{proofof}{Theorem~\ref{thm:IntersectingEdges} Part~\ref{part:IntersectingEdges:Faster}}
	First, we perform expander decomposition, remove the inter-cluster edges, and add their endpoints as terminals.
	
	Now, we describe the modifications we need to make to Algorithm \textsc{RecursiveTerminalCuts} as shown in Figure~\ref{fig:GetIntersectingEdgesTerminal}.
	Let $\tsethat$ be the set of terminals at the top level of recursion.
	The recursion gives that at most $O(|\tsethat|)$ distinct terminals are created in the recursion.
	
	First, we terminate if $|\tset| \le 500c^2\phi^{-1}$ and use the result of Lemma \ref{lem:StrongerVersion}.
	Otherwise, instead of using Theorem \ref{thm:ColeH03} for line \ref{ln:compmincut}, we compute the terminal $s \in \tset$ with minimal value of $\localcut(s).$
	This gives us a minimum terminal cut.
	If the corresponding cut is a non-trivial $\tset$-separating cut then we recurse as in Line~\ref{ln:Return}.
	Otherwise, we perform the loop in Line~\ref{ln:Condition}.
	
	We now give implementation details for computing the terminal $s \in \tset$ with minimal value of $\localcut(s)$.
	By Lemma \ref{lem:CutMonotone} we can see that for a terminal $s$, $\localcut(s)$ is monotone throughout the algorithm.
	For each terminal $s$, our algorithm records the previous value of $\localcut(s)$ computed.
	Because this value is monotone, we need only check vertices $s$ whose value of $\localcut(s)$ could still possibly be minimal.
	Now, either $\localcut(s)$ is certified to be minimal among all $s$, or the value of $\localcut(s)$ is higher than the previously recorded value.
	Note that this can only occur $O(c|\tsethat|)$ times, as we stop processing a vertex $s$ if $\localcut(s) > c$.
	
	We now analyze the runtime.
	We first bound the runtime from the cases $|\tset| \le 500c^2\phi^{-1}$.
	The total number of vertices and edges in the leaves of the recursion tree is at most $O(mc)$.
	Therefore, by Lemma \ref{lem:StrongerVersion}, the total runtime from these is
	at most \[ \O(500c^2\phi^{-1} \cdot mc \cdot c^4) = \O(m\phi^{-1}c^7). \]
	
	Now, the loop of Line~\ref{ln:Condition} can only execute $c\phi^{-1}$ times because the volume of any $s$-isolating cut has size at most $c\phi^{-1}$. 
	Each iteration of the loop requires $\O(c^3\phi^{-1})$ time by Lemma \ref{lem:ExpanderCutFaster}.
	Therefore, the total runtime of executing the loop and calls to it is bounded by
	\[ \O\left(c|\tsethat| \cdot c\phi^{-1} \cdot c^3\phi^{-1} \right) = \O(|\tsethat|\phi^{-2}c^5).\]
	Combining these shows Theorem~\ref{thm:IntersectingEdges} Part~\ref{part:IntersectingEdges:Faster}.
\end{proofof}

\section{Applications} 
\label{sec:Applications}

We now discuss the applications of our \sparsifiers in dynamic graph data structures and parameterized algorithms.
\subsection{Dynamic Offline \texorpdfstring{$c$}{c}-edge-connectivity}
\label{sec:dynaconapplication}
In this section we formally show how to use \sparsifier to obtain offline dynamic connectivity routines.

\begin{lemma}
\label{lemma:dynacon}
Suppose that an algorithm $A(G',S,c)$ returns a \sparsifier with $f(c)|S|$ edges for terminals $S$ on a graph $G'$ in time $\O(g(c)|E(G')|).$ Then there is an offline algorithm that on an initially empty graph $G$ answers $q$ edge insertion, deletion, and $c$-connectivity queries in total time $\O(f(c)(g(c)+c)q)$.
\end{lemma}
Lemma \ref{lemma:dynacon} directly implies an analogous result when graph $G$ is not initially empty, as we can make the first $m$ queries simply insert the edges of $G$.

We now state the algorithm \OfflineConnectivity~in Figure \ref{fig:OfflineDynacon} which shows Lemma \ref{lemma:dynacon}. In Figure \ref{fig:OfflineDynacon} graph $G_i$ for $0 \le i \le q$ denotes the current graph after queries $Q_1, \cdots, Q_i$ have been applied.
\begin{figure}[!h]
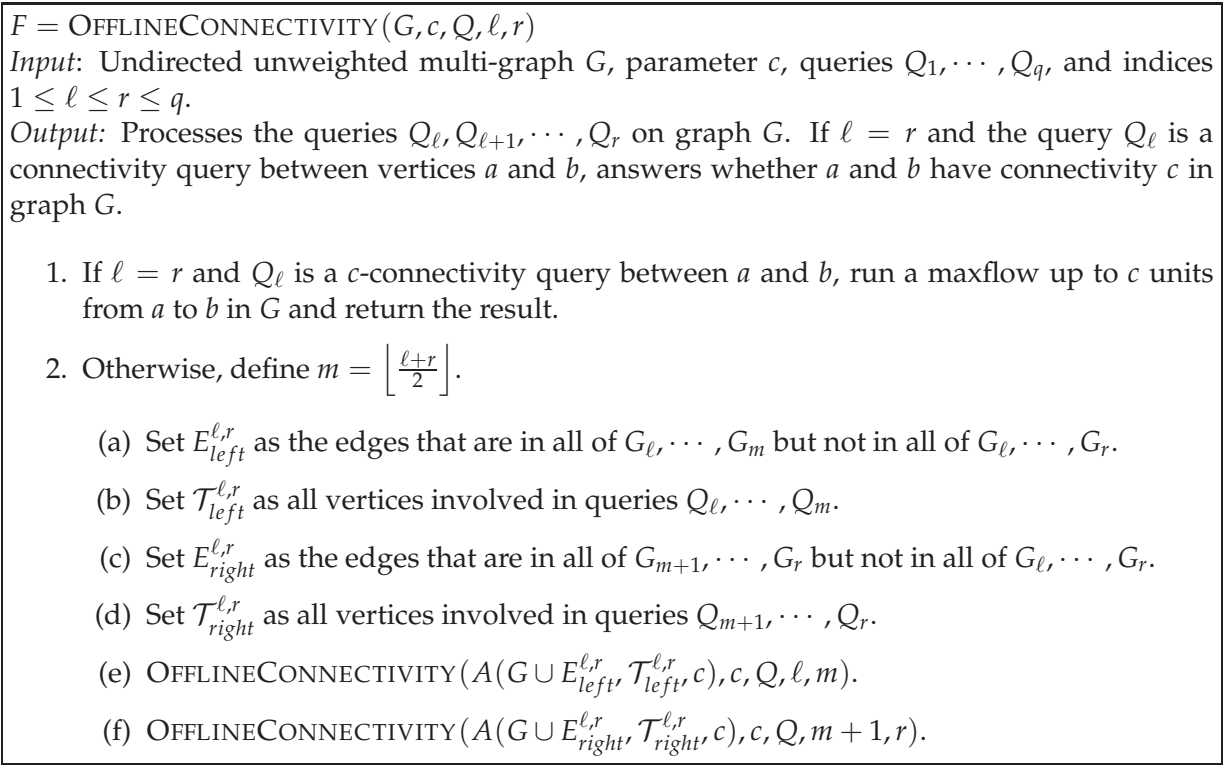

	\begin{algbox}
		$F = \OfflineConnectivity(G, c, Q, \ell, r)$
		
		\emph{Input}: Undirected unweighted multi-graph $G$, parameter $c$, queries $Q_1, \cdots, Q_q$, and indices $1 \le \ell \le r \le q$.
		
		\emph{Output:} Processes the queries $Q_\ell, Q_{\ell+1}, \cdots, Q_r$ on graph $G$. If $\ell = r$ and the query $Q_{\ell}$ is a connectivity query between vertices $a$ and $b$, answers whether $a$ and $b$ have connectivity $c$ in graph $G$.
		
		\begin{enumerate}
			\item If $\ell = r$ and $Q_{\ell}$ is a $c$-connectivity query between $a$ and $b$, run a maxflow up to $c$ units from $a$ to $b$ in $G$ and return the result. \label{line:maxflowfinish}
			\item Otherwise, define $m = \left\lfloor \frac{\ell+r}{2} \right\rfloor$.
			\begin{enumerate}
				\item Set $E^{\ell,r}_{left}$ as the edges that are in all of $G_{\ell}, \cdots, G_m$ but not in all of $G_{\ell}, \cdots, G_r.$ \label{line:left}
				\item Set $\T^{\ell,r}_{left}$ as all vertices involved in queries $Q_\ell, \cdots, Q_m$. \label{line:left2}
				\item Set $E^{\ell,r}_{right}$ as the edges that are in all of $G_{m+1}, \cdots, G_r$ but not in all of $G_{\ell}, \cdots, G_r.$ \label{line:right}
				\item Set $\T^{\ell,r}_{right}$ as all vertices involved in queries $Q_{m+1}, \cdots, Q_r.$ \label{line:right2}
				\item $\OfflineConnectivity(A(G\cup E^{\ell,r}_{left},\T^{\ell,r}_{left},c), c, Q, \ell, m).$
				\item $\OfflineConnectivity(A(G\cup E^{\ell,r}_{right},\T^{\ell,r}_{right},c), c, Q, m+1, r).$
			\end{enumerate}
		\end{enumerate}
	\end{algbox}
	\caption{Algorithm for offline connectivity}
	\label{fig:OfflineDynacon}
\end{figure}
\paragraph{Description of algorithm \OfflineConnectivity.} The algorithm does a divide and conquer procedure, computing \sparsifier on the way down the recursion tree. As the algorithm moves down the recursion tree, it adds edges to our graph that will exist in all children of the recursion node. This is done in line \ref{line:left} and \ref{line:right}. The algorithm then computes all vertices involved in queries in the children of a recursion node in lines \ref{line:left2} and \ref{line:right2}, and computes a \sparsifier treating those vertices as terminals, and recurses.

\begin{proof}
Correctness directly follows from the algorithm description and Lemma \ref{lem:AddEdges} -- we are only adding edges to our \sparsifier as we progress down the recursion tree. It suffices to bound the runtime.

It is straightforward to compute all the sets $E^{\ell,r}_{left}$ and $E^{\ell,r}_{right}$ in $O(q \log q)$ time. Additionally, $\sum_{\ell,r} |E^{\ell,r}_{left}| + |E^{\ell,r}_{right}| \le O(q \log q)$, where the sum runs over all pairs $(\ell, r)$ encountered in an execution of \OfflineConnectivity.

The graph $G$ in a call to $\OfflineConnectivity(G,c,Q,\ell,r)$ has at most $O(f(c)(r-\ell))$ edges by the guarantees of algorithm $A$. Therefore, the runtime of calls to algorithm $A$ is bounded by
\begin{align*} &\O(g(c)) \cdot \sum_{\ell,r} \left((r-\ell)f(c)+|E^{\ell,r}_{left}| + |E^{\ell,r}_{right}|\right) \\ &\le \O(g(c)) \cdot \left(O(q \log q) + \sum_{k=0}^{\log q} 2^k \cdot O(f(c)q \cdot 2^{-k}) \right) \le \O(g(c)f(c)q).
\end{align*}
The cost of running line \ref{line:maxflowfinish} is bounded by $O(c \cdot f(c)q)$ as each graph in the leaf recursion nodes has at most $O(f(c))$ vertices. Hence the total runtime is at most $\O(f(c)(g(c)+c)q)$ as desired.
\end{proof}

Combining Lemma \ref{lemma:dynacon} and Theorem \ref{thm:Main} Part \ref{part:Main2} immediately gives a proof of Theorem \ref{thm:Dynacon}.

Additionally, by adding / deleting edges from source / sinks, we can query for $c$-edge connectivity between multiple sets of vertices efficiently on a static graph.

\begin{corollary}
	Given a graph $G$ with $m$ edges, as well as query subsets $(A_1, B_1), (A_2, B_2) \ldots (A_k, B_k)$, we can compute the value of $\mincut_{G}^{c}(A_i, B_i)$ for all $1 \le i \le k$ in $\O\left((m+\sum_i |A_i|+|B_i|)c^{O(c)}\right)$ time.
\end{corollary}

\subsection{Parameterized Algorithms for Network Design}
\label{sec:sndpapplication}
In this section, we consider the rooted survivable network design problem
(\SNDP), in which we are given a graph $G$ with edge-costs, as well as $h$
demands $(v_i, d_i) \in V \times \Z$, $i \in [h]$, and a root $r \in V$.
The goal is to find a minimum-cost subgraph that contains, for every demand
$(v_i, d_i)$, $i \in [h]$, $d_i$ edge-disjoint paths connecting $r$ to $v_i$.

We will show how to solve \SNDP optimally in the running time of $f(c, \tw(G)) n$, where
$c = \max_i d_i$ is the maximum demand, and $\tw(G)$ is the treewidth of $G$.
Our algorithm uses the ideas of Chalermsook et al.~\cite{ChalermsookDELV18}
together with  \sparsifiers.
Our running time is $n \exp\left(O(c^4\tw(G)\log(\tw(G)c)\right)$ which is double-exponential in $c$, but only single-exponential in $\tw(G)$ (whereas the result by Chalermsook et al.~\cite{ChalermsookDELV18} is double-exponential in both $c$ and $\tw(G)$).

Let $(T,X)$ be a tree decomposition of $G$ (see Section \ref{sec:sndp:prelim} for a definition).
The main idea of our algorithm is to assign, to each $t \in T$, a state
representing the connectivity of the solution restricted to $X_t$.
By assigning these states in a manner that they are consistent across $T$, we
can piece together the solutions by looking at the states for each individual
node.
We will show that representing connectivity by two \sparsifiers is sufficient for our purposes, and that we can achieve
consistency across $T$ by using very simple local rules between the state for
a node $t$ and the states for its children $t_1$, $t_2$.
These rules can be applied using dynamic programming to compute the optimum solution.

\begin{theorem}
\label{thm:sndp}
There is an exact algorithm for \textup{\SNDP} on a graph $G$ with treewidth $\tw(G)$ and maximum demand $c$ with a running time of $n \exp\left(O(c^4\tw(G)\log(\tw(G)c)\right)$.
\end{theorem}

The rest of this section is dedicated to proving the theorem above.
In Section \ref{sec:sndp:prelim} we introduce some concepts and assumptions
used in our result;
in Section \ref{sec:sndp:connrules} we show how to represent the solution
locally using \sparsifiers, and how to make sure that all these local
representations are consistent;
finally, in Section \ref{sec:sndp:algo} we show how to use these ideas to
solve \SNDP.

\subsubsection{Preliminaries}
\label{sec:sndp:prelim}

\paragraph{Tree Decomposition}
Let $G$ be an undirected graph. A \emph{tree decomposition} is a pair
$(T,X)$ where $T$ is a tree and $X = \{X_t \subseteq V(G)\}_{t
\in V(T)}$ is a collection of \emph{bags} such that:
\begin{enumerate}
	\item $V(G) = \bigcup_{t \in V(T)} X_t$, that is, every $v \in V(G)$ is contained in some bag $X_t$; \label{not:def:tw1}
	\item For any edge $uv \in E$, there is a bag $X_t$ that contains both $u$ and $v$, i.e., $u,v \in X_t$; \label{not:def:tw2}
	\item For each vertex $v \in V(G)$, the collection of nodes $t$ whose bags $X_t$ contain $v$ induces a
connected subgraph of $T$, that is, $T[\{t \in V(T): v \in X_t\}]$ is a (connected) subtree.\label{not:def:tw3}
\end{enumerate}
We will use the term \emph{node} to refer to an element $t \in V(T)$, and
\emph{bag} to refer to the corresponding subset $X_t$.

The treewidth of $G$, denoted $\tw(G)$, is the minimum \emph{width} of any
tree decomposition $(T,X)$ for $G$. The width of $(T,X)$ is
given by $\max |X_t|-1$.

Let $G$ be a graph and $(T, X)$ be its tree decomposition.
We will say that each edge $uv \in E(G)$ \emph{belongs} to a unique bag $X_t$,
and write $e \in E_t$ if $t \in T$ is the node closest to the root such
that $u,v \in X_t$.
For a subset $S \subseteq V(T)$, we define $X\paren{S} := \bigcup_{t
\in S} X_t$.
Given a node $t \in V(T)$, we denote by
$T_t$ the subtree of $T$ rooted at $t$ and by $p(t)$ the parent node
of $t$ in $V(T)$.
We also define $G_t$ as the subgraph with vertices $X(T_t)$ and edges
$E(G_t) = \bigcup_{t' \in T_t} E_{t'}$.
For each $v \in V$, we denote by $t_v$ the node closest to the root for
which $v \in X_{t_v}$.

Throughout this section, we will consider a tree decomposition $(T,X)$ of $G$
satisfying the following properties (see~\cite{ChalermsookDELV18}):
\begin{inparaenum}[(i)]
\item $T$ has height $O(\log n)$;
\item $|X_t| \leq O(\tw(G))$ for all $t \in T$;
\item every leaf bag contains no edges ($E_t = \emptyset$ for all leaves $t \in T$); 
\item every non-leaf has exactly $2$ children.
\end{inparaenum}
Additionally, we add the root $r$ to every bag $X_t$, $t \in T$.

\paragraph{Vertex Sparsification}

In our application of \sparsifiers to \SNDP, we need graphs that
preserve the thresholded minimum cuts \emph{for any disjoint sets} $\tset_1,
\tset_2 \subseteq \tset$ (i.e.\ $\tset_1 \cup \tset_2$ may not include all terminals).
Lemma \ref{lem:prelimcuts:disjointsets} shows that this formulation is
equivalent to that of Definition \ref{def:DefMain}.
We write $G \equiv^c_{\tset} H$ if $G$ and $H$ are $(\tset, c)$-equivalent
according to the definition of Lemma \ref{lem:prelimcuts:disjointsets}.

\begin{lemma}
\label{lem:prelimcuts:disjointsets}
	Let $G$, $H$ be graphs both containing a set of terminals $\tset$.
	$G$ and $H$ are $(\tset, c)$-equivalent if and only if
	for any \emph{disjoint} subsets of terminals $\tset_1, \tset_2 \subseteq \tset$,
	the thresholded minimum cuts are preserved in $H$, i.e.,  
	\[
		\mincut^c_H\left(\tset_1, \tset_2\right)
		=
		\mincut^c_G\left(\tset_1, \tset_2\right).
	\]
\end{lemma}

\begin{proof}
Note that if the condition above holds, $G$ and $H$ are trivially $(\tset,
c)$-equivalent, since for any partition $\tset = \tset_1 \cupdot \tset_2$,
$\tset_1$ and $\tset_2$ are disjoint.

We now prove that if Definition \ref{def:DefMain} is satisfied, thresholded
minimum cuts are preserved for any disjoint subsets of terminals.
Let $\tset_1, \tset_2 \subseteq \tset$ be disjoint sets of terminals.
Let $(A_G, B_G)$, $(A_H, B_H)$ be the minimum cuts separating $\tset_1$ and
$\tset_2$ in $G$ and $H$, respectively.
We know that
\[
\mincut^c_G(A_H \cap \tset, B_H \cap \tset) \geq 
\mincut^c_G(A_G \cap \tset, B_G \cap \tset),
\]
since $(A_G, B_G)$ is the minimum cut separating $\tset_1$, $\tset_2$ in $G$,
and $\mincut^c_G(A_H \cap \tset, B_H \cap \tset)$ represents a cut which also
separates $\tset_1$, $\tset_2$. A similar statement is also true for $H$.

Furthermore, we know $\mincut^c_G(A_H \cap \tset, B_H \cap \tset) =
\mincut^c_H(A_H \cap \tset, B_H \cap \tset)$ (and similarly for $(A_G, B_G)$)
by $(\tset_c)$-equivalence of $G$ and $H$.

Combining everything, we get,
\begin{alignat*}{2}
\mincut^c_H(A_H \cap \tset, B_H \cap \tset)
&= \mincut^c_G(A_H \cap \tset, B_H \cap \tset)
&&\geq \mincut^c_G(A_G \cap \tset, B_G \cap \tset) \\
&= \mincut^c_H(A_G \cap \tset, B_G \cap \tset)
&&\geq \mincut^c_H(A_H \cap \tset, B_H \cap \tset)
\end{alignat*}

The circular chain of inequalities implies that 
\begin{align*}
\mincut^c_G(A_G \cap \tset, B_G \cap \tset) 
= \mincut^c_H(A_H \cap \tset, B_H \cap \tset)
\end{align*}

The definition of $(A_G, B_G)$, $(A_H, B_H)$ then implies that 
\begin{align*}
\mincut^c_G(\tset_1, \tset_2) 
= \mincut^c_G(A_G \cap \tset, B_G \cap \tset) 
= \mincut^c_H(A_H \cap \tset, B_H \cap \tset)
= \mincut^c_H(\tset_1, \tset_2)
\end{align*}
\end{proof}

\subsubsection{Local Connectivity Rules}
\label{sec:sndp:connrules}

In this section, we will introduce the local connectivity rules which will
allow us to assign states in a consistent manner to the nodes of $T$.
The states we will consider consist of two  \sparsifiers
roughly corresponding to the connectivity of the solution in $E(G_t)$ and $E
\setminus E(G_t)$.
We then present some rules that make these states consistent across $T$,
while only being enforced for a node and its children.

We remark that this notation deviates from the one used by Chalermsook et
al.~\cite{ChalermsookDELV18}, in which states represent connectivity in
$E(G_t)$ and $E$.
We do so because taking the union of overlapping \sparsifiers would lead
to overcounting of the number of edge-disjoint paths.

The following local definition of connectivity introduces the desired
consistency rules that we can use to define a dynamic program for the problem.
Lemma \ref{def:kgst:weakmimnets:unify} shows that a collection of \sparsifiers
satisfy the local definition if and only if they represent the connectivity in
$G$ with terminals given by a bag.

\begin{definition}[Local Connectivity]
\label{def:kgst:weakmimnets:unify}

We say that the pairs of \sparsifiers $\{(\Gnet_t, \Dnet_t)\}_{t\in V(T)}$ satisfy the
\emph{local connectivity definition} if
\begin{align*}
\Gnet_t &\equiv^c_{X_t} (X_t, \emptyset)     & \text{for every leaf node $t$ of $T$} \\
\Dnet_{\rootn(T)} &\equiv^c_{X_t} (X_t, \emptyset)
\end{align*}
and for every internal node $t \in V(T)$ with children $t_1$ and $t_2$,
\begin{align*}
\Gnet_t &\equiv^c_{X_t} (X_t, E_t) \cup \Gnet_{t_1} \cup \Gnet_{t_2} \\
\Dnet_{t_1} &\equiv^c_{X_t} (X_t, E_t) \cup \Gnet_{t_2} \cup \Dnet_{t} 
\end{align*}
where $A \equiv^c_{X_t} B$ means that $\mincut^c_{A}(S_1, S_2) = \mincut^c_{B}
(S_1, S_2)$ for all disjoint sets $S_1, S_2 \subseteq X_t$.
\end{definition}

\begin{lemma}
\label{lem:kgst:weakmimnets:unify}

Let $G=(V,E)$ be a graph, and $(\Tcal, X)$ its tree decomposition satisfying
[the usual properties].
For every $t \in V(T)$, let $(\Gnet_t, \Dnet_t)$ be a
pair as in Definition \ref{def:kgst:weakmimnets:unify}.

Then, the pairs $\set{(\Gnet_t, \Dnet_t)}_{t \in T}$ satisfy the local
definitions if and only if for every $t \in V(\Tcal)$,
\begin{align*}
\Gnet_t &\equiv^c_{X_t} G_t \\
\Dnet_t &\equiv^c_{X_t} G \setminus E(G_t)
\end{align*}
where $A \equiv^c_{X_t} B$ means that $\mincut^c_{A}(S_1, S_2) = \mincut^c_{B}
(S_1, S_2)$ for all disjoint sets $S_1, S_2 \subseteq X_t$.

\end{lemma}
\begin{proof}
We start by proving the statement for $\Gnet$ by bottom-up induction, and then
the one for $\Dnet$ by top-down induction.
We will show that 

Let $t \in \Tcal$ be a leaf of the tree decomposition. Then $E(G_t) =
\emptyset$, so the statement immediately follows. Consider now an
internal node $t$ with children $t_1$, $t_2$, and assume
that the claim follows for $t_1$, $t_2$.
We will define $H'_t = (X_t, E_t) \cup \Gnet_{t_1} \cup \Gnet_{t_2}$, and prove that
$H'_t \equiv^c_{X_t} G_t$.
That implies that $\Gnet_t \equiv^c_{X_t} H'_t$ (that is, $\Gnet_t$ satisfies
the local connectivity definition) if and only if $\Gnet_t \equiv^c_{X_t}
G_t$.

Let $S_1, S_2 \subseteq X_t$, and $F$ be the cutset for a mincut between $S_1$
and $S_2$ in $E(G_t)$.
We will use $c_G(S_1,S_2) = \mincut^c_{G}(S_1, S_2)$ for
conciseness (in this proof only).
Then
\begin{align*}
c_{G_t}&(S_1, S_2) \\
&= \min(c, |F|) \\
&= \min(c, |F \cap E_t| + |F \cap E(G_{t_1})| + |F \cap E(G_{t_2})|) \\
&\geq \min\parenbig{c, c_{E_t}(S_1, S_2) + c_{E(G_{t_1})}(S_1 \cap X_{t_1}, S_2 \cap X_{t_1}}
			+ c_{E(G_{t_2})}(S_1 \cap X_{t_2}, S_2 \cap X_{t_2})) \\
&= \min\parenbig{c, c_{E_t}(S_1, S_2) + c_{\Gnet_{t_1}}(S_1 \cap X_{t_1}, S_2 \cap X_{t_1}} 
			+ c_{\Gnet_{t_2}}(S_1 \cap X_{t_2}, S_2 \cap X_{t_2})) \\
&\geq c_{H'_t}(S_1, S_2)
\end{align*}

The third inequality follows because each of the three terms corresponds to a
min-cut between $S_1$ and $S_2$ for the respective edge sets.
The fourth inequality follows by induction hypothesis, and the final one
follows by definition of $H'_t$.
For this last step, we crucially use that $X_{t_1} \cap X_{t_2} \subseteq
X_t$, which means that any cut for $E_t$, $\Gnet_{t_1}$ and $\Gnet_{t_2}$ uses
disjoint edges and disjoint vertices outside of $X_t$. These edges provide an
upper bound for the cut $c_{H'_t}$.

Analogously, we can prove that $c_{G_t} \leq c_{H'_t}$, by
taking a set of edges $F'$ of $H'_t$ that realizes the minimum cut in that
graph. The same steps then apply to prove the desired inequality. 
\begin{align*}
c_{H'_t}(S_1, S_2) 
&= \min(c, |F'|) \\
&= \min(c, |F' \cap E_t| + |F' \cap E(\Gnet_{t_1})| + |F' \cap E(\Gnet_{t_2})|) \\
&\geq \min\parenbig{c, c_{E_t}(S_1, S_2) + c_{\Gnet_{t_1}}(S_1 \cap X_{t_1}, S_2 \cap X_{t_1}}
			+ c_{\Gnet_{t_2}}(S_1 \cap X_{t_2}, S_2 \cap X_{t_2})) \\
&= \min\parenbig{c, c_{E_t}(S_1, S_2) + c_{G_{t_1}}(S_1 \cap X_{t_1}, S_2 \cap X_{t_1}} 
			+ c_{G_{t_2}}(S_1 \cap X_{t_2}, S_2 \cap X_{t_2})) \\
&\geq c_{G_t}(S_1, S_2)
\end{align*}
This concludes the first part of the proof.

For the second part of the proof, we will use top-down induction.
For $t=r$, notice that $E\setminus E(G_t) = \emptyset$, so the statement
follows.
We now prove the equality for a node $t_1$ with parent $t$ and sibling $t_2$.
Let  $H_{t_1} = (X_t, E_t) \cup \Gnet_{t_2} \cup \Dnet_{t}$, and prove that $H_t \equiv^c_{X_t} G \setminus G_t$.
This implies the statement, as it shows that $\Dnet_t \equiv^c_{X_t} H_t$
(that is, $\Dnet_t$ satisfies the local connectivity definition) if and only
if $\Dnet_t \equiv^c_{X_t} G \setminus G_t$.

\newcommand{\twolineeq}[3]{
	\begin{aligned}[t]
	#1 & #2 \\
	   & #3
	\end{aligned}
}

Let $S_1, S_2 \subseteq X_{t_1}$, and $F$ be the cutset for a
mincut between $S_1$ and $S_2$ in $E\setminus E(G_{t_1})$. Then
\begin{align*}
c_{E\setminus E(G_{t_1})}(S_1, S_2) 
&= \min(c, |F|) \\
&= \min(c, |F \cap E_t| + |F \cap (E\setminus E(G_t))| + |F \cap E(G_{t_2})|) \\
&\geq 
\twolineeq{\min(c, c_{E_t}(S_1 \cap X_t, S_2 \cap X_t)}{
+ c_{E\setminus E(G_t)}(S_1 \cap X_t, S_2 \cap X_t)}{
+ c_{E(G_{t_2})}(S_1 \cap X_{t_2}, S_2 \cap X_{t_2}))} \\
&= 
\twolineeq{\min(c, c_{E_t}(S_1 \cap X_t, S_2 \cap X_t)}{
+ c_{\Dnet_t}(S_1 \cap X_t, S_2 \cap X_t)}{
+ c_{\Gnet_{t_2}}(S_1 \cap X_{t_2}, S_2 \cap X_{t_2}))} \\
&\geq c_{H_{t_1}}(S_1, S_2)
\end{align*}

Similarly to the proof above, we use the fact that $F \cap E_t$, $F \cap
(E\setminus E(G_t))$, $F \cap E(G_{t_2})$ are cuts in the subgraphs
$E_t$, $G \setminus E(G_t)$, $E(G_{t_2})$ respectively. The last step
follows from the fact that the three terms correspond to cuts in $E_t$, $\Dnet_t$
and $\Gnet_{t_2}$, and therefore their union forms a cut in $\Dnet_t \cup E_t \cup
\Gnet_{t_2}$. Since $H_{t_1} \equiv^c_{X_{t_1}} \Dnet_t \cup E_t \cup \Gnet_{t_2}$,
the inequality follows. The converse follows similarly:
\begin{align*}
c_{H_{t_1}}(S_1, S_2)
&= \min(c, |F|) \\
&= \min(c, |F \cap E_t| + |F \cap E(\Dnet_t)| + |F \cap E(\Gnet_{t_2})|) \\
&\geq 
\twolineeq{\min(c, c_{E_t}(S_1 \cap X_t, S_2 \cap X_t)}{ 
  + c_{\Dnet_t}(S_1 \cap X_t, S_2 \cap X_t)}{
  + c_{\Gnet_{t_2}}(S_1 \cap X_{t_2}, S_2 \cap X_{t_2}))} \\
&= 
\twolineeq{ \min(c, c_{E_t}(S_1 \cap X_t, S_2 \cap X_t)}{
  + c_{E\setminus E(G_t)}(S_1 \cap X_t, S_2 \cap X_t)}{
  + c_{E(G_{t_2})}(S_1 \cap X_{t_2}, S_2 \cap X_{t_2}))} \\
&\geq c_{E\setminus E(G_{t_1})}(S_1, S_2) 
\end{align*}
This completes the proof.
\end{proof}

\subsubsection{Dynamic Program for \SNDP}
\label{sec:sndp:algo}

In this section, we present an algorithm for \SNDP on bounded-treewidth
graphs, which uses dynamic programming to compute a solution bottom-up.
Our goal is to assign two \sparsifiers $\Gnet_t$, $\Dnet_t$ to each node
$t \in T$, corresponding to the connectivity of the solution in $E(G_t)$
and $E\setminus E(G_t)$.
We argue that any solution for $G_t$, $t \in T$ that is compatible with a
state $(\Gnet_t, \Dnet_t)$ can be interchangeably used, which implies that the
dynamic program will obtain the minimum-cost solution.

We define a dynamic programming table $D$, with entries $D[t, \Gnet, \Dnet]$,
$t \in T$, $\Gnet$, $\Dnet$ a  \sparsifiers with
terminal set $X_t$.
The entry $D[t, \Gnet, \Dnet]$ represents the minimum cost of a solution $F$
that is consistent with $\Gnet$ (i.e.\ $F \equiv^c_{X_t} \Gnet$), such that $F
\cup \Dnet_t$ satisfies all the demands contained in $G_t$.

We compute $D[t, \Gnet, \Dnet]$ as follows:
\begin{itemize}
	\item For any leaf $t$, set $D[t, \emptyset, \Dnet] = 0$ and $D[t, \Gnet,
	\Dnet] = +\infty$ for $\Gnet \neq \emptyset$;
	\item For the root node $\rootn(T)$, set $D[\rootn(T), \Gnet, \Dnet] = +\infty$ if $\Dnet
	\neq \emptyset$;
	\item For any demand $(v_i, d_i)$, and $t \in T$ such that $v_i \in X_t$, set $D[t, \Gnet, \Dnet] = +\infty$ if $\Gnet \cup \Dnet$ contains fewer than $d_i$ edge-disjoint paths connecting $r$ to $v_i$.
\end{itemize}
For all other entries of $T$, compute it recursively as:
\begin{align*}
D[t, \Gnet, \Dnet] 
= \min \Big\{ w(Y) + D[t_1, \Gnet_1, \Dnet_1] &+ D[t_2, \Gnet_2, \Dnet_2] :
Y \subseteq E_t, \\
&\Gnet \equiv^c_{X_t} Y \cup \Gnet_1 \cup \Gnet_2, \\
&\Dnet_1 \equiv^c_{X_{t_1}} Y \cup \Dnet \cup \Gnet_2, \\
&\Dnet_2 \equiv^c_{X_{t_2}} Y \cup \Dnet \cup \Gnet_1
\Big\}
\end{align*}

We now want to prove that the dynamic program is feasible, i.~e.~that the
entries $D[\rootn(T), \Gnet, \emptyset]$ correspond to feasible solutions;
and that it is optimal, meaning that we will obtain the optimum solution to
the problem.

To prove that the dynamic program is feasible, notice that, by definition, any
solution obtained induces a choice of $Y_t$, $\Gnet_t$, $\Dnet_t$ for each $t
\in T$. Let $Y = \cup_{t \in T}$.
The recursion formula of the dynamic program implies that the pairs
$\set{(\Gnet_t, \Dnet_t)}_{t \in T}$ satisfy the local connectivity
definition with regard to the graph $(V, Y)$.

By Lemma \ref{lem:kgst:weakmimnets:unify}, this implies that 
\begin{align*}
\Gnet_t \equiv^c_{X_t} G_t[Y], 
\Dnet_t \equiv^c_{X_t} G[Y] \setminus E(G_t), 
\end{align*}
and hence, $\Gnet_t \cup \Dnet_t \equiv^c_{X_t} G[Y]$.

Let $(v_i, d_i)$ be a demand and $t \in T$ be a node such that $v_i \in
X_t$.
Since we know that $\Gnet_t \cup \Dnet_t$ contains $d_i$ edge-disjoint paths
from $r$ to $v_i$ (otherwise $D[t, \Gnet_t, \Dnet_t] = + \infty$), then we
know that the minimum cut separating $r$ from $v_i$ has at least $d_i$ edges,
which implies that $Y$ must also contain $d_i$ edge-disjoint paths connecting
$r$ and $v_i$.

For the converse, we will prove that any feasible solution $F$ can be captured
by the dynamic program.
Given $F$, it is sufficient to take $\Gnet_t$, $\Dnet_t$ to be
 \sparsifiers for $G_t[F]$, $G[F] \setminus E(G_t)$,
respectively.
By Lemma \ref{lem:kgst:weakmimnets:unify} (applied to graph $(V,F)$), we know
that $\set{(\Gnet_t, \Dnet_t)}_{t \in T}$ satisfy the local connectivity
definition for $(V,F)$, and therefore $D[t,\Gnet_t,\Dnet_t]$ can be computed
recursively from $D[t_1, \Gnet_{t_1}, \Dnet_{t_1}]$, $D[t_2, \Gnet_{t_2},
\Dnet_{t_2}]$, $Y_t = F \cap E_t$.

Let $(v_i, d_i)$ be a demand and $t \in T$ be a node such that $v_i \in
X_t$.
Since $F$ is a feasible solution, it contains $d_i$ edge-disjoint paths from
$r$ to $v_i$, and therefore $\mincut^c_F(\set{r}, \set{v_i}) \geq d_i$. This
implies that $\mincut^c_{\Gnet_t \cup \Dnet_t}(\set{r}, \set{v_i}) \geq d_i$,
and thus $\Gnet_t \cup \Dnet_t$ contains $d_i$ edge-disjoint paths from $r$ to
$v_i$ (and is a valid entry of $T$).

We conclude that the dynamic program above computes an optimum solution for
\SNDP.
By Theorem \ref{thm:Main}, there is a  \sparsifier containing
$O(wc^4)$ edges (and $O(wc^4)$ vertices as well), for any graph with $w$ terminals.
Hence, there are at most
\[
	\paren{|V|^2}^{|E|} = \parenbig{w^2c^8}^{wc^4} = \exp\left(O(c^4 w \log(wc))\right)
\]
possibilities for such \sparsifiers.
This implies that the dynamic programming table has $n \exp\left(O(c^4\tw(G)\log(\tw(G)c)\right)$ entries. 
The work required to compute the value of each entry takes time 
\[
\exp\left(O(c^4\tw(G)\log(\tw(G)c)\right) \cdot w^{O(c)} \cdot \poly(c,w)
\]
(considering all combinations of states for children nodes, all disjoint
subsets of terminals, compute the min-cuts and check if they match).

We conclude that the running time of the algorithm is $n \exp\left(O(c^4\tw(G)\log(\tw(G)c)\right)$, which completes the proof.

\paragraph{Acknowledgements:} Parinya Chalermsook has been supported by European Research Council (ERC) under the European Union's Horizon 2020 research and innovation programme (grant agreement No 759557) and by the Academy of Finland Research Fellowship, under the grant number 310415.
Richard Peng is partially supported by the
US National Science Foundation under grant number 1846218. Yang P. Liu has been supported by the Department of Defense (DoD) through the National Defense Science and Engineering Graduate Fellowship (NDSEG) Program.
Bundit Laekhanukit has been supported by the 1000-talents award by the Chinese government and by Science and Technology Innovation 2030 -- ``New Generation of Artificial Intelligence'' Major Project No.(2018AAA0100903), NSFC grant 61932002, Program for Innovative Research Team of Shanghai University of Finance and Economics (IRTSHUFE) and the Fundamental Research Funds for the Central Universities. 
Daniel Vaz has been supported by the Alexander von Humboldt Foundation with funds from the German Federal Ministry of Education and Research (BMBF). 

\bibliographystyle{alpha}
\bibliography{ref}

\newcommand{\etalchar}[1]{$^{#1}$}
\begin{thebibliography}{ADK{\etalchar{+}}16}

\bibitem[ADK{\etalchar{+}}16]{AbrahamDKKP16}
I.~{Abraham}, D.~{Durfee}, I.~{Koutis}, S.~{Krinninger}, and R.~{Peng}.
\newblock On fully dynamic graph sparsifiers.
\newblock In {\em 2016 IEEE 57th Annual Symposium on Foundations of Computer
  Science (FOCS)}, pages 335--344, Oct 2016.

\bibitem[AGK14]{andoni2014towards}
Alexandr Andoni, Anupam Gupta, and Robert Krauthgamer.
\newblock Towards (1+ $\varepsilon$)-approximate flow sparsifiers.
\newblock In {\em Proceedings of the twenty-fifth annual ACM-SIAM symposium on
  Discrete algorithms}, pages 279--293. Society for Industrial and Applied
  Mathematics, 2014.

\bibitem[AKLT15]{AssadiKYV15}
Sepehr Assadi, Sanjeev Khanna, Yang Li, and Val Tannen.
\newblock Dynamic sketching for graph optimization problems with applications
  to cut-preserving sketches.
\newblock In {\em FSTTCS}, 2015.
\newblock Available at \url{https://arxiv.org/abs/1510.03252}.

\bibitem[AKT19]{AbboudKT19:arxiv}
Amir Abboud, Robert Krauthgamer, and Ohad Trabelsi.
\newblock New algorithms and lower bounds for all-pairs max-flow in undirected
  graphs.
\newblock {\em CoRR}, abs/1901.01412, 2019.
\newblock Available at: http://arxiv.org/abs/1901.01412.

\bibitem[AW14]{AbboudW14}
Amir Abboud and Virginia~Vassilevska Williams.
\newblock Popular conjectures imply strong lower bounds for dynamic problems.
\newblock In {\em 55th {IEEE} Annual Symposium on Foundations of Computer
  Science, {FOCS} 2014, Philadelphia, PA, USA, October 18-21, 2014}, pages
  434--443, 2014.

\bibitem[AWY15]{AbboudWY15}
Amir Abboud, Virginia~Vassilevska Williams, and Huacheng Yu.
\newblock Matching triangles and basing hardness on an extremely popular
  conjecture.
\newblock In {\em Proceedings of the Forty-Seventh Annual {ACM} on Symposium on
  Theory of Computing, {STOC} 2015}, pages 41--50, 2015.

\bibitem[BHKP08]{BhalgatHKP08}
Anand Bhalgat, Ramesh Hariharan, Telikepalli Kavitha, and Debmalya Panigrahi.
\newblock Fast edge splitting and edmonds' arborescence construction for
  unweighted graphs.
\newblock In {\em Proceedings of the nineteenth annual ACM-SIAM symposium on
  Discrete algorithms}, pages 455--464. Society for Industrial and Applied
  Mathematics, 2008.

\bibitem[BK96]{BenczurK96}
Andr\'{a}s~A. Bencz\'{u}r and David~R. Karger.
\newblock Approximating s-t minimum cuts in $\tilde{O}(n^2)$ time.
\newblock In {\em Proceedings of the twenty-eighth annual ACM symposium on
  Theory of computing}, STOC '96, pages 47--55, New York, NY, USA, 1996. ACM.

\bibitem[CDE{\etalchar{+}}18]{ChalermsookDELV18}
Parinya Chalermsook, Syamantak Das, Guy Even, Bundit Laekhanukit, and Daniel
  Vaz.
\newblock Survivable network design for group connectivity in low-treewidth
  graphs.
\newblock In Eric Blais, Klaus Jansen, Jos{\'{e}} D.~P. Rolim, and David
  Steurer, editors, {\em Approximation, Randomization, and Combinatorial
  Optimization. Algorithms and Techniques, {APPROX/RANDOM} 2018, August 20-22,
  2018 - Princeton, NJ, {USA}}, volume 116 of {\em LIPIcs}, pages 8:1--8:19.
  Schloss Dagstuhl - Leibniz-Zentrum fuer Informatik, 2018.

\bibitem[CFK{\etalchar{+}}15]{CyganFKLMPPS15}
Marek Cygan, Fedor~V. Fomin, Lukasz Kowalik, Daniel Lokshtanov, D{\'{a}}niel
  Marx, Marcin Pilipczuk, Michal Pilipczuk, and Saket Saurabh.
\newblock {\em Parameterized Algorithms}.
\newblock Springer, 2015.

\bibitem[CH03]{ColeH03}
Richard Cole and Ramesh Hariharan.
\newblock A fast algorithm for computing steiner edge connectivity.
\newblock In {\em Proceedings of the 35th Annual {ACM} Symposium on Theory of
  Computing, June 9-11, 2003, San Diego, CA, {USA}}, pages 167--176, 2003.

\bibitem[Che18]{Chechik18}
Shiri Chechik.
\newblock Near-optimal approximate decremental all pairs shortest paths.
\newblock In {\em 59th {IEEE} Annual Symposium on Foundations of Computer
  Science, {FOCS} 2018, Paris, France, October 7-9, 2018}, pages 170--181,
  2018.

\bibitem[Chu12]{Chuzhoy12}
Julia Chuzhoy.
\newblock On vertex sparsifiers with steiner nodes.
\newblock In Howard~J. Karloff and Toniann Pitassi, editors, {\em Proceedings
  of the 44th Symposium on Theory of Computing Conference, {STOC} 2012, New
  York, NY, USA, May 19 - 22, 2012}, pages 673--688. {ACM}, 2012.

\bibitem[CLLM10]{CharikarLLM10}
Moses Charikar, Tom Leighton, Shi Li, and Ankur Moitra.
\newblock Vertex sparsifiers and abstract rounding algorithms.
\newblock In {\em 51th Annual {IEEE} Symposium on Foundations of Computer
  Science, {FOCS} 2010, October 23-26, 2010, Las Vegas, Nevada, {USA}}, pages
  265--274, 2010.

\bibitem[CMSV17]{CohenMSV17}
Michael~B. Cohen, Aleksander Madry, Piotr Sankowski, and Adrian Vladu.
\newblock Negative-weight shortest paths and unit capacity minimum cost flow in
  $\widetilde{{O}}(m^{10/7} \log{W})$ time (extended abstract).
\newblock In {\em {SODA}}, pages 752--771. {SIAM}, 2017.

\bibitem[CSWZ00]{chaudhuri2000computing}
Shiva Chaudhuri, KV~Subrahmanyam, Frank Wagner, and Christos~D Zaroliagis.
\newblock Computing mimicking networks.
\newblock {\em Algorithmica}, 26(1):31--49, 2000.

\bibitem[DGGP19]{DurfeeGGP19}
David Durfee, Yu~Gao, Gramoz Goranci, and Richard Peng.
\newblock Fully dynamic spectral vertex sparsifiers and applications.
\newblock In {\em Proceedings of the 51st Annual {ACM} {SIGACT} Symposium on
  Theory of Computing, {STOC} 2019, Phoenix, AZ, USA, June 23-26, 2019.}, pages
  914--925, 2019.
\newblock Available at: \url{https://arxiv.org/abs/1906.10530}.

\bibitem[DKP{\etalchar{+}}17]{DurfeeKPRS17}
David Durfee, Rasmus Kyng, John Peebles, Anup~B Rao, and Sushant Sachdeva.
\newblock Sampling random spanning trees faster than matrix multiplication.
\newblock In {\em Proceedings of the 49th Annual ACM SIGACT Symposium on Theory
  of Computing}, pages 730--742. ACM, 2017.
\newblock Available at: \url{https://arxiv.org/abs/1611.07451}.

\bibitem[DV94]{DinitzV94}
Yefim Dinitz and Alek Vainshtein.
\newblock The connectivity carcass of a vertex subset in a graph and its
  incremental maintenance.
\newblock In {\em Proceedings of the twenty-sixth annual ACM symposium on
  Theory of computing}, pages 716--725. ACM, 1994.

\bibitem[DV95]{DinitzV95}
Ye. Dinitz and A.~Vainshtein.
\newblock Locally orientable graphs, cell structures, and a new algorithm for
  the incremental maintenance of connectivity carcasses.
\newblock In {\em Proceedings of the Sixth Annual ACM-SIAM Symposium on
  Discrete Algorithms}, SODA '95, pages 302--311, 1995.

\bibitem[DW98]{DinitzW98}
Yefim Dinitz and Jeffery~R. Westbrook.
\newblock Maintaining the classes of 4-edge-connectivity in a graph on-line.
\newblock {\em Algorithmica}, 20(3):242--276, 1998.

\bibitem[EGK{\etalchar{+}}10]{EGKRTT10}
Matthias Englert, Anupam Gupta, Robert Krauthgamer, Harald R{\"{a}}cke, Inbal
  Talgam{-}Cohen, and Kunal Talwar.
\newblock Vertex sparsifiers: New results from old techniques.
\newblock {\em CoRR}, abs/1006.4586, 2010.

\bibitem[EGK{\etalchar{+}}14]{EnglertGKRTT14}
Matthias Englert, Anupam Gupta, Robert Krauthgamer, Harald Racke, Inbal
  Talgam{-}Cohen, and Kunal Talwar.
\newblock Vertex sparsifiers: New results from old techniques.
\newblock {\em {SIAM} J. Comput.}, 43(4):1239--1262, 2014.

\bibitem[Epp94]{Eppstein94}
David Eppstein.
\newblock Offline algorithms for dynamic minimum spanning tree problems.
\newblock {\em J. Algorithms}, 17(2):237--250, September 1994.

\bibitem[FH75]{FulkersonH75}
Delbert~Ray Fulkerson and Gary Harding.
\newblock On edge-disjoint branchings.
\newblock Technical report, CORNELL UNIV ITHACA NY DEPT OF OPERATIONS RESEARCH,
  1975.

\bibitem[FHKQ16]{FafianieHKQ16}
Stefan Fafianie, Eva{-}Maria~C. Hols, Stefan Kratsch, and Vuong~Anh Quyen.
\newblock Preprocessing under uncertainty: Matroid intersection.
\newblock In {\em 41st International Symposium on Mathematical Foundations of
  Computer Science, {MFCS} 2016, August 22-26, 2016 - Krak{\'{o}}w, Poland},
  pages 35:1--35:14, 2016.
\newblock Available at \url{https://core.ac.uk/download/pdf/62922404.pdf}.

\bibitem[FKQ16]{FafianieKQ16}
Stefan Fafianie, Stefan Kratsch, and Vuong~Anh Quyen.
\newblock Preprocessing under uncertainty.
\newblock In {\em 33rd Symposium on Theoretical Aspects of Computer Science
  (STACS 2016)}, volume~47, pages 33:1--33:13, 2016.

\bibitem[FLPS16]{fomin2016efficient}
Fedor~V Fomin, Daniel Lokshtanov, Fahad Panolan, and Saket Saurabh.
\newblock Efficient computation of representative families with applications in
  parameterized and exact algorithms.
\newblock {\em Journal of the ACM (JACM)}, 63(4):1--60, 2016.

\bibitem[FY19]{ForsterY19:arxiv}
Sebastian Forster and Liu Yang.
\newblock A faster local algorithm for detecting bounded-size cuts with
  applications to higher-connectivity problems.
\newblock {\em CoRR}, abs/1904.08382, 2019.

\bibitem[GH61]{GomoryH61}
R.~E. Gomory and T.~C. Hu.
\newblock Multi-terminal network flows.
\newblock {\em Journal of the Society for Industrial and Applied Mathematics},
  9(4):pp. 551--570, 1961.

\bibitem[GHP17a]{GoranciHP17b}
Gramoz Goranci, Monika Henzinger, and Pan Peng.
\newblock Improved guarantees for vertex sparsification in planar graphs.
\newblock In {\em 25th Annual European Symposium on Algorithms, {ESA} 2017,
  September 4-6, 2017, Vienna, Austria}, pages 44:1--44:14, 2017.
\newblock Available at: \url{https://arxiv.org/abs/1702.01136}.

\bibitem[GHP17b]{goranci2017improved}
Gramoz Goranci, Monika Henzinger, and Pan Peng.
\newblock Improved guarantees for vertex sparsification in planar graphs.
\newblock In {\em 25th Annual European Symposium on Algorithms (ESA 2017)}.
  Schloss Dagstuhl-Leibniz-Zentrum fuer Informatik, 2017.

\bibitem[GHP17c]{GoranciHP17a}
Gramoz Goranci, Monika Henzinger, and Pan Peng.
\newblock The power of vertex sparsifiers in dynamic graph algorithms.
\newblock In {\em 25th Annual European Symposium on Algorithms, {ESA} 2017,
  September 4-6, 2017, Vienna, Austria}, pages 45:1--45:14, 2017.
\newblock Available at: \url{https://arxiv.org/abs/1712.06473}.

\bibitem[GHP18]{GoranciHP18}
Gramoz Goranci, Monika Henzinger, and Pan Peng.
\newblock Dynamic effective resistances and approximate schur complement on
  separable graphs.
\newblock In {\em 26th Annual European Symposium on Algorithms, {ESA} 2018,
  August 20-22, 2018, Helsinki, Finland}, pages 40:1--40:15, 2018.
\newblock Available at: \url{https://arxiv.org/abs/1802.09111}.

\bibitem[GHT16]{GoranciHT16}
Gramoz Goranci, Monika Henzinger, and Mikkel Thorup.
\newblock Incremental exact min-cut in poly-logarithmic amortized update time.
\newblock In {\em LIPIcs-Leibniz International Proceedings in Informatics},
  volume~57. Schloss Dagstuhl-Leibniz-Zentrum fuer Informatik, 2016.

\bibitem[GR16]{GoranciR16}
Gramoz Goranci and Harald R{\"{a}}cke.
\newblock Vertex sparsification in trees.
\newblock In {\em Approximation and Online Algorithms - 14th International
  Workshop, {WAOA} 2016, Aarhus, Denmark, August 25-26, 2016, Revised Selected
  Papers}, pages 103--115, 2016.
\newblock Available at: \url{https://arxiv.org/abs/1612.03017}.

\bibitem[GT14]{GoldbergT14}
Andrew~V. Goldberg and Robert~Endre Tarjan.
\newblock Efficient maximum flow algorithms.
\newblock {\em Commun. {ACM}}, 57(8):82--89, 2014.

\bibitem[HdLT01]{HolmLT01}
Jacob Holm, Kristian de~Lichtenberg, and Mikkel Thorup.
\newblock Poly-logarithmic deterministic fully-dynamic algorithms for
  connectivity, minimum spanning tree, 2-edge, and biconnectivity.
\newblock {\em J. {ACM}}, 48(4):723--760, 2001.

\bibitem[HK99]{HenzingerK99}
Monika~Rauch Henzinger and Valerie King.
\newblock Randomized fully dynamic graph algorithms with polylogarithmic time
  per operation.
\newblock {\em J. {ACM}}, 46(4):502--516, 1999.

\bibitem[HKNR98]{hagerup1998characterizing}
Torben Hagerup, Jyrki Katajainen, Naomi Nishimura, and Prabhakar Ragde.
\newblock Characterizing multiterminal flow networks and computing flows in
  networks of small treewidth.
\newblock {\em Journal of Computer and System Sciences}, 57(3):366--375, 1998.

\bibitem[HKP07]{HariharanKP07}
Ramesh Hariharan, Telikepalli Kavitha, and Debmalya Panigrahi.
\newblock Efficient algorithms for computing all low st edge connectivities and
  related problems.
\newblock In {\em Proceedings of the eighteenth annual ACM-SIAM symposium on
  Discrete algorithms}, pages 127--136. Society for Industrial and Applied
  Mathematics, 2007.

\bibitem[JS20]{JS20}
Wenyu Jin and Xiaorui Sun.
\newblock Fully dynamic c-edge connectivity in subpolynomial time.
\newblock {\em CoRR}, abs/2004.07650, 2020.

\bibitem[Kar00]{Karger00}
David~R. Karger.
\newblock Minimum cuts in near-linear time.
\newblock {\em J. ACM}, 47(1):46--76, January 2000.

\bibitem[KLOS14]{KelnerLOS14}
Jonathan~A. Kelner, Yin~Tat Lee, Lorenzo Orecchia, and Aaron Sidford.
\newblock An almost-linear-time algorithm for approximate max flow in
  undirected graphs, and its multicommodity generalizations.
\newblock In {\em Proceedings of the Twenty-Fifth Annual {ACM-SIAM} Symposium
  on Discrete Algorithms, {SODA} 2014, Portland, Oregon, USA, January 5-7,
  2014}, pages 217--226, 2014.

\bibitem[KLP{\etalchar{+}}16]{KyngLPSS16}
Rasmus Kyng, Yin~Tat Lee, Richard Peng, Sushant Sachdeva, and Daniel~A
  Spielman.
\newblock Sparsified cholesky and multigrid solvers for connection laplacians.
\newblock In {\em Proceedings of the 48th Annual ACM SIGACT Symposium on Theory
  of Computing}, pages 842--850. ACM, 2016.
\newblock Available at http://arxiv.org/abs/1512.01892.

\bibitem[KPZP18]{karpov2018exponential}
Nikolai Karpov, Marcin Pilipczuk, and Anna Zych-Pawlewicz.
\newblock An exponential lower bound for cut sparsifiers in planar graphs.
\newblock {\em Algorithmica}, pages 1--14, 2018.

\bibitem[KR13]{krauthgamer2013mimicking}
Robert Krauthgamer and Inbal Rika.
\newblock Mimicking networks and succinct representations of terminal cuts.
\newblock In {\em Proceedings of the twenty-fourth annual ACM-SIAM symposium on
  Discrete algorithms}, pages 1789--1799. SIAM, 2013.

\bibitem[KR14]{khan2014mimicking}
Arindam Khan and Prasad Raghavendra.
\newblock On mimicking networks representing minimum terminal cuts.
\newblock {\em Information Processing Letters}, 114(7):365--371, 2014.

\bibitem[KR17]{KrauthgamerR17}
Robert Krauthgamer and Inbal Rika.
\newblock Refined vertex sparsifiers of planar graphs.
\newblock {\em CoRR}, abs/1702.05951, 2017.

\bibitem[KS16]{KyngS16}
Rasmus Kyng and Sushant Sachdeva.
\newblock Approximate gaussian elimination for laplacians - fast, sparse, and
  simple.
\newblock In {\em {IEEE} 57th Annual Symposium on Foundations of Computer
  Science, {FOCS} 2016, 9-11 October 2016, Hyatt Regency, New Brunswick, New
  Jersey, {USA}}, pages 573--582, 2016.
\newblock Available at http://arxiv.org/abs/1605.02353.

\bibitem[KW12]{Kratsch12}
Stefan Kratsch and Magnus Wahlstrom.
\newblock Representative sets and irrelevant vertices: New tools for
  kernelization.
\newblock In {\em Proceedings of the 2012 IEEE 53rd Annual Symposium on
  Foundations of Computer Science}, FOCS '12, pages 450--459, 2012.
\newblock Available at \url{https://arxiv.org/abs/1111.2195}.

\bibitem[LM10]{leighton2010extensions}
F~Thomson Leighton and Ankur Moitra.
\newblock Extensions and limits to vertex sparsification.
\newblock In {\em Proceedings of the forty-second ACM symposium on Theory of
  computing}, pages 47--56. ACM, 2010.

\bibitem[Lov77]{lovasz1977flats}
L{\'a}szl{\'o} Lov{\'a}sz.
\newblock Flats in matroids and geometric graphs.
\newblock In {\em Combinatorial Surveys (Proc. 6th British Combinatorial
  Conference}, pages 45--86, 1977.

\bibitem[Mad13]{Madry13}
Aleksander Madry.
\newblock Navigating central path with electrical flows: From flows to
  matchings, and back.
\newblock In {\em Foundations of Computer Science (FOCS), 2013 IEEE 54th Annual
  Symposium on}, pages 253--262. IEEE, 2013.
\newblock Available at http://arxiv.org/abs/1307.2205.

\bibitem[Mad16]{Madry16}
Aleksander Madry.
\newblock Computing maximum flow with augmenting electrical flows.
\newblock In {\em {FOCS}}, pages 593--602. {IEEE} Computer Society, 2016.

\bibitem[Mar06]{Marx06}
D{\'{a}}niel Marx.
\newblock Parameterized graph separation problems.
\newblock {\em Theor. Comput. Sci.}, 351(3):394--406, 2006.

\bibitem[Mar09]{marx2009parameterized}
D{\'a}niel Marx.
\newblock A parameterized view on matroid optimization problems.
\newblock {\em Theoretical Computer Science}, 410(44):4471--4479, 2009.

\bibitem[MM10]{MakarychevM10}
Konstantin Makarychev and Yury Makarychev.
\newblock Metric extension operators, vertex sparsifiers and lipschitz
  extendability.
\newblock In {\em 51th Annual {IEEE} Symposium on Foundations of Computer
  Science, {FOCS} 2010, October 23-26, 2010, Las Vegas, Nevada, {USA}}, pages
  255--264, 2010.

\bibitem[Moi09]{moitra2009approximation}
Ankur Moitra.
\newblock Approximation algorithms for multicommodity-type problems with
  guarantees independent of the graph size.
\newblock In {\em 2009 50th Annual IEEE Symposium on Foundations of Computer
  Science}, pages 3--12. IEEE, 2009.

\bibitem[MS18]{MolinaS18:unpublished}
Antonio Molina and Bryce Sandlund.
\newblock Historical optimization with applications to dynamic higher edge
  connectivity.
\newblock 2018.

\bibitem[NI92]{NagamochiI92}
Hiroshi Nagamochi and Toshihide Ibaraki.
\newblock A linear-time algorithm for finding a sparse $k$-connected spanning
  subgraph of a $k$-connected graph.
\newblock {\em Algorithmica}, 7(1-6):583--596, 1992.

\bibitem[NSY19a]{NanongkaiSY19a}
Danupon Nanongkai, Thatchaphol Saranurak, and Sorrachai Yingchareonthawornchai.
\newblock Breaking quadratic time for small vertex connectivity and an
  approximation scheme.
\newblock In {\em Proceedings of the 51st Annual {ACM} {SIGACT} Symposium on
  Theory of Computing, {STOC} 2019, Phoenix, AZ, USA, June 23-26, 2019.}, pages
  241--252, 2019.
\newblock Available at: \url{https://arxiv.org/abs/1904.04453}.

\bibitem[NSY19b]{NanongkaiSY19b:arxiv}
Danupon Nanongkai, Thatchaphol Saranurak, and Sorrachai Yingchareonthawornchai.
\newblock Computing and testing small vertex connectivity in near-linear time
  and queries.
\newblock {\em CoRR}, abs/1905.05329, 2019.
\newblock Available at : \url{http://arxiv.org/abs/1905.05329}.

\bibitem[Pen16]{Peng16}
Richard Peng.
\newblock Approximate undirected maximum flows in ${O}(m \poly\log(n))$ time.
\newblock In {\em Proceedings of the Twenty-Seventh Annual {ACM-SIAM} Symposium
  on Discrete Algorithms, {SODA} 2016, Arlington, VA, USA, January 10-12,
  2016}, pages 1862--1867, 2016.

\bibitem[PSS19]{PengSS19}
Richard Peng, Bryce Sandlund, and Daniel~Dominic Sleator.
\newblock Optimal offline dynamic 2, 3-edge/vertex connectivity.
\newblock In {\em Algorithms and Data Structures - 16th International
  Symposium, {WADS} 2019, Edmonton, AB, Canada, August 5-7, 2019, Proceedings},
  pages 553--565, 2019.
\newblock Available at: \url{https://arxiv.org/abs/1708.03812}.

\bibitem[Ree97]{reed1997tree}
Bruce~A Reed.
\newblock Tree width and tangles: A new connectivity measure and some
  applications.
\newblock {\em Surveys in combinatorics}, pages 87--162, 1997.

\bibitem[RST14]{RackeST14}
Harald R{\"{a}}cke, Chintan Shah, and Hanjo T{\"{a}}ubig.
\newblock Computing cut-based hierarchical decompositions in almost linear
  time.
\newblock In Chandra Chekuri, editor, {\em Proceedings of the Twenty-Fifth
  Annual {ACM-SIAM} Symposium on Discrete Algorithms, {SODA} 2014, Portland,
  Oregon, USA, January 5-7, 2014}, pages 227--238. {SIAM}, 2014.

\bibitem[She13]{Sherman13}
Jonah Sherman.
\newblock Nearly maximum flows in nearly linear time.
\newblock In {\em {FOCS}}, pages 263--269. {IEEE} Computer Society, 2013.

\bibitem[She17]{Sherman17}
Jonah Sherman.
\newblock Area-convexity, l\({}_{\mbox{{\(\infty\)}}}\) regularization, and
  undirected multicommodity flow.
\newblock In {\em Proceedings of the 49th Annual {ACM} {SIGACT} Symposium on
  Theory of Computing, {STOC} 2017, Montreal, QC, Canada, June 19-23, 2017},
  pages 452--460, 2017.

\bibitem[ST04]{SpielmanT04}
Daniel~A. Spielman and Shang{-}Hua Teng.
\newblock Nearly-linear time algorithms for graph partitioning, graph
  sparsification, and solving linear systems.
\newblock In L{\'{a}}szl{\'{o}} Babai, editor, {\em Proceedings of the 36th
  Annual {ACM} Symposium on Theory of Computing, Chicago, IL, USA, June 13-16,
  2004}, pages 81--90. {ACM}, 2004.

\bibitem[SW19]{SaranurakW19}
Thatchaphol Saranurak and Di~Wang.
\newblock Expander decomposition and pruning: Faster, stronger, and simpler.
\newblock In {\em {SODA}}, pages 2616--2635. {SIAM}, 2019.
\newblock Available at: \url{https://arxiv.org/abs/1812.08958}.

\bibitem[vdBS19]{vdBrandS19:arxiv}
Jan van~den Brand and Thatchaphol Saranurak.
\newblock Sensitive distance and reachability oracles for large batch updates.
\newblock {\em CoRR}, abs/1907.07982, 2019.

\end{thebibliography}

\begin{appendix}
\section{Deferred Proofs}
\label{sec:proofs}

\subsection{Proof of Lemma \ref{lem:EdgeReduction}}
\label{proofs:EdgeReduction}
	Consider the following routine: repeat $c$ iterations of
	finding a maximal
	spanning forest from $G$, remove it from $G$ and add it to $H$.
	
	Each of the steps takes $O(m)$ time, for a total of $O(mc)$.
	Also, a maximal spanning tree has the property that for every non-empty
	cut, it contains at least one edge from it.
	Thus, for any cut $\partial(S)$, the $c$ iterations add at least
	\[{}
	\min\left\{c, \abs{\partial\left( S \right)} \right\}
	\]
	edges to $H$, which means that up to a value of $c$, all cuts
	in $G$ and $H$ are the same.

\subsection{Proof of Lemma \ref{lem:ContractWellLinked}}
\label{proofs:ContractWellLinked}

Let $G' = G/X$ be the contracted graph and $v_{X}$ be the contracted vertex in $G'$ that is obtained by contracting $G[X]$. Since we do not contract the terminals, it suffices to show that, for any two subsets $\Xset_A, \Xset_B \subseteq \tset$, we have $\mincut^c_{G'}(\Xset_A, \Xset_B) = \mincut^c_G(\Xset_A, \Xset_B)$.

Starting with $\mincut^c_{G'}(\Xset_A, \Xset_B) \geq \mincut^c_G(\Xset_A,
\Xset_B)$, we can see that all the edges in $G'$ are also in $G$, which
implies that any cutset in $G'$ is also in $G$. We conclude that the size of
the minimum cut in $G$ must be at most the size of the minimum cut in $G'$,
for any pair of terminals sets. In general, we can say that contraction of
edges only ever increases connectivity, which implies the above.

Let us now show the converse, that is, $\mincut^c_{G'}(\Xset_A, \Xset_B) \leq
\mincut^c_G(\Xset_A, \Xset_B)$.
Since we are in the unweighted setting, it is sufficient to consider 
$|\Xset_A|, |\Xset_B| \leq c$.
Suppose that $\mincut^c_{G'}(\Xset_A, \Xset_B) = \ell \leq c$.  
Then there must be $\ell$ disjoint paths connecting $X'_A \subseteq X_A$ to $X'_B \subseteq X_B$ such that $|X'_A| = |X'_B| = \ell$. Denote the set of such paths in $G'$ by $\pset'$. 

We will construct the set of edge-disjoint paths $\pset$ in $G$ connecting $X'_A$ to $X'_B$, thus implying that $\mincut^c_G(X_A, X_B) \geq \ell$. 
We write $\pset'$ as $\pset' = \pset'_1 \cup \pset'_2$ where $\pset'_1$ are the paths that do not go through the contracted vertex $v_{X}$. 
We add the paths in $\pset'_1$ to $\pset$, since they correspond to edge disjoint paths in the original graph $G$.  
For paths in $\pset'_2$, we will need to specify their behavior inside the contracted set $G[X]$. Let $E_{in} \subseteq \partial(X)$ be the set of boundary edges of $X$ that paths in $\pset'_2$ use to enter $v_{X}$; analogously, we define $E_{out} \subseteq \partial(X)$. Notice that $|E_{in}| = |E_{out}| = |\pset'_2| \leq c$. 
Since $X$ is connectivity-$c$ well-linked, there is a collection of disjoint paths $\pset_{X}$ connecting $E_{in}$ to $E_{out}$.  
We stitch the three parts of the paths in $\pset'_2$ and $\pset_{X}$ together to add to $\pset$: (1) a subpath of some path $P \in \pset'_2$ from a node in $X'_A$ to $E_{in}$, (2) a path in $\pset_{X}$ from that edge in $E_{in}$ to an edge in $E_{out}$, and (3) a subpath of some path $Q \in \pset'_2$ from the same an edge in $E_{out}$ to a node in $X'_B$.
We remark that, even though $\pset$ contains $\ell$ edge-disjoint paths connecting $X'_A$ to $X'_B$, the pairing induced by $\pset$ and $\pset'$ may be different.

\subsection{Proof of Theorem~\ref{thm:Main} Part~\ref{part:Main2}}
\label{proofs:Main2}
	\newcommand{\cintersect}{C_{\mathit{int}}}
	
	By Lemma \ref{lem:IntersectToContain}, Algorithm \ref{fig:GetContainingEdges} computes a set $\Econtain$ of edges that contains all $(\tset,c)$-cuts.
	Reduce to $m \le nc$ by Lemma \ref{lem:EdgeReduction}.
	Let $\cintersect$ be a constant such that part \ref{part:IntersectingEdges:Faster} of Theorem \ref{thm:IntersectingEdges} gives us a set $\Eintersect$ of edges intersecting all $(\tset,c)$-cuts of size at most $\cintersect(\phi m \log^4{n} + |\tset|)c^2$ in $\O(m \phi^{-2} c^7)$ time.
	Let $\Econtain_i$ be the set of edges $\Econtain$ after the iteration $\chat = i$.
	Let $\tsethat$ be the terminals at the start of the algorithm.
	We show by induction that before processing $\chat = i$ in the second line of Figure \ref{fig:GetContainingEdges} that
	\[
	\left|\Econtain_i\right|
	\le
	\left(4 \cintersect\right)^{c-i}
	\frac{(c!)^2}{(i!)^2}
	\left(\phi m \log^4 n + \left|\tsethat\right|\right)
	\]
	and
	\[
	\left|V(\Econtain_i)\right|
	\le
	2\left(4 \cintersect\right)^{c-i}\frac{(c!)^2}{(i!)^2}
	\left(\phi m \log^4 n + \left|\tsethat\right|\right).
	\]
	
	Since $|V(\Ehat)| \le 2|\Ehat|$ for any set of edges $\Ehat$,
	it suffices to bound $|\Econtain_i|$.
	The induction hypothesis holds for $i = c$.
	By Part~\ref{part:IntersectingEdges:Faster} of Theorem~\ref{thm:IntersectingEdges}
	we have the size of $\Econtain$ after processing $\chat = i$
	is at most
	\begin{align*}
	&\cintersect
	\left(\phi m \log^4{n} + |V(\Econtain_i)|\right)\chat^2 + |\Econtain_i|
	\\ &\le
	\cintersect
	\left(\phi m \log^4{n} + 2\left(4\cintersect\right)^{c-i}\frac{(c!)^2}{(i!)^2}
	\left(\phi m \log^4 n + \left|\tsethat\right| \right)\right)i^2 + \left(4 \cintersect\right)^{c-i}
	\frac{(c!)^2}{(i!)^2}
	\left(\phi m \log^4 n + \left|\tsethat\right|\right) \\
	&\le
	\left(4 \cintersect\right)^{c-i+1}
	\frac{(c!)^2}{\left(i-1\right)!^2}
	\left(\phi m \log^4 n + \left|\tsethat\right|\right)
	\end{align*}
	as desired.
	Taking $i = 0$ shows that the final size of $\Econtain$ is at most
	$(4 \cintersect)^c (c!)^2(\phi m \log^4 n + |\tsethat|)$.
	
	Then, we use the choice of conductance threshold
	\[
	\phi = \frac{1}{5c\left(4 \cintersect\right)^c (c!)^2\log^4n}.
	\]
	Because $m \le nc$, the final size of $\Econtain$ is at most
	\[
	\left(4 \cintersect\right)^c c!\left(\phi m \log^4 n + \left|\tsethat\right|\right)
	\le \frac{n}{5} + \left(4 \cintersect \right)^c c! \left|\tsethat\right|.
	\]
	Now, we apply Lemma \ref{lem:EdgesToSparsifier} to produce a graph $H$ with at most $\frac{n}{5} + (4\cintersect)^c (c!)^2 |\tsethat| + 1$ vertices that is $(\tsethat,c)$-equivalent to $G$.
	Now, we can repeat the process on $H$ $O(\log n)$ times.
	The number of vertices in the graphs we process decrease geometrically until they have at most $2(4 \cintersect)^c (c!)^2 |\tsethat| = O(c)^{2c} |\tsethat|$ many vertices.
	
	Now, combining the runtime of $\O(m\phi^{-2}c^7)$ along with our choice of $\phi$ above
	gives a vertex sparsifier with $|\tsethat| \cdot O(c)^{2c}$ edges in time $O(m \cdot c^{O(c)} \cdot \log^{O(1)}n)$, as desired.

\section{Efficiently Finding a Violating Cut}  
\label{sec:CutFinding}

Although our proof in Section \ref{sec:MainSection} of existence of \sparsifiers with \optimal edges uses the concept of a violating cut, we do not explicitly find the violating cuts. In this section, we present a parametrized algorithm running in time $2^{O(c^2)}k^2m$ for finding violating cuts.

Let $G=(V,E)$ be a graph and $\tset \subseteq V$ a set of terminals,
and let $X = V \setminus \tset$ be the set of non-terminal vertices in $G$. %
For simplicity, we will assume that our terminals are in one-to-one
correspondence with $\partial(X) = E_G(X, V(G)-X)$, that is, that all edges in
$\partial(X)$ have different endpoints outside $X$. %
By abuse of notation, we write $\tset = \partial(X)$ and $k = |\partial(X)|$.
Furthermore, this assumption implies that all terminals have degree $1$.

Observe first that a violating cut can be found in $k^{O(c)} \O(m)$ time by
simply computing all possible minimum cuts separating any 
disjoint subsets of terminals $\tset_0, \tset_1 \subseteq \tset$
of size $q \leq c$ whose minimum cut contains less than $q$ edges. %
However, as we are aiming for a running time of $f(c) \poly(k,m)$, we cannot afford
to enumerate all the possible minimum cuts to find the ``correct'' disjoint subsets $\tset_0,
\tset_1 \subseteq \tset$.

Our algorithm actually solves a more general problem. 
We say that a cut $(A_0, A_1)$ of $G$ is a valid $(Q_0,Q_1,c_0,c_1,\ell)$-constrained cut if 

\begin{itemize}
	
	\item $Q_0 \subseteq A_0 \setminus \tset$ and $Q_1 \subseteq A_1 \setminus \tset$.
	
	\item $|A_j \cap \tset| \geq c_j$ for $j = 0,1$.  
	
	\item $E_G(A_0,A_1)$ contains at most $\ell$ edges. 
\end{itemize}

In other words, $Q_0$ and $Q_1$ are the non-terminals that are ``constrained''
to be on different sides. The values of $c_0$ and $c_1$ are the minimum
required number of terminals on the sides of $A_0$ and $A_1$ respectively. %
We will refer to the two sides of the cuts as \emph{zero side} and \emph{one
side}, respectively. %

\begin{observation}
	Given a subroutine that finds a valid $(Q_0,Q_1,c_0,c_1,\ell)$-constrained cut in time  given by some function $T(m,k,\max(c_0, c_1, \ell))$,
	there exists an algorithm that either returns a violating cut in $G$ or reports that such a cut does not exist in time $O(c T(m,k,c))$.
\end{observation}

In the rest of the section, we shall describe an algorithm that finds a valid $(Q_0,Q_1,c_0,c_1,\ell)$-constrained cut. 
Let $c = \max (c_0,c_1,\ell)$. 
Our algorithm has two steps, encapsulated in the following two lemmas. 

\begin{lemma}[Reduction]
	\label{lem:reduction to base} 
	There is an algorithm that runs in time $2^{O(c^2)}\cdot k^2 \cdot m$, and reduces the problem of finding a valid $(Q_0,Q_1,c_0, c_1,\ell)$-constrained cut to at most $2^{O(c^2)}$ instances of finding valid $(Q'_0,Q'_1,c'_0,c'_1,\ell')$-constrained cut where $\min(c'_0,c'_1) = 0$. 
\end{lemma}
We remark that each of these instances may have different parameters (of the constrained cut). The only property they have in common is that $\min(c'_0,c'_1) = 0$; that is, there is only a one-sided terminal requirement.  

\begin{lemma} [Base case]
	\label{lem:algo for base} 
	For $\ell \leq c$, 
	there is an algorithm that finds a valid $(Q_0,Q_1,0,c,\ell)$-constrained cut (and analogously, $(Q_0,Q_1,c,0,\ell)$-constrained) in time $2^{O(c^2)}\cdot k^2 \cdot m$. 
\end{lemma}

The following theorem follows in a straightforward manner since every violating
cut is also $(\emptyset, \emptyset, \ell+1, \ell+1, \ell)$-constrained, for some
$\ell \in [c-1]$.

\begin{theorem}
	There is an algorithm that runs in time $2^{O(c^2)}\cdot k^2 \cdot m$ and either returns a violating cut or reports that such a cut does not exist. 
\end{theorem}

\subsection{The reduction to the base case} 
In this subsection, we prove Lemma~\ref{lem:reduction to base}. 
The main ingredient for doing so is the following lemma. 

\begin{lemma}
	\label{lem: one step} 
	There is a reduction from $(Q_0,Q_1,c_0,c_1,\ell)$-constrained cut to solving at most $2^{O(c)}$ instances of finding valid $(Q'_0,Q'_1,c'_0,c'_1, \ell')$-constrained cut where $(c'_0 +c'_1) < (c_0 +c_1)$. 
\end{lemma}

In other words, this lemma allows us to reduce the number of required terminals on at least one of the sides by one. 
Applying Lemma~\ref{lem: one step} recursively will allow us to turn an input instance of $(Q_0, Q_1, c_0,c_1,\ell)$ constrained cut into at most $2^{O(c^2)}$ instances of the base problem: This follows from the fact that at every recursive call, the value of at least one of $c_0$ and $c_1$ decreases by at least one. Therefore, the depth of the recursion is at most $2c$, and the ``degree'' of the recursion tree is at most $2^{O(c)}$ as guaranteed by the above lemma.

Let $(G,\tset)$ be an input. We now proceed to prove Lemma~\ref{lem: one step}, i.e., we show how to compute a $(Q_0,Q_1, c_0,c_1, \ell)$-constrained cut in $(G,\tset)$.

\paragraph{Our algorithm:}Let $(A'_0, A'_1)$ be a minimum cut in $G$ such that $Q_0 \subseteq A'_0$ and $Q_1 \subseteq A'_1$ and each side contains at least one terminal.  This cut can be found by a standard minimum $s$-$t$ cut algorithm. Observe that the value of this cut is at most $\ell$ if there is a valid constrained cut.  

Such a cut can be used for our recursive approach to solve smaller subproblems by recursing on $G[A'_i]$ as follows. 
Denote by $\tset_i = A'_i \cap \tset$ for $i=0,1$. 
By definition, each set $\tset_i$ is non-empty, and this is crucial for us. 

If $|E_G(A'_0, A'_1)| > \ell$, the procedure terminates and reports no valid solution. Or, if $|\tset_i| \geq k_i$ for all $i=0,1$, then we have found our desired constrained cut. 
Otherwise, assume that $|\tset_0| < k_0$ (the other case is symmetric). We create a collection of $2^{O(c)}$ instances of smaller subproblems as follows.

\begin{center}
	\framebox[\textwidth]{\begin{minipage}{0.9\textwidth}
			\noindent {\bf Sub-Instances.}
			\begin{itemize}
				\item First, we guess the ``correct'' way to partition terminals in $\tset_0$ into $\tset_0 = \tset_0^0 \cup \tset_0^1$. 
				There are at most $2^c$ possible guesses. 
				
				\item Second, we guess the ``correct'' partition of the (non-terminal) boundary vertices in $V(E_G(A'_0, A'_1)) - \tset$ into $B_0 \cup B_1$ where $B_0$ and $B_1$ are the vertices supposed to be on the zero-side and one-side respectively. 
				Let $\tilde E = E_G(B_0, B_1)$. 
				There are $2^{c}$ possible guesses. 
			\end{itemize}
			
	\end{minipage}
	}    
\end{center}

Now we will solve subproblems in $G[A'_0]$ and $G[A'_1]$. Notice that $G[A'_0]$ has small number of terminals, so we could solve it by brute force. For $G[A'_1]$, we will solve it recursively. 

Let $E_0$ be the minimum cut in $G[A'_0]$ that separates $S_0 = Q_0 \cup (B_0 \cap A'_0) \cup \tset_0^0$ and $T_0 = (B_1 \cap A'_0) \cup \tset_0^1$.   
Next, we solve an instance of valid $(Q'_0, Q'_1, c'_0, c'_1, \ell')$-constrained cut  in $G[A'_1]$ with terminal set $\tset_1$, where $Q'_0 = (B_0 \cap A'_1)$,  $Q'_1 = Q_1\cup (B_1 \cap A'_1)$,  $c'_0 = \max(c_0 - |\tset_0^0|, 0)$, $c'_1 = \max(c_1- |\tset_0^1|,0)$, and $\ell' = \ell - |\tilde E| - |E_0|$. Let $E_1$ be a $(Q'_0, Q'_1, c'_0, c'_1, \ell')$-constrained cut. Our algorithm outputs $E_0 \cup E_1 \cup \tilde{E}$.

\paragraph{Analysis.} 
Clearly, $c'_0 +c'_1 < c_0 +c_1$. 
The following lemma will finish the proof.

\begin{lemma}
	There is a $(Q_0,Q_1,c_0, c_1, \ell)$-constrained cut in $(G,\tset)$ if and only if there exist correct guesses $(B_0,B_1, \tset_0^0, \tset_0^1)$ such that a $(Q'_0, Q'_1,c'_0,c'_1, \ell')$-constrained cut exists in $(G[A'_1], \tset_1)$. 
\end{lemma}

\begin{proof}
	First, we prove  the ``if'' part. Suppose that there exists such a guess $(B_0,B_1, \tset_0^0, \tset_0^1)$. 
	We claim that $E_0 \cup E_1 \cup \tilde{E}$ is actually a $(Q_0,Q_1,c_0,c_1,\ell)$-constrained cut that we are looking for. Observe that the size of the cut is at most $\ell$. 
	
	We argue that there are two subsets of terminals $\widetilde{\tset}_0$ of size $c_0$  and $\widetilde{\tset}_1$ of size $c_1$ that are separated after removing $E_0 \cup E_1 \cup \tilde{E}$. 
	Let $\tset_1^0$ and $\tset_1^1$ be the sets of terminals in $\tset_1$ that are on the side of $Q'_0$ and $Q'_1$, respectively (in particular, $\tset_1^0$ cannot reach $Q'_1$  in $G[A'_1]$ after removing $E_1$). 
	Notice that $| \tset_0^0 \cup \tset_1^0| \geq c_0$ and $|\tset_0^1 \cup \tset_1^1| \geq c_1$. The following claim completes the proof of the ``if'' part. 
\end{proof}
\begin{claim}
	$Q_0 \cup \tset_0^0 \cup \tset_1^0$ and $Q_1 \cup \tset_0^1 \cup \tset_1^1$ are not connected in $G$ after removing $\tilde{E} \cup E_0 \cup E_1$. 
\end{claim}
\begin{proof}
	Let us consider a path $P$ from $Q_0$ to $Q_1$ in $G$; we view it such that the first vertex starts in $Q_0$ and so on until the last vertex on the path is in $Q_1$.  
	Let $u$ be the last vertex the path from the start lies completely in $G[A'_0]$ and $v$ be the first vertex such that the path from $v$ to the end lies completely in $G[A'_1]$. Break path $P$ into $P_1 P_2 P_3$ where $P_1$ is the path from the first vertex to $u$, $P_2$ is the path from $u$ to $v$, and $P_3$ the path from $v$ to the last vertex of $P$ in $Q_1$. 
	If $|\{u,v\} \cap B_0|= 1$, then we are done because $P_2$ contains some edge in $\tilde{E}$. 
	So, it must be that  (i) $u,v \in B_0$ or (ii) $u,v \in B_1$. In case (i), we have $v \in Q'_0$ while the last vertex of $P$ is in $Q_1 \subseteq Q'_1$, so path $P_3$ is path in $G[A'_1]$ connecting $Q'_0$ to $Q'_1$. Hence,  $P_3$ contains an edge in $E_1$. 
	In case (ii), we have that $u \in T_0$, while the first vertex in $P$ is in $Q_0 \subseteq S_0$. Therefore, path $P_1$ is a path in $G[A'_0]$ connecting $S_0$ to $T_0$, which must be cut by $E_0$. 
	
	Similar analysis can be done when considering the path $P$ that connects $Q_0$ and $\tset_1^1$, or between $\tset_0^0$ and $Q_1 \cup \tset_1^1$. 
	The only (somewhat) different case is when the path $P$ connects $Q_0$ to $\tset^1_0$. Assume that $P$ is not completely contained in $G[A'_0]$ (otherwise, it would be trivial). 
	Let $u$ be the last vertex on $P$ such that the path from the start to $u$ lies completely inside $G[A'_0]$, and let $v$ be the first vertex on $P$ such that the path from $v$ to the end of $P$ lies completely inside $G[A'_0]$. 
	Again, we break $P$ into three subpaths $P_1 P_2 P_3$ similarly to before. If $u \in B_1$, then we are done because $P_1$ would contain an edge in $E_0$; or, if $v \in B_0$, then we are also done since $P_3$ would contain an edge in $E_0$.
	Therefore, $u \in B_0$ and $v \in B_1$, implying that $P_2$ must contain an edge in $\tilde{E}$. 
\end{proof}

To prove the ``only if'' part, assume that $(A_0,A_1)$ is a valid $(Q_0,Q_1, c_0,c_1, \ell)$-constrained cut. 
We argue that there is a choice of guess such that the subproblem also finds a valid $(Q'_0,Q'_1,c'_0,c'_1,\ell')$-constrained cut. 
We define $B_i = V(E_G(A_0,A_1)) \cap A_i$ for $i=0,1$, and $\tset_0^i = \tset_0 \cap A_i$ for $i=0,1$. With these choices, we have determined the values of $Q'_0$, $Q'_1$, $c'_0$ and $c'_1$. 
The following claim will finish the proof. 

\begin{claim}
	There exists a cut $E_0$ that separates $S_0$ and $T_0$ in $G[A'_0]$ and a cut $E_1$ that is a  $(Q'_0, Q'_1, c'_0, c'_1, \ell')$-constrained cut. 
\end{claim}

\begin{proof}
	First, we remark that $|E_G(A_0, A_1)| \leq \ell$ and 
	\[E_G(A_0, A_1) = E_G(B_0, B_1) \cup E_G(A'_0 \cap A_0, A'_0 \cap A_1) \cup E_G(A'_1 \cap A_0, A'_1 \cap A_1)\] 
	To complete the proof of the claim, it suffices to show that  $E_G(A'_0 \cap A_0, A'_0 \cap A_1)$ is an $(S_0,T_0)$ cut in $G[A'_0]$ and that $E_G(A'_1 \cap A_0, A'_1 \cap A_1)$ is a valid constrained cut in $G[A'_1]$. 
	
	The first claim is simple: Since $S_0 \subseteq A_0$ and $T_0 \subseteq A_1$, any path from $S_0$ to $T_0$ in $G[A'_0]$ must contain an edge in $E_G(A'_0 \cap A_0, A'_0 \cap A_1)$. 
	
	The second claim is also simple: (i) $Q'_0 \subseteq A_0$ and $Q'_1 \subseteq A_1$, so the edge set $E_G(A'_1 \cap A_0, A'_1 \cap A_1)$ separates $Q'_0$ and $Q'_1$, (ii) For $i=0,1$, the number of terminals on the $Q'_i$-side must be at least $c_i - |\tset_0^i|$ because otherwise this would contradict the fact that $(A_0,A_1)$ is a $(Q_0,Q_1,c_0,c_1,\ell)$-constrained cut.  
\end{proof}

\begin{lemma}
	\label{lem:runtime}
	Let $c = \max\{\ell, c_0, c_1\}$. The algorithm to reduce the problem of finding a $(Q_0,Q_1,c_0,c_1,\ell)$-constrained cut with $\min{c_0, c_1} > 0$ to the problem of finding a $(Q'_0,Q'_1,c'_0,c'_1,\ell)$-constrained cut with $\min{c_1, c_0} =0$ terminates in time $2^{O(c^2)}\cdot k^2\cdot {O}(m)$.
\end{lemma}

\begin{proof}
	Lemma~\ref{lem: one step} implies that the depth of the recursion tree is at the most $2c$ and that each recursive step reduces to solving $2^{O(c)}$ sub-instances. Hence, the total number of nodes in the recursion tree is $2^{O(c^2)}$. 
	
	The total runtime outside the recursive calls is dominated by a  minimum $s-t$-cut computation. However, we observe that we are only interested in minimum cuts that are of value at the most $c$. Hence, such a cut can be found in time $O(mc)$ using any standard augmentation path based algorithm.
	Also, recall that we are looking for cuts that have at least one terminal on each side and hence we need to make $k^2$ guesses. The total runtime for this procedure is $k^2\cdot {O}(mc)$ and we have the lemma.
\end{proof}

\subsection{Handling the base case} 
In this subsection, we prove Lemma~\ref{lem:algo for base}, i.e., we present
an algorithm that finds a $(Q_0,Q_1,c_0,0,\ell)$-constrained cut $(A'_0,A'_1)$. %
We first consider the case of $c_0=0$: since neither side of the cut is required to
contain a terminal, we can simply compute a minimum-cut between $Q_0$ and
$Q_1$. %
If one of these is empty (say $Q_1$), we take $A'_0 = V(G)$, $A'_1 =
\emptyset$. %
In any case, let $E_1$ be the edges of the cut. %
Now there are two possibilities: if $|E_1| \leq \ell$, then our cut is a solution
to the subproblem; %
if $|E_1| > \ell$, then there is no cut separating $Q_0$ from $Q_1$ with at
most $\ell$ edges, and therefore, there is no valid constrained cut.

We can now focus on the case where $c_0 >0$. 
We can further assume that $|\tset| \geq c_0$; otherwise, there is no feasible
solution. %
For simplification, we also assume that $Q_0$ is connected; if it is not, we
can add auxiliary edges to make it connected in the run of the algorithm, which we
can remove afterwards (these edges will never be cut since $Q_0 \subseteq A'_0$). %

\paragraph{Important Cuts.}
The main tool we will be using is the notion of important cuts, introduced by
Marx \cite{Marx06} (see \cite{CyganFKLMPPS15} and references within for other
results using this concept).

\begin{definition}[Important cut]
	Let $G$ be a graph and $X,Y \subseteq V(G)$ be disjoint subsets of vertices of
	$G$.
	
	A cut $(S_X, S_Y)$, $X \subseteq S_X$, $Y \subseteq S_Y$ is an important cut
	if it has (inclusionwise) maximal reachability (from $X$) among all cuts with
	at most as many edges. %
	In other words, there is no cut $(S'_X, S'_Y)$, $X \subseteq S'_X$, $Y
	\subseteq S'_Y$, such that $|E(S'_X, S'_Y)| \leq |E(S_X, S_Y)|$ and $S_X
	\subsetneq S'_X$. %
\end{definition}

\begin{proposition}[{\cite{CyganFKLMPPS15}}]
	\label{prop:impcut:exist}
	Let $G$ be an undirected graph and $X, Y \subseteq V(G)$ two disjoint sets of vertices. %
	
	Let $(S_X, S_Y)$ be an $(X,Y)$-cut. Then there is an important $(X,Y)$-cut $(S'_X, S'_Y)$ (possibly $S_X = S'_X$) such that $S_X \subseteq S'_X$ and $|E(S'_X, S'_Y)| \leq |E(S_X, S_Y)|$.
\end{proposition}

\begin{theorem}[{\cite{CyganFKLMPPS15}}]
	Let $G$ be an undirected graph, $X, Y \subseteq V(G)$ be two disjoint sets of
	vertices and $c\geq 0$ be an integer. %
	There are at most $4^c$ important $(X,Y)$-cuts of size at most $c$. %
	Furthermore, the set of all important $(X,Y)$-cuts of size at most $c$ can
	be enumerated in time $O(4^c \cdot c \cdot m)$.
\end{theorem}

\begin{proposition}
	\label{prop:violcut:impcutforall}
	Let $G$ be an undirected graph and $X, Y \subseteq V(G)$ two disjoint sets of
	vertices, and let $(S_X, S_Y)$ be an important $(X,Y)$-cut. %
	
	Then $(S_X, S_Y)$ is also an important $(X',Y)$-cut for all $X' \subseteq S_X$.
\end{proposition}

\begin{proof}
	Assume that the statement is false for contradiction. Then there is an
	important cut $(S'_X, S'_Y)$ for $(X', Y)$, with $|E(S'_X, S'_Y)| \leq |E(S_X,
	S_Y)|$ and $S_X \subsetneq S'_X$ by Proposition \ref{prop:impcut:exist}. %
	But then, $X \subseteq S_X \subseteq S'_X$, which means $(S'_X, S'_Y)$ is an
	$(X,Y)$-cut, and therefore, $(S_X, S_Y)$ is not an important cut for $(X,Y)$,
	which is a contradiction.
\end{proof}

\paragraph{Cut profile vectors.}
In order to make the exposition of the algorithm clearer, we introduce the concept of cut profile vectors.

\begin{definition}
	\label{def:k0:slot}
	Let $c, \ell \geq 0$. A \emph{cut profile vector} is a vector of $\lambda \leq c$
	pairs of numbers $\{(\kappa_i, \ell_i)\}_{i \in [\lambda]}$, with $\kappa_i \in
	[c-1]$, $\ell_i \in [\ell]$, satisfying
	\[
	c \leq \sum_{i = 1}^\lambda \kappa_i \leq 2c, \quad \quad \sum_{i=1}^\lambda \ell_i \leq \ell
	\]
	Each of the pairs $(\kappa_i, \ell_i)$ is called a \emph{slot} of this profile. %
	We say a cut $(A,B)$ is compatible with a slot $(\kappa_i, \ell_i)$ if $|A
	\cap \tset| = \kappa_i$ and $|E(A,B)| = \ell_i$
\end{definition}

\begin{observation}
	There are at most $c^c \cdot \ell^c$ different cut profile vectors. 
\end{observation}

Given a cut $(A,B)$, a cut profile vector represents the bounds for terminals
covered and cut edges for each of the components of $G[A]$: there are $\lambda$
connected components, and component $C_i$ contains $\kappa_i$ terminals and
has $\ell_i$ cut edges. %
Our algorithm will enumerate all the possible cut profile vectors and, for each of
them, try to find a solution that fits the constraints given by the input. %
If there is a solution to the problem, there must be a corresponding profile 
vector, and therefore the algorithm finds a solution. %

\paragraph{Algorithm} Our algorithm works by guessing the number of connected
components of $G[A]$ as well as the number of terminals that are contained in each
component, and then proceeding to find cuts that fit these guesses. %
This is made easier by the following two facts:
\begin{inparaenum}
\item there is a solution such that $A$ is a disjoint union of important $(Q_0, Q_1)$-cuts or $(t, Q_1)$-cuts, $t \in \tset$;
\item we can find a set of $O(k^2)$ terminals such that there is a solution
where each connected component of $G[A]$ contains one of these terminals.
\end{inparaenum}

The strategy of the algorithm is as follows: it starts by guessing the component
$C_0$ that contains $Q_0$, out of all important $(Q_0, Q_1)$-cuts. %
If $(C_0, \bar C_0)$ is feasible, it returns. %
Otherwise, it guesses the cut profile vector of the solution. %
Then it tries to greedily fill all of the slots using important cuts
containing disjoint sets of terminals. %
The goal of this stage is not yet to obtain a solution, but to accumulate
terminals for the second stage. %
This process of trying to fill each slot is repeated $c$ times so that we may
have $c$ candidates for each slot. %
All of the terminals contained in each of the candidates found this way form
our base set of terminals, denoted $S$. %
The solution is finally obtained by enumerating tuples of up to $c$
components out of important $(t,Q_1)$-cuts, for $t \in S$. %

We refer to Figure \ref{alg:CC:k0} for a formal description of the algorithm.

We will now show that if there is a solution to the problem, our algorithm
always finds a solution. This implies that, when we output ``No Valid
Solution'', there is no solution. %
From now on, we assume that there is a solution to the problem. %
Let $(A,B)$ be a solution that minimizes the number of connected components of
$G[A]$. %

Let $\cset_0$ be the set of all important cuts $(C, \bar C)$ for $(Q_0, Q_1)$, and let
$\cset$ be the set of all important cuts $(C, \bar C)$ for $(t, Q_1)$, for any $t \in
\tset$.

\begin{lemma}
	\label{lem:violcut:impcutcomps}
	There is a solution $(A', B')$ such that every connected component $C$ of
	$G[A']$ corresponds to an important cut $(C, \bar C)$ in $\cset_0$ or $\cset$. %
	Furthermore, the number of connected components of $G[A']$ is not greater than
	that of $G[A]$.
\end{lemma}

\begin{proof}
	We will show an iterative process that turns a solution $(A,B)$ into a
	solution $(A', B)$ where every component corresponds to an important cut as
	above. %
	
	Let $C$ be a component of $G[A]$ that does not correspond to an important cut in
	$\cset_0$ or $\cset$. %
	Notice that $C$ cannot contain a proper non-empty subset of $Q_0$ since $Q_0
	\subseteq A$ and we assume that $Q_0$ is connected. %
	If $C$ does not contain any terminals or $Q_0$, we move $C$ to $B$ (resulting
	in the cut $(A\setminus C, B \cup C)$). %
	Since $C$ is a connected component of $G[A]$, all of the neighbors of $C$ are in
	$B$, and therefore moving $C$ to $B$ does not add any cut edges. %
	
	In the remaining case, $C$ contains a terminal $t \in \tset$ or $Q_0$
	but is not an important cut. %
	By Proposition \ref{prop:impcut:exist}, there is an important cut $(C',
	\bar C')$ with at most as many cut edges as $(C, \bar C)$ and $C \subsetneq
	C'$. %
	We can replace $C$ by a component corresponding to an important cut by taking
	the cut $(A \cup C', B \setminus C')$. %
	This is still a valid solution since all terminals contained in $A$ are
	contained in $A \cup C'$, and $Q_0 \subseteq A$, $Q_1 \subseteq B \setminus C'$.
	Additionally, the number of edges crossing the cut does not increase: since
	$|E(C', \bar C')| \leq |E(C, \bar C)|$, the number of edges added to the
	cutset is at most the number of edges removed.
	
	We can apply the operations above until the constraints in the lemma are
	satisfied. %
	Notice that when applying the operations above, the number of components of
	$G[A]$ never increases and the number of vertices in $A$ connected to
	terminals in $G[A]$ never decreases. %
	Furthermore, each operation changes at least one of the two quantities above,
	so this process must finish after a finite number of operations.
\end{proof}

\begin{lemma}
If a feasible cut $(A,B)$ exists, then our algorithm returns a feasible solution.
\end{lemma}

\begin{proof}
Due to Lemma \ref{lem:violcut:impcutcomps}, we can assume that every connected
component of $G[A]$ corresponds to an important cut. %
Now, let $C^*_0, C^*_1, \ldots, C^*_\lambda$ be the connected components of
$G[A]$, with $C^*_0 \in \cset_0$ being the component that contains $Q_0$. %
Let $\set{(\kappa_i, \ell_i)}_\lambda$ be the cut profile vector corresponding
to the cuts $(C^*_i, \bar C^*_i)$ for $i \in \set{1, \ldots, \lambda}$
(excluding $C^*_0$), meaning that $\kappa_i$, $\ell_i$ are the number of
terminals in $C^*_i$ and the number of edges in the cutset, $E(C^*_i, \bar
C^*_i)$, respectively. %
Notice that, if $C^*_0$ or $C^*_0 \cup C^*_i$ (for some $i \in [\lambda]$)
contain at least $c$ terminals, then we can remove all the other components of
$A$. %
In this case, the algorithm finds $C_0\in\cset_0$ or $C_0 \in \cset_0$, $C_1
\in \cset$ by enumeration and returns a valid solution. %
Otherwise, all the components contain at most $c-1$ terminals each (and thus
$A$ induces a slot vector as in Definition \ref{def:k0:slot}).

Consider the iteration of the algorithm in which the cut profile
vector defined above is considered and $C_0 = C^*_0$. %
The next part of the algorithm (Lines
\ref{alg:k0:forround}--\ref{alg:k0:forroundend}) greedily fills the slots with
compatible important cuts from $\cset$, while making sure that each set
contains a disjoint set of terminals from the others. %
Though it seems that our goal at this stage is to obtain a feasible solution,
what we intend is to obtain a set of terminals, denoted $S$, such that the set
of important cuts for terminals in $S$ contains a feasible solution. %
For instance, if $S$ contains at least one terminal from each $C^*_i$,
$i\in[\lambda]$, our goal is achieved.

The above considerations motivate the following definition. %
We say a slot $i$ is \emph{hit} by $S$ if $S \cap C^*_i \neq \emptyset$. %
Notice that slot $i$ is hit by $S$ if $C_{ji} = C^*_i$ for some $j$, since
the terminals in $C^*_i$ is added to $S$. %
Slot $i$ is also hit by $S$ if, for some $j$, we cannot find a set $C_{ji}$,
since that implies that $C_{ji} = C^*_i$ is not a valid choice, and thus $S
\cap C^*_i \neq \emptyset$. %
Furthermore, if slot $i$ is not hit by $S$, then $C_{ji}$ is found in all
$(c+1)$ rounds. %

Let $\cset_S \subseteq \cset$ be the subset of important cuts containing
terminals in $S$ %
(by Proposition \ref{prop:violcut:impcutforall} these are the important $(t,
Q_1)$-cuts for $t \in S$). %
It is now suficient to show that there is a sequence of $\lambda$ cuts
$\set{(C_i,\bar C_i)}_{i \in \lambda}$ from $\cset_S$, such that all $C_i$
contain disjoint sets of terminals (also disjoint with the terminals in
$C_0$), and such that $(C_i,\bar C_i)$ is compatible with $(\kappa_i,
\ell_i)$. %
Taking $C=C_0 \cup \bigcup_{i=1}^\lambda C_i$, we obtain a feasible solution
$(C, \bar C)$, which may be different from $(A,B)$, but has the same number of
connected components as $G[A]$, and the same numbers of terminals
contained in each component and cut edges separating each component from the
other side of the cut. %
Since the algorithm enumerates all such sequences of $\lambda$ sets, it will
find either $(C, \bar C)$ or a different feasible solution. %

We now define the sets $C_i$: if a slot $i$ is hit by $S$ %
we can set $C_i = C^*_i$, since there is $t \in C^*_i \cap S$, and therefore,
$(C^*_i, \bar C^*_i) \in \cset_S$. %
This cut is trivially compatible with $(\kappa_i, \ell_i)$, and is disjoint to
all other sets defined similarly. 
Let $I_H$ be the set of all $i \in [\lambda]$ such that slot $i$ is hit by $S$,
and let $C^* = \bigcup \set{ C_i: i \in I_H}$. %
All that is left to prove is that, for every slot $i$ that is not hit by $S$,
there is an important cut $(C_i, \bar C_i) \in \cset_S$, which
contains terminals not in any previous $C_{i'}$, $i' \leq i$, or in $C^*$. %
Notice that we have covered at most $c$ terminals so far %
(if we covered more, then the components so far are sufficient and therefore the
number of components of $G[A]$ is not minimal). %
Since there are $c+1$ important cuts $(C_{ji}, \bar C_{ji})$, $j\in[c+1]$, all
compatible with slot $i$ and containing disjoint sets of terminals %
(since the terminals of $C_{ji}$ are added to $S$ after being picked), %
there must be one set $C_{ji}$ that does not contain any of the at most $c$
terminals in $\bigcup_{i' < i} C_{i'}$, or in $C^*$, and we can set
$C_i=C_{ji}$. %
Therefore, a sequence $\set{(C_i,\bar C_i)}_{i \in \lambda}$ exists, and the
algorithm outputs a feasible solution. %
\end{proof}

Finally, we analyze the running time.

\begin{lemma}
The described algorithms terminates in time $2^{O(c^2)} \cdot k + 2^{O(c)}\cdot k \cdot m$.
\end{lemma}
\begin{proof}
Computing all the relevant important cuts takes time $2^{O(c)}\cdot k\cdot m$. %
There are at most $c^{O(c)}$ cut profile vectors, and for each of these %
the algorithm fills the slots at most $c$ times, which takes time $c^2 \cdot k
\cdot 2^{O(c)}$; %
then, once it has computed $S$, it enumerates at most $c$ components out of
$2^{O(c)} \cdot c^2$ possible important cuts, which takes time
$\parenbig{2^{O(c)} \cdot c^2}^c = 2^{O(c^2)}$. %
Once the right combination of components is found, it takes $O(n)$ time to
obtain the corresponding feasible cut.
The total running time is
\[
2^{O(c)}\cdot k\cdot m + c^{O(c)} \parenbig{c^2 \cdot k \cdot 2^{O(c)} + 2^{O(c^2)}} = 2^{O(c^2)} \cdot k + 2^{O(c)}\cdot k \cdot m
\]
\end{proof}
\begin{figure}[t]
\label{alg:CC:k0}
\small
\begin{algorithmic}[1]
\Function{$\CC$}{$G$, $\tset$, $Q_0$, $Q_1$,$c$, $\ell$}
	\If {$c = 0$}
	\label{alg:k0:ifk0}
		\State Compute a min-cut $(A'_0, A'_1)$ such that $Q_0\subseteq A'_0, Q_1\subseteq A'_1$
		\If {$|E(A'_0, A'_1)| \leq \ell$} 
			\State \Return $(A'_0, A'_1)$
		\Else 
			\State\Return \nosolution
		\EndIf
	\EndIf
	\label{alg:k0:ifk0end}
	\Statex

	\State Compute the set $\cset_0$ of all  important $(Q_0, Q_1)$-cuts with at most $\ell$ cut edges
	\State Compute the set $\cset$ of all important $(t, Q_1)$-cuts for $t \in \tset$ with at most $\ell$ cut edges
	\label{alg:k0:impcut}
	\Statex

	\State Find an important cut $(C_0, \bar{C_0}) \in \cset_0$ such that $|C_0 \cap \tset| \geq c$
	\label{alg:k0:spec1}
	\If {$(C_0, \bar{C_0})$ exists}
		\State\Return $(C_0, \bar{C_0})$
	\EndIf
	\State  \parbox[t]{0.9\textwidth}{
			Find important cuts $(C_0, \bar{C_0}) \in \cset_0$, $(C_1, \bar{C_1}) \in \cset$, 
			such that $(C_0 \cap \tset) \cap (C_1 \cap \tset) = \emptyset$, \par
			\hskip\algorithmicindent $|C_0\cap\tset| + |C_1 \cap \tset| \geq c$, 
			and $\card{E(C_0 \cup C_1, \bar{C_0} \cap \bar{C_1})} \leq \ell$
			}
	\If {$(C_0, \bar{C_0})$, $(C_1, \bar{C_1})$ exist}
		\State\Return $(C_0 \cup C_1, \bar{C_0} \cap \bar{C_1})$
	\EndIf
	\label{alg:k0:spec2}

	\Statex
	\For {all cut profile vectors $\set{(\kappa_i, \ell_i)}_\lambda$, and all $(C_0, \bar{C_0}) \in \cset_0$ }
	\label{alg:k0:forslotvec}
		\State $S \gets C_0 \cap \tset$
		\For {$j \in \set{1, \ldots, c+1}$} \InlineComment{Round $j$}
		\label{alg:k0:forround}
			\For {$i \in \set{1, \ldots, \lambda}$}
			\label{alg:k0:forslot}
				\State Find $(C_{ji}, \bar C_{ji}) \in \cset$ compatible with slot $(\kappa_i, \ell_i)$, 
						such that $C_{ji} \cap S = \emptyset$
				\If {$C_{ji}$ exists} 
					\State $S \gets S \cup (C_{ji} \cap \tset)$
				\EndIf
			\EndFor
		\EndFor
		\label{alg:k0:forroundend}

		\Statex
		\State Let $\cset_S = \set{(C, \bar C) \mid C \cap S \neq \emptyset}$
		\State \parbox[t]{0.8\textwidth}{
			Find (by enumeration) $\set{(C_i,\bar C_i)}_{i \in \lambda}$, with $(C_i,\bar C_i) \in \cset_S$ compatible with slot $i$, \par 
			\hskip\algorithmicindent and all sets $C_i\cap \tset$ are disjoint (including $C_0 \cap \tset$)
			}
		\label{alg:k0:enumset}
		\If {$\set{(C_i,\bar C_i)}_{i \in \lambda}$ exists}
			\State Let $C=C_0 \cup \bigcup_{i=1}^\lambda C_i$
			\State \Return $(C, \bar C)$
		\EndIf
	\EndFor

	\Statex
	\State\Return \nosolution
	\InlineComment{No solution found for any cut profile}
\EndFunction
\end{algorithmic}
\caption{Algorithm to find a constrained cut in the base case}
\end{figure}
\end{appendix}

\end{document}